%% file: main-Matrix.tex
\pdfoutput=1

\RequirePackage{luatex85}

\documentclass[11pt,a4paper]{article}

\input{our_preamble}

\usepackage[font=small,labelfont=bf]{caption}

\newcommand{\ith}{
	\ifmmode
	^\mathrm{th}%
	\else%
	\textsuperscript{th}\xspace%
	\fi%
}

\newcommand{\F}{\mathbb{F}}
\newcommand{\Z}{\mathbb{Z}}
\newcommand{\N}{\mathbb{N}}
\newcommand{\I}{\mathbb{I}}
\newcommand{\adj}{\operatorname{adj}}

\newcommand{\slowExponent}{{1.529}}

\newcommand{\fastExponent}{{1.407}}
\newcommand{\slowExponentQuery}{{0.529}}
\newcommand{\fastExponentLB}{{1.406}}
\newcommand{\slowExponentLB}{{1.528}}
\newcommand{\slowExponentQueryLB}{{0.528}}
\newcommand{\matrixExponent}{{2.3728639}}

\newcommand{\OLDfastExponent}{{1.447}}
\newcommand{\OLDIdealExponent}{{1+1/3}}

\newcommand{\hintedOMv}{{v-hinted Mv}} 
\newcommand{\vHintedMv}{{v-hinted Mv}\xspace}
\newcommand{\doubleHintedOMv}{{Mv-hinted Mv}} 
\newcommand{\MvHintedMv}{{Mv-hinted Mv}\xspace}
\newcommand{\doubleHintedOuMv}{{uMv-hinted uMv}} 
\newcommand{\uMvHintedMv}{{uMv-hinted uMv}\xspace}
\newcommand{\streach}{{$st$-reachability}}

\newcommand\elementColor[1]{\textcolor{blue}{#1}}
\newcommand\columnColor[1]{\textcolor{red}{#1}}
\definecolor{darkgreen}{rgb}{0.05, 0.5, 0.06}
\newcommand\mixedColor[1]{\textcolor{darkgreen}{#1}}

\usepackage{xcolor}
\usepackage{titlesec}
\usepackage{mdframed}
\usepackage{lipsum}

\colorlet{mycolor}{gray!50}

\newmdenv[linecolor=mycolor,skipabove=\topsep,skipbelow=\topsep,
leftmargin=-5pt,rightmargin=-5pt,
innerleftmargin=5pt,innerrightmargin=5pt]{mybox}

\usepackage{comment} 

\ifdefined\DEBUG

\newcommand\jan[1]{\textcolor{blue}{#1}}
\def\janfoot#1{\marginpar{$\leftarrow$\fbox{J}}\footnote{$\Rightarrow$~{\sf\textcolor{blue}{#1 --Jan}}}}
\def\danupon#1{\marginpar{$\leftarrow$\fbox{D}}\footnote{$\Rightarrow$~{\sf\textcolor{orange}{#1 --Danupon}}}}
\def\thatchaphol#1{\marginpar{$\leftarrow$\fbox{T}}\footnote{$\Rightarrow$~{\sf\textcolor{purple}{#1 --Thatchaphol}}}}

\newcommand{\todo}[1]{ \vspace{5 mm} \par
	\noindent \framebox{
		\begin{minipage}
			[c]{0.94 \textwidth}
			\tt #1 \flushright ... todo
		\end{minipage}
	}
	\vspace{5 mm} \par
}

\else

\newcommand\jan[1]{}
\def\danupon#1{}
\def\thatchaphol#1{}
\def\janfoot#1{}

\newcommand{\todo}[1]{}

\fi

\title{Dynamic Matrix Inverse: Improved Algorithms and Matching Conditional Lower Bounds}

\author[1]{Jan van den Brand}
\author[1]{Danupon Nanongkai}
\author[2]{Thatchaphol Saranurak\thanks{Works partially done while at KTH Royal Institute of Technology, Sweden.}}
\affil[1]{KTH Royal Institute of Technology, Sweden}
\affil[2]{Toyota Technological Institute at Chicago, USA}

\date{}

\begin{document}

\begin{titlepage}
\maketitle
\pagenumbering{roman}
\ifdefined\DEBUG
\begin{center}
{\centering\huge\textcolor{red}{DEBUG VERSION}}
\end{center}
\fi

\begin{abstract}

\input{abstract}
\end{abstract}

\newpage 

\setcounter{tocdepth}{2}
\tableofcontents

\end{titlepage}


\newpage
\pagenumbering{arabic}

\input{introduction}

\input{overview}

\input{preliminaries.tex}

\input{inverse.tex}

\input{lowerbounds.tex}

\input{lookahead.tex}

\input{openproblems.tex}

\input{acknowledgment}

\newpage 

%
\appendix
%
\input{appendix.tex}
%

\printbibliography[heading=bibintoc] 

\end{document}

%% file: our_preamble.tex
\usepackage{iftex}

\ifPDFTeX
\usepackage[utf8]{inputenc}
\usepackage[T1]{fontenc} 
\fi

\usepackage{authblk} 

\usepackage{amsmath}
\usepackage{amssymb}
\usepackage{amsthm}
\usepackage{thmtools}
\usepackage{mathtools}

\ifPDFTeX
\usepackage{libertine}
\usepackage[libertine]{newtxmath} 
\usepackage{MnSymbol}
\usepackage[scaled=0.83]{beramono} 
\else
\usepackage{unicode-math}
\fi

\usepackage[left=1in,top=1in,right=1in,bottom=1in]{geometry} 

\usepackage{microtype} 

\usepackage{enumitem}
\usepackage{xcolor}
\usepackage{xspace}

\usepackage[
	style=alphabetic,
	backref=true,
	doi=false,
	url=false,
	maxcitenames=3,
	mincitenames=3,
	maxbibnames=10,
	minbibnames=10,
	backend=bibtex8,
	sortlocale=en_US
]{biblatex}

\usepackage{tikz} 
\usepackage{pdflscape}

\usepackage{algorithm,algorithmic}

\renewcommand{\algorithmicensure}{\textbf{Output:}}

\usepackage{array}
\usepackage{multirow}

\usepackage[colorlinks]{hyperref} 

\ifdefined\DEBUG

\usepackage[inline]{showlabels,rotating} 

\else
\fi

\usepackage{cleveref}

\input{commands}

\colorlet{DarkRed}{red!50!black}
\colorlet{DarkGreen}{green!50!black}
\colorlet{DarkBlue}{blue!50!black}

\hypersetup{
	linkcolor = DarkRed,
	citecolor = DarkGreen,
	urlcolor = DarkBlue,
	bookmarksnumbered = true,
	linktocpage = true
}

\declaretheorem[numberwithin=section]{theorem}
\declaretheorem[numberlike=theorem]{lemma}

\declaretheorem[numberlike=theorem]{corollary}
\declaretheorem[numberlike=theorem]{definition}

\declaretheorem[numberlike=theorem,name={Conjecture},refname={Conjecture,Conjectures},Refname={Conjecture,Conjectures}]{conjecture}
\declaretheorem[numberlike=theorem,refname={Fact,Facts},Refname={Fact,Facts},name={Fact}]{fact}

\declaretheorem[numberlike=theorem]{problem}

%% file: commands.tex
\makeatletter
\g@addto@macro\bfseries{\boldmath}
\makeatother

\makeatletter
\g@addto@macro\mdseries{\unboldmath}
\g@addto@macro\normalfont{\unboldmath}
\g@addto@macro\rmfamily{\unboldmath}
\g@addto@macro\upshape{\unboldmath}
\makeatother

\DeclareCiteCommand{\citem}
    {}
    {\mkbibbrackets{\bibhyperref{\usebibmacro{postnote}}}}
    {\multicitedelim}
    {}

\renewcommand*{\multicitedelim}{\addcomma\space}

\AtEveryCitekey{\clearfield{note}}

\newcommand{\myhref}[1]{%
  \iffieldundef{doi}
    {\iffieldundef{url}
       {#1}
       {\href{\strfield{url}}{#1}}}
    {\href{http://dx.doi.org/\strfield{doi}}{#1}}%
}

\DeclareFieldFormat{title}{\myhref{\mkbibemph{#1}}}
\DeclareFieldFormat
  [article,inbook,incollection,inproceedings,patent,thesis,unpublished]
  {title}{\myhref{\mkbibquote{#1\isdot}}}

\makeatletter
\AtBeginDocument{%
    \newlength{\temp@x}%
    \newlength{\temp@y}%
    \newlength{\temp@w}%
    \newlength{\temp@h}%
    \def\my@coords#1#2#3#4{%
      \setlength{\temp@x}{#1}%
      \setlength{\temp@y}{#2}%
      \setlength{\temp@w}{#3}%
      \setlength{\temp@h}{#4}%
      \adjustlengths{}%
      \my@pdfliteral{\strip@pt\temp@x\space\strip@pt\temp@y\space\strip@pt\temp@w\space\strip@pt\temp@h\space re}}%
    \ifpdf
      \typeout{In PDF mode}%
      \def\my@pdfliteral#1{\pdfliteral page{#1}}
      \def\adjustlengths{}%
    \fi
    \ifxetex
      \def\my@pdfliteral #1{}
      \def\adjustlengths{\setlength{\temp@h}{-\temp@h}\addtolength{\temp@y}{1in}\addtolength{\temp@x}{-1in}}%
    \fi%
    \def\Hy@colorlink#1{%
      \begingroup
        \ifHy@ocgcolorlinks
          \def\Hy@ocgcolor{#1}%
          \my@pdfliteral{q}%
          \my@pdfliteral{7 Tr}
        \else
          \HyColor@UseColor#1%
        \fi
    }%
    \def\Hy@endcolorlink{%
      \ifHy@ocgcolorlinks%
        \my@pdfliteral{/OC/OCPrint BDC}%
        \my@coords{0pt}{0pt}{\pdfpagewidth}{\pdfpageheight}%
        \my@pdfliteral{F}
        \my@pdfliteral{EMC/OC/OCView BDC}%
        \begingroup%
          \expandafter\HyColor@UseColor\Hy@ocgcolor%
          \my@coords{0pt}{0pt}{\pdfpagewidth}{\pdfpageheight}%
          \my@pdfliteral{F}
        \endgroup%
        \my@pdfliteral{EMC}%
        \my@pdfliteral{0 Tr}
        \my@pdfliteral{Q}%
      \fi
      \endgroup
    }%
}
\makeatother

%% file: abstract.tex
The dynamic matrix inverse problem is to maintain the inverse of a matrix undergoing element and column updates. It is the main subroutine behind the best algorithms for many dynamic problems whose complexity is not yet well-understood, such as maintaining the largest eigenvalue, rank and determinant of a matrix and maintaining reachability, distances, maximum matching size, and $k$-paths/cycles in a graph. Understanding the complexity of dynamic matrix inverse is a key to understand these problems.

In this paper, we present 
(i) improved algorithms for dynamic matrix inverse and their extensions to some incremental/look-ahead variants, and 
(ii) variants of the Online Matrix-Vector conjecture [Henzinger~et~al. STOC'15] that, if true, imply that these algorithms are tight. 
Our algorithms automatically lead to faster dynamic algorithms for the aforementioned problems, some of which are also tight under our conjectures, e.g. reachability and maximum matching size (closing the gaps for these two problems was in fact asked by Abboud and V. Williams [FOCS'14]).  Prior best bounds for most of these problems date back to more than a decade ago [Sankowski FOCS'04, COCOON'05, SODA'07; Kavitha FSTTCS'08; Mucha and Sankowski Algorithmica'10;  Bosek~et~al. FOCS'14].

Our improvements stem mostly from the ability to use fast matrix multiplication ``one more time'', to maintain a certain transformation matrix which could be maintained only combinatorially previously (i.e. without fast matrix multiplication). Oddly, unlike other dynamic problems where this approach, once successful, could be repeated several times (``bootstrapping''), our conjectures imply that this is not the case for dynamic matrix inverse and some related problems. However,  when a small additional ``look-ahead'' information is provided we can perform such repetition to drive the bounds down further. 

%% file: introduction.tex
\section{Introduction}\label{sec:intro}

In the {\em dynamic matrix inverse} problem, we want to maintain the inverse of an $n\times n$ matrix $A$ over any field,
when $A$ undergoes some updates. There were many variants of this problem considered \cite{Sankowski04,Sankowski07,LeeS15,CohenLS18}: Updates can be {\em element updates}, where we change the value of one element in $A$, or {\em column updates}, where we change the values of all elements in one column.\footnote{There are other kinds of updates which we do not consider in this paper, such as rank-1 updates in \cite{LeeS15,CohenLS18}.} 
%
The inverse of $A$ might be maintained explicitly or might be answered through an {\em element query} or a {\em row/column query}; the former returns the value of a specified element of the inverse, and the latter answers the values of all elements in a specified row/column of the inverse.
The goal is to design algorithms with small {\em update time} and {\em query time}, denoting the time needed to handle each update and each query respectively. Time complexity is measured by the number of {\em field operations}.\footnote{\label{foot:time} Later when we consider other kinds of dynamic problems, such as dynamic graphs, the time refer to the standard notion of time in the RAM model.}
Variants where elements are polynomials and where some updates are known ahead of time (the {\em look-ahead} setting) were also considered (e.g. \cite{SankowskiM10,Kavitha14,KhannaMW98,Yannakakis90}).
%

%
Dynamic matrix inverse algorithms played a central role in designing algorithms for many dynamic problems such as maintaining matrix and graph properties. 
Its study can be traced back to the 1950 algorithm of Sherman and Morrison \cite{ShermanM50} which can be used to maintain the inverse explicitly in $O(n^2)$ time. The previous best bounds are due to Sankowski's FOCS'04 paper \cite{Sankowski04} and its follow-ups  \cite{Sankowski07,SankowskiM10,Sankowski05}. 
Their time guarantees depend on how fast we can multiply matrices.
For example, with the state-of-the-art matrix multiplication algorithms \cite{GallU18,Gall14a}, Sankowski's algorithm \cite{Sankowski04} can handle an element update and answer an element query for matrix inverse in $O(n^\OLDfastExponent)$ time. 
Consequently, the same update time\footnote{See \Cref{foot:time}.} can be guaranteed for, e.g., maintaining largest eigenvalues, ranks and determinants of matrices undergoing entry updates
and 
maintaining maximum matching sizes, reachability between two nodes ({\em $st$-reachability}), existences of a directed cycle, numbers of spanning trees, and numbers of paths (in directed acyclic graphs; DAG) in graphs undergoing edge insertions and deletions.\footnote{We note that while the update and query time for the matrix inverse problem is defined to be the number of arithmetic operations, most of time the guarantees translate into the same running time in the RAM model. Exceptions are the numbers of spanning trees in a graph and numbers of paths in a DAG, where the output might be a very big number. In this case the running time is different from the number of arithmetic operations.}
(Unless specified otherwise, all mentioned update times are {\em worst-case} (as opposed to being amortized\footnote{\label{footnote:amortized}Amortized time is not the focus of this paper, and we are not aware of any better amortized bounds for problems we consider in this paper}).)  See \Cref{tbl:dynamic inverse,tbl:look ahead} for lists of known results for dynamic matrix inverse and \Cref{tbl:applications,tbl:applications2} for lists of applications.

{\em Is the $O(n^\OLDfastExponent)$ bound the best possible for above problems?} This kind of question exhibits the current {\em gap} between existing algorithmic and lower bound techniques and our limited understanding of the power of {\em algebraic techniques} in designing dynamic algorithms. 
First of all, despite many successes in the last decade in proving tight bounds for a host of dynamic problems (e.g. \cite{HenzingerKNS15,AbboudW14,Patrascu10}), 
conditional lower bounds for most of these problems got stuck at $\Omega(n)$ in general. Even for a very special case where the preprocessing time is limited to $o(n^\omega)$ (which is too limited as discussed in \Cref{sec:related}), the best known conditional lower bound of $\Omega(n^{\omega-1})=\Omega(n^{1.3728639})$ \cite{AbboudW14} is still not tight (\cite{AbboudW14} mentioned that ``closing this gap
is a very interesting open question''). 
Note that while the upper bounds might be improved in the future with improved rectangular matrix multiplication algorithms, there will still be big gaps even in the {\em best-possible} scenario: even if there is a {\em linear-time} rectangular matrix multiplication algorithm, the upper bounds will still be only $O(n^\OLDIdealExponent)$, while the lower bound will be $\Omega(n)$.

Secondly, it was shown that 
{\em algebraic techniques} -- techniques based on fast matrix multiplication algorithms initiated by Strassen \cite{Strassen69Gaussian} -- 
are {\em inherent} in any upper bound improvements for some of these problems: Assuming the {\em Combinatorial Boolean Matrix Multiplication (BMM) conjecture}, without algebraic techniques we cannot maintain, e.g., maximum matching size and  $st$-reachability faster than $O(n^2)$ per edge insertion/deletion~\cite{AbboudW14}\footnote{More precisely, assuming BMM, no ``combinatorial'' algorithm can maintain maximum matching size and $st$-reachability in $O(n^{2-\epsilon})$ time, for any constant $\epsilon>0$. Note that ``combinatorial'' a vague term usually used to refer as an algorithm that does not use subcubic-time matrix multiplication algorithms as initiated by Strassen \cite{Strassen69Gaussian}. We note that this statement only holds for algorithms with $O(n^{3-\epsilon})$ preprocessing time, which are the case for Sankowski's and our algorithms.}.
{\em Can algebraic techniques lead to faster algorithms that may ideally have update time linear in $n$? If not, how can we argue lower bounds that are superlinear in $n$ and, more importantly, match upper bounds from algebraic algorithms?} 


In this paper, we show that it is possible to improve some of the existing dynamic matrix inverse algorithms further and at the same time present conjectures that, if true, imply that they cannot be improved anymore.

\input{intro_results_upper}

\input{intro_results_lower}

\input{related}


%% file: intro_results_upper.tex
\subsection{Our Algorithmic Results (Details in \Cref{sec:dynamicInverse,sec:applications,sec:lookAhead})}

\paragraph{Algorithms in the Standard Setting (Details in \Cref{sec:dynamicInverse,sec:applications}).} We present two faster algorithms as summarized in \Cref{tbl:dynamic inverse}. With known fast matrix multiplication algorithms \cite{GallU18,Gall14a}, our first algorithm requires $O(n^{\fastExponent})$ time to handle each entry update and entry query, and the second requires $O(n^{\slowExponent})$ time to handle each column update and row query.

\input{table_main_results}


The first algorithm improves over Sankowski's decade-old $O(n^{\OLDfastExponent})$ bound, and automatically implies improved algorithms for over 10 problems, such as maximum matching size, $st$-reachability and DAG path counting under edge updates (see upper bounds in blue in \Cref{tbl:applications,tbl:applications2}). 

The second bound leads to first non-trivial upper bounds for maintaining the largest eigenvalue, rank, determinant under column updates, which consequently lead to new algorithms for dynamic graph problems, such as maintaining maximum matching size under insertions and deletions of nodes on one side of a bipartite graph  (see upper bounds in red in \Cref{tbl:applications,tbl:applications2}).

Note that the update time can be traded with the query time, but the trade-offs are slightly complicated. See \Cref{thm:elementUpdate,thm:columnUpdate} for these trade-offs.

\paragraph{Incremental/Look-Ahead Algorithms and Online Bipartite Matching (Details in \Cref{sec:lookAhead}).} We can speed up our algorithms further in a fairly general {\em look-ahead} setting, where we know ahead of time which columns will be updated. Previous algorithms (\cite{SankowskiM10,Kavitha14,KhannaMW98,Yannakakis90}) can only handle some special cases of this (e.g. when the update columns and the new values are known ahead of time). Our update time depends on how far ahead in the future we see. When we see $n$ columns to be updated in the future, the update time is $O(n^{\omega-1})$. (See \Cref{thm:columnLookAhead,thm:elementLookAhead} for detailed bounds.) As a special case, we can handle the {\em column-incremental} setting, where we start from an empty (or identity) $n\times n$ matrix and insert the $i^{th}$ column at the $i^{th}$ update.

As an application, we can maintain the maximum matching size of a bipartite graph under the arrival of nodes on one side in  $O(n^{\omega})$ total time. This problem is known as the {\em online matching} problem. Our bound improves the $O(m\sqrt{n})$ bound in \cite{BosekLSZ14}\footnote{Also see \cite{BernsteinHR18}.} (where $m$ is the number of edges), when the graph is dense, additionally our result matches the bound in the static setting by \cite{Lovasz79} (see also \cite{MuchaS04}). 

See \Cref{sec:related} for further discussions on previous results. 



\paragraph{Techniques (more in \Cref{sec:overview}).}  Our improvements are mostly due to our ability to exploit fast matrix multiplication more often than previous algorithms. In particular, Sankowski  \cite{Sankowski04} shows that to maintain matrix inverse, it suffices to maintain the inverse of another matrix that we call {\em transformation matrix}, which has a nicer structure than the input matrix. To keep this nice structure, we have to ``reset'' the transformation matrix to the identity from time to time; the reset process is where fast matrix multiplication algorithms are used. In more details, Sankowski writes the maintained matrix $A$ as $A=A'T$, 
where $A'$ is an older version of the matrix and $T$ is a ``transformation matrix''. He then shows some methods to quickly maintain $T^{-1}$ by exploiting its nice structures. The query about $A^{-1}$ is then answered by computing necessary parts in $A^{-1}=T^{-1}(A')^{-1}$. From time to time, he ``resets'' the transformation matrix by assigning $A'\leftarrow A'T$, $A'^{-1} \leftarrow T^{-1}A'^{-1}$ and $T \leftarrow \I$ (the identity matrix). 

A natural idea to speedup the above algorithm is to repeat the same idea again and again (``bootstrapping''), i.e. to write $T=T_1T_2$ (thus $A^{-1}=(A'T_1T_2)^{-1}$) and try to maintain $T_2^{-1}$ quickly.  Indeed, finding a clever way to repeat the same ideas several times is a key approach to significantly speed up many dynamic algorithms. (For a recent example, consider the spanning tree problem where \cite{NanongkaiSW17} sped up the $n^{1/2-\epsilon}$ update time of \cite{NanongkaiS17,Wulff-Nilsen17} to $n^{o(1)}$ by appropriately repeating the approach of \cite{NanongkaiS17,Wulff-Nilsen17} for about $\sqrt{\log(n)}$ times. See, e.g., \cite{HenzingerKN-jacm18,HenzingerKN-stoc14,HenzingerKN-icalp13} for other examples.)
The challenge is how to do it right. 
Arguably, this approach has already been taken in \cite{Sankowski04} where $T^{-1}$ is maintained in the form $T^{-1}=(T_1T_2T_3T_4\ldots)^{-1}$.\footnote{\cite{Sankowski04} presents several dynamic matrix inverse algorithms. Algorithm ``Dynamic Matrix Inverse: Simple Updates
II'' is the one with the $T^{-1}=(T_1T_2T_3\ldots)^{-1}$ structure.}
However, we observe that this and other methods previously used to maintain $T^{-1}$ do {\em not} exploit fast matrix multiplication, and in fact the same result can be obtained without writing $T^{-1}$ in this long form. (See the discussion of \Cref{eq:longChain} in \Cref{sec:overview}.) 
An important question here is: {\em Can we use fast matrix multiplication to maintain $T^{-1}$?} 
An attempt to answer the above questions runs immediately to a barrier: 
while it is simple to maintain $T$ explicitly after every update, maintaining $T_2$ {\em explicitly} already takes too much time!

In this paper, we show that one can get around the above barrier and repeat the approach one more time. 
To do this, we develop a dynamic matrix inverse algorithm that can handle updates that are {\em implicit} in a certain way. This algorithm allows us to maintain $T_2$ {\em implicitly}, thus avoid introducing a large running time needed to maintain $T_2$ explicitly. 
It also generalizes and simplifies one of Sankowski's algorithm, giving additional benefits in speeding up algorithms in the look-ahead setting and algorithms for some graph problems.

\paragraph{Further bootstrapping?}
Typically once the approach can be repeated to speed up a dynamic algorithm, it can be repeated several times (e.g. \cite{NanongkaiSW17,HenzingerKN-jacm18,HenzingerKN-icalp13}). Given this, it might be tempting to get further speed-ups by writing $A=A'T_1T_2T_3T_4\ldots$ instead of just $A=A'T_1T_2$. 
Interestingly, it does not seem to help even to write $A=A'T_1T_2T_3$.
{\em Why are we stuck at $A=A'T_1T_2$?}
On a technical level, since we have to develop a new, implicit, algorithm to maintain $T_2$ quickly, it is very unclear how to develop yet another algorithm to maintain $T_3$ quickly.
On a conceptual level, this difficulty is captured by our conjectures below, which do not only explain the difficulties for the dynamic matrix inverse problem, but also for many other problems. 
Thus, these conjectures capture a phenomenon that we have not observed in other problems before.
Interestingly, with a small ``look-ahead'' information, namely the columns to be updated, this approach can be taken further to reduce the update time to match existing conditional lower bounds.

%% file: table_main_results.tex
\begin{figure}
	
\centering
	{\tiny

\begin{tabular}{>{\centering}p{0.15\textwidth}|>{\centering}p{0.1\textwidth}|>{\centering}p{0.1\textwidth}||>{\centering}p{0.1\textwidth}|>{\centering}p{0.12\textwidth}|>{\centering}p{0.15\textwidth}}
\hline 
Variants  & Known upper bound  & Known lower bound  & New upper bound  & New lower bound  & Corresponding conjectures \tabularnewline
\hline 
\hline

Element update  & $O(n^{\OLDfastExponent})$ $[O(n^{1+1/3})]$  & \multirow{3}{0.1\textwidth}{\centering $u+q=\Omega(n)$ via OMv \cite{HenzingerKNS15}} & \elementColor{$O(n^{\fastExponent})$ $[O(n^{1+1/4})]$}  & \multirow{3}{0.12\textwidth}{\centering \elementColor{$u+q=\Omega(n^{\fastExponentLB})$ $[\Omega(n^{1+1/4})]$} \Cref{cor:elementUpdateLB}}  & \elementColor{\multirow{3}{0.15\textwidth}{\doubleHintedOuMv{} (\Cref{con:2hintedOuMv})}}
\tabularnewline
Element query  & $O(n^{\OLDfastExponent})$ $[O(n^{1+1/3})]$ &  & \elementColor{$O(n^{\fastExponent})$ $[O(n^{1+1/4})]$}  &  & \tabularnewline
 & \cite{Sankowski04} &  & \Cref{thm:elementUpdate}  &  & \tabularnewline
\hline

same as above & $O(n^{\slowExponent})$ $[O(n^{1.5})]$ &  &  & \mixedColor{$u=\Omega(n^{\slowExponentLB})$ $[\Omega(n^{1.5})]$} or & \mixedColor{\multirow{3}{0.15\textwidth}{\doubleHintedOMv{} (\Cref{con:doubleHintedOMv})}}\tabularnewline
 & $O(n^{\slowExponentQuery})$ $[O(n^{0.5})]$ &  &  & \mixedColor{$q=\Omega(n^{\slowExponentQueryLB})$ $[\Omega(n^{0.5})]$} & \tabularnewline
 & \cite{Sankowski04} &  &  & \Cref{cor:elementColumnLB} & \tabularnewline
\hline

Element update  & $O(n^{\slowExponent})$ $[O(n^{1.5})]$ & \multirow{3}{0.1\textwidth}{\centering $u+q=\Omega(n)$  via OMv \cite{HenzingerKNS15}} & -  & \multirow{3}{0.12\textwidth}{\centering \mixedColor{$u+q=\Omega(n^{\slowExponentLB})$ $[\Omega(n^{1.5})]$} \Cref{cor:columnUpdateLB}} & \mixedColor{\multirow{3}{0.15\textwidth}{\doubleHintedOMv{} (\Cref{con:doubleHintedOMv})}}\tabularnewline
Row query  & $O(n^{\slowExponent})$ $[O(n^{1.5})]$ &  &  &  & \tabularnewline
 & \cite{Sankowski04} &  &  &  & \tabularnewline
\hline

Column update  & $O(n^{2})$  & \multirow{3}{0.10\textwidth}{\centering $u+q=\Omega(n^{\omega-1})$  $[\Omega(n)]$  [trivial]} & \columnColor{$O(n^{\slowExponent})$} & \multirow{3}{0.12\textwidth}{\centering \columnColor{$u+q=\Omega(n^{\slowExponentLB})$ $[\Omega(n^{1.5})]$} \Cref{cor:columnUpdateLB}}  & \columnColor{\multirow{3}{0.15\textwidth}{\hintedOMv{}\\ (\Cref{con:hintedOMv})}} \tabularnewline
Row query  & $O(n)$  &  & \columnColor{$O(n^{\slowExponent})$ $[O(n^{1.5})]$} &  & \tabularnewline
 & [trivial] &  & \Cref{thm:columnUpdate} &  & \tabularnewline
\hline

Column+Row update  & $O(n^{2})$  & - & - & \multirow{3}{0.12\textwidth}{\centering $u+q=\Omega(n^2)$ \Cref{thm:columnRowUpdateLB}}  & \multirow{3}{0.15\textwidth}{OMv conjecture \cite{HenzingerKNS15}} \tabularnewline
Element query  & $O(1)$  &  &  &  & \tabularnewline
 & [trivial] &  &  &  & \tabularnewline
\end{tabular}

	}
	{
		\caption{Our new upper and conditional lower bounds (in colors) compared to the previous ones. All previous upper bounds were due to Sankowski \cite{Sankowski04}. Bounds in brackets $[ \: \cdot \: ]$ are for the ideal scenario, where there exists a linear-time  rectangular matrix multiplication algorithm. Colors in the bounds are used to connect to applications in \Cref{tbl:applications,tbl:applications2}.
		The exponents in the upper and corresponding lower bounds are different because of the rounding. They are actually the same numbers.
		}\label{tbl:dynamic inverse}
	}
\end{figure}

%% file: intro_results_lower.tex


\subsection{Our Conditional Lower Bounds (Details in \Cref{sec:lowerBounds})}\label{sec:intro:lowerBounds}


We present conjectures that imply tight conditional lower bounds for many problems. We first present our conjectures and their implications here, and discuss existing conjectures and lower bounds in \Cref{sec:related}. 
We emphasize that our goal is not to invent new lower bound techniques, but rather to find a simple, believable, explanation that our bounds are tight. Since the conjectures below are the only explanation we know of, they might be useful to understand other dynamic algebraic algorithms in the future.

\paragraph{Our Conjectures.} 
We present variants of the OMv conjecture. 
To explain our conjectures, recall a special case of the OMv conjecture called {\em Matrix-Vector Multiplication (Mv)} \cite{ChakrabortyKL-stoc18,LarsenW-soda17,CliffordGL15}. The problem has two phases. In Phase~1, we are given a boolean matrix $M$, and in Phase~2 we are given a boolean vector $v$. Then, we have to output the product $Mv$. Another closely related problem is the {\em Vector-Matrix-Vector (uMv)} product problem where in the second phase we are given two vectors $u$ and $v$ and have to output the product $u^\top Mv$.
A naive algorithm for these problems is to spend $O(|M|)$ time in Phase~2 to compute $Mv$ and $uMv$, where $|M|$ is the number of entries in $M$. The OMv conjecture implies that we cannot beat this native algorithm even when we can spend polynomial time in the first phase; i.e. there is {\em no} algorithm that spends polynomial time in Phase~1 and $O(|M|^{1-\epsilon})$ time in Phase~2 for any constant $\epsilon>0$. 
(The OMv conjecture in fact implies that this holds even if the second phase is repeated multiple times, but this is not needed in our discussion here.) 

In this paper, we consider ``hinted'' variants of Mv and uMv, where matrices are given as ``hints'' of $u$, $M$ and $v$, and later their submatrices 
are selected to define $u$, $M$ and $v$. In particular, consider the following problems. 

\begin{enumerate} 
	\item {\em The \vHintedMv Problem} (formally defined in \Cref{def:hintedOMv}): We are given a boolean matrix $M$ in Phase~1, a boolean matrix $V$ in Phase~2 (as a ``hint'' of $v$), and an index $i$ in Phase~3. Then, we have to output the matrix-vector product $Mv$, where $v$ is the $i\ith$ column of $V$. 
	\item {\em The \MvHintedMv Problem} (formally defined in \Cref{def:doubleHintedOMv}): We are given boolean matrices $N$ and $V$ in Phase~1 (as ``hints'' of $M$ and $v$), a set of indices $K$ in Phase~2, and an index $i$ in Phase~3. Then, we have to output the matrix-vector product $Mv$, where $v$ is as above, and $M$ is the submatrix of $N$ obtained by deleting the $k\ith$ rows of $N$ for all $k\notin K$.
	\item {\em The \uMvHintedMv Problem} (formally defined in \Cref{def:2hintedOuMv}): We are given boolean matrices $U$, $N$, and $V$ in Phase~1 (as ``hints'' of $u$, $M$ and $v$), a set of indices  $K$ in Phase~2,  a set of indices  $L$ in Phase~3, and indices $i$ and $j$  in Phase~4. Then, we have to output the vector-matrix-vector product $u^\top Mv$, where  $u$ is the $j\ith$ column of $U$, $v$ is as above, and $M$ is the submatrix of $N$ obtained by deleting the $k\ith$ rows and $\ell\ith$ columns of $N$ for all $k\notin K$ and $\ell\notin L$.
\end{enumerate}

A naive algorithm for the first problem (\vHintedMv) is to either compute $Mv$ naively in $O(|M|)$ time in Phase~3 or precompute $Mv$ for {\em all possible $v$} in Phase~2 by running state-of-the-art matrix multiplication algorithms \cite{GallU18,Gall14a} to multiply $MV$. 
Our {\em \vHintedMv conjecture} (formally stated in \Cref{con:hintedOMv}) says that we cannot beat the running time of this naive algorithm in Phases~2 and 3 simultaneously even when we have polynomial time in Phase~1; i.e. there is {\em no} algorithm that spends polynomial time in Phase~1, time polynomially smaller than computing $MV$ with state-of-the-art matrix multiplication algorithms in Phase~2, and $O(|M|^{1-\epsilon})$ time in Phase~3,
for any constant $\epsilon>0$. Similarly, the {\em \MvHintedMv and \uMvHintedMv conjectures} state that we cannot beat naive algorithms for the \MvHintedMv and \uMvHintedMv problems, which either precompute everything using fast matrix multiplication algorithms in one of the phases or compute $Mv$ and $uMv$ naively in the last phase; see \Cref{con:2hintedOuMv,con:doubleHintedOMv} for their formal statements.

\paragraph{Lower Bounds Based on Our Conjectures.}
The conjectures above allow us to argue tight conditional lower bounds for the dynamic matrix inverse problem as well as some of its applications. In particular, the \uMvHintedMv conjecture leads to tight conditional lower bounds for element queries and updates, as well as, e.g., maintaining rank and determinant of a matrix undergoing element updates, and maintaining maximum matching size, $st$-reachability, cycle detection and DAG path counting in graphs undergoing edge insertions/deletions; see lower bounds in {\em blue} in \Cref{tbl:dynamic inverse,tbl:applications,tbl:applications2} for the full list. 
Our \vHintedMv conjecture leads to tight conditional lower bounds for column update and row query, as well as, e.g., maintaining adjoint and matrix product under the same type of updates and queries, maintaining bipartite maximum matching under node updates on one side of the graph;  see lower bounds in {\em red} in \Cref{tbl:dynamic inverse,tbl:applications,tbl:applications2}. 
Finally, our \MvHintedMv conjecture gives conditional lower bounds that match two algorithms of Sankowski \cite{Sankowski04} that we could not improve, as well as some of their applications; see lower bounds in {\em green} in \Cref{tbl:dynamic inverse,tbl:applications,tbl:applications2}. 
All our tight conditional lower bounds remain tight even if there are improved matrix multiplication algorithms in the future; see, e.g., bounds inside brackets $[ \: \cdot \: ]$ in  \Cref{tbl:dynamic inverse,tbl:applications,tbl:applications2}, which are valid assuming that a linear-time  matrix multiplication algorithm exists.

\paragraph{Remarks.} 
Our conjectures only imply lower bounds for {\em worst-case} update and query time, which are the focus of this paper. To make the same bounds hold against amortized time, one can consider the {\em online} versions of these conjectures, where all phases except the first may be repeated; see \Cref{sec:amortizedLowerBounds}. However, we feel that the online versions are too complicated to be the right conjectures, and that it is a very interesting open problem to either come up with clean conjectures that capture amortized update time, or break our upper bounds using amortization.

The reductions from our conjectures are pretty much the same as the existing ones. As discussed in \Cref{sec:related}, we consider this an advantage of our conjectures. 
Finally, whether to believe our conjectures or not might depend on the readers' opinions.
A more important point is that these easy-to-state conjectures capture the hardness of a number of important dynamic problems. On the way to make further progress on any of these problems is to break naive algorithms from our conjectures first.

%% file: related.tex
\subsection{Other Related Work}\label{sec:related}

\paragraph{Look-Ahead Algorithms.} The look-ahead setting refers to when we know the future changes ahead of time and was considered in, e.g., \cite{SankowskiM10,Kavitha14,KhannaMW98,Yannakakis90}.  
Look-ahead dynamic algorithms did not receive as much attention as the non-look-ahead setting due to limited applications, but it turns out that our algorithms require a rather weak look-ahead assumption, and become useful for the online bipartite matching problem \cite{BosekLSZ14}. 
Our results compared with the previous ones are summarized in \Cref{tbl:look ahead}.

Previously, Sankowski and Mucha \cite{SankowskiM10} showed that algorithms that can look-ahead, i.e. they know which columns will be updated with what values and which rows will be queried, can maintain the inverse and determinant faster than Sankowski's none-look-ahead algorithms \cite{Sankowski04}. Kavitha \cite{Kavitha14} extended this result to maintaining rank under element updates, needing to only know which entries will be updated in the future but not their new values. For the case where the algorithm in \cite{SankowskiM10} know $n$ updates in the future, it is {\em tight}  as a better bound would imply a faster matrix multiplication algorithm.

\begin{figure}
	\centering
	{\tiny
		\input{lookaheadtable.tex}
	}
	
	{\footnotesize
		\caption{Comparison of different look-ahead algorithms. $\omega$ is the exponent of matrix multiplication.
			All inverse algorithms need to know row indices of the queries ahead of time. (Note that the algorithm with element query does not need to know the exact element to be queried.)
			The results by \cite{KhannaMW98} are not included, as they are subsumed by dynamic algorithms without look-ahead. The algorithms maintaining inverse and determinant can also maintain adjoint, linear system and other algebraic problems via the reductions from section \ref{sec:applications}. The rank reduction is adaptive and does not work with the type of look-ahead used in \cite{SankowskiM10}. \label{tbl:look ahead}}
	}
	
\end{figure}
\footnotetext{\label{foot:lookAheadType} A rough estimate for the index is enough: When looking $t$ rounds into the future, we only need an index set of size $O(t)$ as prediction for all the future $t$ update/query positions together.}

In this paper, we present faster look-ahead algorithms when $n^k$ updates are known ahead of time, for any $k<1$. More importantly, our algorithms only need to know ahead of time the columns that will be updated, but not the values of their entries. 
Our algorithms are compared with the previous ones by \cite{SankowskiM10,Kavitha14} in \Cref{tbl:look ahead}. In \Cref{tbl:look ahead} we do not state the time explicitly for all possible $k$, but only state which algorithms are faster and give explicit bounds only for $k=0.25$. For detailed bounds, see \Cref{thm:columnLookAhead,thm:elementLookAhead}.

One special case of our algorithms is maintaining a rank when 
we start from an empty $n\times n$ matrix and insert the $i^{th}$ column at the $i^{th}$ update. We can compute the rank after each insertion in $O(n^{\omega-1})$ time, or  $O(n^{\omega})$ in total over $n$ insertions. Since the maximum matching size in a bipartite graph $G$ corresponds to the rank of a certain matrix $M$, and adding one node to, say, the right side of $G$ corresponds to adding a column to $M$, our results imply that we can maintain the maximum matching size under the arrival of nodes on one side in  $O(n^{\omega})$ total time. This problem is known as the {\em online matching} problem. Our bound improves the $O(m\sqrt{n})$ bound in \cite{BosekLSZ14}\footnote{Also see \cite{BernsteinHR18}.} (where $m$ is the number of edges), when the graph is dense, additionally our result matches the bound in the static setting by \cite{Lovasz79} (see also \cite{MuchaS04}). We note that previous algorithms did not lead to this result because \cite{SankowskiM10} needs to know the new values ahead of time while \cite{Kavitha14} only handles element updates.

\paragraph{Existing Lower Bounds and Conjectures.}
%
Two known conjectures that capture the hardness for most dynamic problems are the Online Matrix-Vector Multiplication conjecture (OMv) \cite{HenzingerKNS15} and the Strong Exponential Time Hypothesis (SETH) \cite{AbboudW14}. Since the OMv conjecture implies (roughly) an $\Omega(n)$ lower bound for dynamic matching and $st$-reachability, it automatically implies a lower bound for dynamic inverse, rank and determinant.
However, it is not clear how to use these conjectures to capture the hardness of dynamic problems whose upper bounds can still possibly be improved with improved fast matrix multiplication algorithms.

More suitable conjectures should have dependencies on $\omega$, the matrix multiplication exponent. Based on this type of conjectures, the best lower bound is $\Omega(n^{\omega-1})=\Omega(n^{1.3728639})$ assuming an $\Omega(n^{\omega})$ lower bound for checking if an $n$-node graph contains a triangle (the Strong Triangle conjecture) \cite{AbboudW14}. This lower bound does not match our upper bounds of $O(n^{\fastExponent})$ (and note that \cite{AbboudW14} mentioned that closing the gap between their lower bound and Sankowski's upper bound is a very interesting open question). 
More importantly, this lower bound applies only for a special case where algorithms' preprocessing time is limited to $o(n^\omega)$ (in contrast to, e.g., SETH- and OMv-based lower bounds that hold against algorithms with polynomial preprocessing time).  
Because of this, it unfortunately does not rule out (i) the possibilities to improve the update time of Sankowski's or our algorithms which have $O(n^\omega)$ preprocessing time and more generally (ii) the existence of algorithms with lower update time but high preprocessing time, which are typically desired. In fact, with such limitation on the preprocessing time, it is easy to argue that maintaining some properties requires $n^\omega$ update time, which is {\em higher} than Sankowski's and our upper bounds. For example, assuming that any static algorithm for computing the matrix determinant requires $\Theta(n^\omega)$ time, we can argue that an algorithm that uses $o(n^\omega)$ time to preprocess a matrix $A$ requires $n^\omega$ time to maintain the determinant of $A$ even when an update does nothing to $A$.
(See \Cref{sub:trivialLowerBounds} for more lower bounds of this type.) Because of this, we aim to argue lower bounds for algorithm with polynomial preprocessing time.

In light of the above discussions, the next appropriate choice is to make new conjectures.
While there are many possible conjectures to make, we select the above because they are {\em simple} and {\em similar to the existing ones}.
We believe that this provides some advantages: 
(i) It is easier to develop an intuition whether the new conjectures are true or not, based on the knowledge of the existing conjectures; for example, we discuss what a previous attempt to refute the OMv conjecture \cite{LarsenW-soda17} means to our new conjectures in \Cref{sub:implicationOfPreviousResults}.
(ii) There is a higher chance that existing reductions (from known conjectures) can be applied to the new ones. Indeed, this is why our conjectures imply tight lower bounds for many problems beyond dynamic matrix inverse. 

We note that while the term ``hinted'' was not used before in the literature, the concept itself is not that unfamiliar. For example, Patrascu's {\em multiphase problem} \cite{Patrascu10} is a hinted version of the vector-vector product problem: given a boolean matrix $U$ in Phase~1, vector $v$  
in Phase~2, and an index $i$ in Phase~3, compute the inner product $u^\top v$ where $u$ is the $i\ith$ column of matrix $U$.
%

%% file: lookaheadtable.tex

\begin{tabular}{>{\centering}p{0.07\textwidth}|>{\centering}p{0.15\textwidth}|>{\centering}p{0.16\textwidth}|>{\centering}p{0.05\textwidth}|>{\centering}p{0.1\textwidth}|>{\centering}p{0.1\textwidth}|>{\centering}p{0.07\textwidth}}
\hline 
 & Problem  & Type of  & Type of  & \multicolumn{3}{c}{Update time for $n^{k}$ look-ahead}\tabularnewline
\cline{5-7} 
 &  & look-ahead  & updates  & $k=1$ & $k<1$  & $k=0.25$ \tabularnewline
\hline 
\hline 
\cite{SankowskiM10}  & inverse (row query)

and determinant but \emph{not} rank  & column index and values & column  & $O(n^{\omega-1})$

(amortized) & 3: slower than 2 for every $k < 1$ & $O(n^{1.75})$ \tabularnewline
\hline 
\cite{Kavitha14}  & rank  & column \emph{and} row index  & element  & $O(n^{\omega-1})$

(amortized) & 4: slowest for every $k < 1$ & $O(n^{2.122})$ \tabularnewline
\hline 
\Cref{thm:columnUpdate}  & inverse (row query),

determinant and rank  & column index\footnotemark
& column  & $O(n^{\omega-1})$ & 2: slower than 1 for every $k < 1$ & $O(n^{1.453})$ \tabularnewline
\hline 
\Cref{thm:elementUpdate}  & inverse (element query),

determinant and rank  & column index\textsuperscript{\ref{foot:lookAheadType}} & element & $O(n^{\omega-1})$ & 1: fastest for every $k < 1$ & $O(n^{1.392})$ \tabularnewline
\hline 
\end{tabular}

%% file: overview.tex
\section{Overview of Our Algorithms }
\label{sec:overview}

Let $A^{(0)}$ be the initial matrix $A$ before any updates and denote with $A^{(t)}$ the matrix $A$ after it received $t$ updates. For now we will focus only on the case where $A^{(t)}$ is always invertible, as a reduction from \cite{Sankowski07} allows us to extend the algorithm to the setting where $A^{(t)}$ may become singular (\Cref{thm:singularUpdates}). We will also focus only on the case of element updates and queries. The same ideas can be extended to other cases.

\paragraph{Reduction  to Transformation Inverse Maintenance (Details in \Cref{sub:TimpliesInverse}).}
The core idea of the previously-best dynamic inverse algorithms of Sankowski \cite{Sankowski04} is to express the change of matrix $A^{(0)}$ to $A^{(t)}$ via some {\em transformation matrix} $T^{(0,t)}$, i.e. we write
\begin{align}
A^{(t)} = A^{(0)}T^{(0,t)}
\end{align}
This approach is beneficial since $T^{(0,t)}$ has more structure than $A^{(t)}$: 
(i) obviously $T^{(0,0)}=\I$ (the identity matrix), and (ii) changing the $(i,j)$-entry of $A^{(t)}$ changes only the $j^{th}$ column of $T^{(0,t)}$ (and such change can be computed in $O(n)$ time).\footnote{To see this, write $T^{(0,t)}=(A^{(0)})^{-1}A^{(t)}$. The $(i,j)$-entry of $A^{(t)}$ will multiply only with the $i^{th}$ column of $(A^{(0)})^{-1}$ and affects the $j^{th}$ column of the product. \label{footnote:updateToT}
}
Moreover, to get the $(i,j)$-entry of $(A^{(t)})^{-1}$ notice that 
\begin{align}\label{eq:InverseViaTransform}
(A^{(t)})^{-1} = (T^{(0,t)})^{-1}(A^{(0)})^{-1},
\end{align}
and thus we just have to multiply the $i^{th}$ row of $(T^{(0,t)})^{-1}$ with the $j^{th}$ column of $(A^{(0)})^{-1}$.
This motivates the following problem.
\begin{problem}[Maintaining inverse of the transformation, $(T^{(0,t)})^{-1}$]\label{prob:InverseTransform}
	We start with $T^{(0,0)}=\I$. Each update is a change in one column. A query is made on a row of $(T^{(0,t)})^{-1}$. It can be assumed that $T^{(0,t)}$ is invertible for any $t$. 
\end{problem}
As we will see below, there are many fast algorithms for \Cref{prob:InverseTransform} when $t$ is small. 
A standard ``resetting technique'' can then convert these algorithms into fast algorithms for maintaining matrix inverse: An element update to $A^{(t)}$ becomes a column update to $T^{(0,t)}$. When $t$ gets large (thus algorithms for \Cref{prob:InverseTransform} become slow), we use fast matrix multiplication to compute $(A^{(t)})^{-1}$ explicitly so that $T^{(0,t)}$ is ``reset'' to $\I$. 

To summarize, it suffices to solve \Cref{prob:InverseTransform}. 
Our improvements follow directly from improved algorithms for this problem, which will be our focus in the rest of this section.

\begin{figure}
	\center
	\begin{tabular}{l|ccc}
		& \cite{Sankowski04} & \cite{Sankowski04} & Our \\
		\hline
		Column update & $O(nt)$ & $O(t^2)$ & $O(n^x t + n^{\omega(1,x,\log_n t)-x})$ $[O(\sqrt{n} t)]$ \\
		Row query & $O(t)$ & $O(t^2)$ & $O(n^x t + n^{\omega(1,x,\log_n t)-x})$ $[O(\sqrt{n} t)]$
	\end{tabular}
	\caption{\label{fig:Tmaintenance}Comparison of different transformation maintenance algorithms. The task is to support column updates to $T^{(0,t)}$ and row queries to $(T^{(0,t)})^{-1}$, where $t$ is the number of updates so far. Values in $[ \cdot ]$ correspond to the case of optimal matrix multiplication ($\omega = 2$) and are given for easier comparison of the complexities.}
\end{figure}

\paragraph{Previous  maintenance of $(T^{(0,t)})^{-1}$.} Sankowski \cite{Sankowski04} presented two algorithms for maintaining $(T^{(0,t)})^{-1}$; see \Cref{fig:Tmaintenance}.
The first algorithm maintains $(T^{(0,t)})^{-1}$ {\em explicitly} by observing {\em if a matrix $M$ differs from $\I$ in at most $k$ columns, so is its inverse}.
%
This immediately implies that querying a row of 
$(T^{(0,t)})^{-1}$ needs $O(t)$ time, since $(T^{(0,t)})^{-1}$ differs from $\I$ in at most $t$ columns. Moreover, expressing an update by a linear transformation, i.e.
\begin{align}\label{eq:SankFirst}
T^{(0,t)}=T^{(0,t-1)}T^{(t-1,t)}
\end{align}
for some matrix $T^{(t-1,t)}$, and using the fact that $T^{(t-1,t)}$ and $(T^{(t-1,t)})^{-1}$ differs from $\I$ in only one column, computing $(T^{(0,t)})^{-1}$ boils down to multiplying a vector with $(T^{(0,t)})^{-1}$, thus taking $O(nt)$ update time. 
\begin{mybox}
\underline{More details} (may be skipped at first reading): 
We can write 
%
\begin{align}
T^{(0,t)}&=T^{(0,t-1)}\underbrace{\left(\I+(T^{(0,t-1)})^{-1}\underbrace{[T^{(0,t)}-T^{(0,t-1)}]}_C\right)}_{T^{(t-1,t)}}, \mbox{ thus}\nonumber\\
(T^{(0,t)})^{-1}&= \underbrace{\left(\I+(T^{(0,t-1)})^{-1}\underbrace{[T^{(0,t)}-T^{(0,t-1)}]}_C\right)^{-1}}_{(T^{(t-1,t)})^{-1}}(T^{(0,t-1)})^{-1}.\label{eq:SankFirstDetail}
\end{align}
Since $C=[T^{(0,t)}-T^{(0,t-1)}]$ contains only one non-zero column, $T^{(t-1,t)}$ differs from $\I$ only in one column. Consequently, $(T^{(t-1,t)})^{-1}$ can be computed in $O(n)$ time and differs from $\I$ only in one column. Thus $T^{(t-1,t)}(T^{(0,t-1)})^{-1}$ takes $O(nt)$ time to compute. 
\end{mybox}
The $O(nt)$ update time of this algorithm is optimal in the sense that one column update to $T^{(0,t)}$ may cause  $\Omega(nt)$ entries in  $(T^{(0,t)})^{-1}$ to change; thus maintaining $(T^{(0,t)})^{-1}$ explicitly requires  $\Omega(nt)$ update time in the worst case.\footnote{For an example of one changed column inducing $\Omega(nt)$ changes in the inverse, we refer to \Cref{app:manyChangesExample}}
Sankowski's second algorithm breaks this bound (with the cost of higher query time) by expressing updates by a long chain of linear transformation
%
\begin{align}
T^{(0,t)} &= T^{(0,1)} T^{(1,2)} \ldots T^{(t-2,t-1)} T^{(t-1,t)}, \mbox{thus}\nonumber\\
(T^{(0,t)})^{-1} &= (T^{(t-1,t)})^{-1}(T^{(t-2,t-1)})^{-1}\ldots (T^{(1,2)})^{-1}(T^{(0,1)})^{-1}.
\label{eq:longChain}
\end{align}
Here each matrix $(T^{(i,i+1)})^{-1}$ is very sparse.
The sparsity leads to the update time improvement over the first algorithm, since computing some entries of $(T^{(0,t)})^{-1}$ does not require all entries of each $(T^{(i,i+1)})^{-1}$ to be known (intuitively because most entries will be multiplied with zero).\footnote{Note that the update time of Sankowski's second algorithm in  \Cref{fig:Tmaintenance} is presented in a slightly simplified form. In particular, this bound only holds when $t=\Omega(\sqrt{n})$ (which is the only case we need in this paper), or otherwise it should be the bound of the number of arithmetic operations only.}
The sparsity, however, also makes it hard to exploit fast matrix multiplication. Exploiting fast matrix multiplication one more time is the new aspect of our algorithm.

\paragraph{Our new  maintenance of $(T^{(0,t)})^{-1}$ via fast matrix multiplication.} 
As discussed above and as can be checked in \cite{Sankowski04}, both algorithms of Sankowski do {\em not} use fast matrix multiplication to maintain $(T^{(0,t)})^{-1}$; it is used only to compute $(A^{(t)})^{-1}$ as in \Cref{eq:InverseViaTransform} (to ``reset'').\footnote{
	In particular, both of Sankowski's algorithms maintain $(T^{(0,t)})^{-1}$ by performing matrix-vector products. 
}
Our improvements are mostly because {\em we can use fast matrix multiplication to maintain $(T^{(0,t)})^{-1}$}. To start with, we write $T^{(0, t)}$ as
\begin{align}
T^{(0, t)} = T^{(0, t')} T^{(t',t)}, \quad \mbox{thus}\quad (T^{(0, t)} )^{-1} = (T^{(t',t)})^{-1}(T^{(0, t')})^{-1}.\label{eq:TransformViaTransform}
\end{align}
This looks very much like what Sankowski's first algorithm (see \Cref{eq:InverseViaTransform}) {\em except that we may have $t'\ll t$}; this allows us to benefit from fast matrix multiplication 
when we compute $T^{(0, t')} T^{(t',t)}$, since both matrices are quite dense.
%
Like the discussion above \Cref{prob:InverseTransform}, a column update of $T^{(0, t)}$ leads to a column update of $T^{(t',t)}$, and a row query to $(T^{(0, t)})^{-1}$ needs a row query to $(T^{(t',t)})^{-1}$. 
%
This seems to suggest that maintaining $(T^{(0, t)} )^{-1}$ can be once again reduced to solving the same problem for $T^{(t',t)}$, and by repeating Sankowski's idea
we should be able to exploit fast matrix multiplication and maintain $(T^{(0, t)} )^{-1}$ faster. 

There is, however, on obstacle to execute this idea: even just maintaining $T^{(t', t)}$ {\em explicitly}  (without its inverse) already takes {\em too much time}. 
To see this, suppose that at time $t$ we add a vector $v$ to the $j^{th}$ column of $T^{(0, t-1)}$; 
with  $e_j$ being a unit vector which has value $1$ at the $j^{th}$ coordinate and $0$ otherwise, 
this can be expressed as
\begin{align}
T^{(0,t)} = T^{(0,t-1)} + e_j^\top v, 
\mbox{ thus (by \eqref{eq:TransformViaTransform}) }
T^{(t',t)} 
= T^{(t',t-1)} + e_j^\top [(T^{(0,t')})^{-1}v].\label{eq:T2update}
\end{align}
%
%
%
This means that for every column update to $T^{(0,t)}$, we have to compute a matrix-vector product $(T^{(0,t')})^{-1}v$ just to obtain $T^{(t',t-1)}$.
So for every update we have to read the entire inverse $(T^{(0,t')})^{-1}$, which has $\Omega(nt')$ non-zero entries.
Given that we repeatedly reset the algorithm to exploit fast matrix multiplication by setting $t' \leftarrow t$, this yields a $\Omega(nt)$ lower bound on our approach, i.e. no improvement over Sankowski's first algorithm (column 1 of \Cref{fig:Tmaintenance}).

So to summarize, just maintaining $T^{(t',t)}$ is already too slow.

\paragraph{Implicit input, simplification and generalization of Sankowski's second algorithm (details  in \Cref{sub:Timplicit}).}
To get around the above obstacle, we consider when updates to $T^{(t',t)}$ are given {\em implicitly}:

\begin{problem}[Maintaining inverse of the transformation under implicit column updates]\label{prob:InverseTransformImplicit}
	We start with $T^{(t',t')}=\I$ at time $t'$. Each update is an index $j$, indicating that some change happens in the $j^{th}$ column. Whenever the algorithm wants to know a particular entry in $T^{(t',t)}$ (at time $t\geq t'$), it can make a query to an {\em oracle}. The algorithm also has to answer a query made on a row of $(T^{(t',t)})^{-1}$ at any time $t$. 
	The algorithm's performance is measured by its running time {\em and} the number of oracle queries.
	It can be assumed that $T^{(t',t)}$ is invertible for any $t$. 
\end{problem}
%
In \Cref{sub:Timplicit}, we develop an algorithm for the above problem. It has the same update and query time as Sankowski's second algorithm, i.e. $O((t-t')^2)$ and additionally makes $O(t-t')$ oracle queries to perform each operation. 
Moreover, our algorithm does not need to maintain a chain of matrices as in \Cref{eq:longChain}. Eliminating this chain allows a further use of fast matrix multiplication, which yields an additional runtime improvement for the setting of batch-updates and batch-queries, i.e. when more than one entry is changed/queried at a time. This leads to improvements in the look-ahead setting and for some graph problems such as online-matching.

The starting point of our algorithm for \Cref{prob:InverseTransformImplicit} is the fact that $T^{(t',t)}$ and $(T^{(t',t)})^{-1}$  differs in at most $t-t'$ columns from the identity. Thus, by appropriately permuting rows and columns, we can write them as
\begin{align}
T^{(t', t)} = \left(
\begin{array}{c|c}
C_1 & 0 \\
\hline
C_2 & \I
\end{array}
\right)
\hspace{50pt}
(T^{(t', t)})^{-1} = \left(
\begin{array}{c|c}
C_1^{-1} & 0 \\
\hline
-C_2 C_1^{-1} & \I
\end{array}
\right)
\label{eq:TinverseImplicit}
\end{align}
Here, $C_1$ and $C_2$ are $(t-t')\times(t-t')$-  and $(n-t+t')\times (t-t')$-matrices, respectively.
This observation immediately yields the following solution to our problem: (i) In order to maintain $(T^{(t',t)})^{-1}$ implicitly, we only need to know the $C_1$ block of $T^{(t',t)}$. Since a column update to $T^{(t',t)}$ may change $C_1$ in $O(t-t')$ entries (i.e. either a column of $C_1$ is modified or a new row and column is added to $C_1$), we only need $O(t-t')$ oracle queries to keep track of $C_1^{-1}$ after each update.\footnote{To maintain $C_1^{-1}$ we use an extended version of \cite[Theorem 1]{Sankowski04}.}
(ii) For answer a query about some row of $(T^{(t',t)})^{-1}$ we may need a row of the $C_2$ block and compute the vector-matrix product of such row of $C_2$ with $C_1^{-1}$. Getting such row of $C_2$ requires $O(t-t')$ oracle queries. 

In summary, we do not require to fully know the matrix $T^{(t',t)}$ in order to maintain its inverse. This algorithm for maintaining $(T^{(t',t)})^{-1}$ is formalized in \Cref{lem:Timplicit} (\Cref{sub:Timplicit}).

\paragraph{Back to maintaining $(T^{(0, t)})^{-1}$: Using implicit $(T^{(t',t)})^{-1}$ maintenance (Details in \Cref{sub:combineT}).}
We now sketch how we use the algorithm hat maintains $(T^{(t',t)})^{-1}$ with implicit updates (cf. \Cref{prob:InverseTransformImplicit}) to maintain $(T^{(0, t)})^{-1}$ (cf. \Cref{prob:InverseTransform}).
The main idea is that we will implicitly maintain $T^{(t',t)}$ by explicitly maintaining $(T^{(0,t')})^{-1}$ and matrix
\begin{align}\label{eq:defS}
S^{(t',t)}:=T^{(0,t)}-T^{(0,t')}.
\end{align}
Like \Cref{eq:T2update}, we can derive
\begin{align}
T^{(0,t)} = T^{(0,t')} + S^{(t',t)}, 
\mbox{ thus (by \eqref{eq:TransformViaTransform}) }
T^{(t',t)} = \I + (T^{(0,t')})^{-1}S^{(t',t)}.\label{eq:OracleImplementation}
\end{align}

Thus, we can implement an oracle that provide an entry of $T^{(t',t)}$ by multiplying a row of $(T^{(0,t')})^{-1}$ with a column of $S^{(t',t)}$.
This can be done pretty fast by exploiting the fact that these matrices are rather sparse.

\paragraph{Summary.} In a nutshell, our algorithm maintains 
\begin{align*}
A^{(t)} = A^{(0)}T^{(0,t')}T^{(t',t)}, \quad \mbox{thus}\quad (A^{(t)})^{-1} = (T^{(t',t)})^{-1}(T^{(0,t')})^{-1}(A^{(0)})^{-1}.
\end{align*}
We keep the explicit values of $A^{(0)}$ and $T^{(0,t')}$ any at time. Additionally, we maintain explicitly a matrix $S^{(t',t)}$ satisfying \Cref{eq:OracleImplementation} (i.e. it collects all updates to $T^{(0,t)}$ since time $t'$).
As a subroutine we run our algorithm for \Cref{prob:InverseTransformImplicit} to maintain $(T^{(t',t)})^{-1}$ with implicit updates; call this algorithm $\mathcal{L}^{(t',t)}$, and see its detailed description in \Cref{sub:Timplicit}.

When, say, the entry $(i,j)$ of $A^{(t)}$ is updated, we 
(i) update matrix $S^{(t',t)}$, and (ii) implicitly update $T^{(t',t)}$ by sending index $j$ to $\mathcal{L}^{(t',t)}$. 
The first task is done by computing each update to $T^{(0,t)}$, which is not hard: since $T^{(0,t)}=(A^{(0)})^{-1}A^{(t)}$, we have to change the $j^{th}$ column of $T^{(0,t)}$ to the product of the $i^{th}$ column of $(A^{(0)})^{-1}$ and the changed entry $(i,j)$ of $A^{(t)}$ (see footnote \ref{footnote:updateToT}).  
For the second task, $\mathcal{L}^{(t',t)}$ might make some oracle queries. By \Cref{eq:OracleImplementation},  each query can be answered by multiplying a row of $(T^{(0,t')})^{-1}$ with a column of $S^{(t',t)}$.

When, say,  the  $i^{th}$ row of $(A^{(t)})^{-1}$ is queried, we need to multiply a row of $(T^{(t',t)})^{-1}$ with $(T^{(0,t')})^{-1}(A^{(0)})^{-1}$. Such row is obtained by making a query to algorithm  $\mathcal{L}^{(t',t)}$; again, we use $(T^{(0,t')})^{-1}$ and $S^{(t',t)}$ to answer oracle queries made by $\mathcal{L}^{(t',t)}$.  When multiplying the vector-matrix-matrix product from left to right, each vector-matrix product takes time linear to the product of the number of non-zero entries in the vector and the number of non-identity columns in the matrix.

\Cref{sub:combineT} describes in details how we implement the two operations above.

The running time $\mathcal{L}^{(t',t)}$ depend on $t-t'$.  When $t-t'$ gets large, we ``reset'' $T^{(t',t)}$ to $\I$ by setting $t'\leftarrow t$ and compute $T^{(0, t)} = T^{(0, t')} T^{(t',t)}$ using fast matrix multiplication. The latter is done in a similar way to Sankowski's first algorithm. In particular, we write down equations similar to \Cref{eq:SankFirstDetail}, except that now we have $C=[T^{(0,t)}-T^{(0,t')}]$. Given that $C$ is quite dense (since $t'\ll t$), we can exploit fast matrix multiplication here while the original algorithm that uses \Cref{eq:SankFirstDetail} cannot.  See details in \Cref{sub:Texplicit}.

Computing $T^{(0, t)} = T^{(0, t')} T^{(t',t)}$ also becomes slow when $t$ is large. In this case, we ``reset'' both $T^{(0,t')}$ and $T^{(t',t)}$ to $\I$ by computing $A^{(t)} = A^{(0)}T^{(0,t')}T^{(t',t)}$ and pretend that $A^{(t)}$ is our new $A^{(0)}$. Once again we can exploit fast matrix multiplication here. See details in \Cref{sub:TimpliesInverse}.

\paragraph{Discussions.} Now that we can exploit fast matrix multiplication one more time compared to previous algorithms, it is natural to ask whether we can exploit it another time. A technical obstacle is that to use  fast matrix multiplication twice we already have to solve a different problem (\Cref{prob:InverseTransformImplicit} vs. \Cref{prob:InverseTransform}); thus it is unclear whether and how we should define another problem to be able to use fast matrix multiplication another time. A more fundamental obstacle is our conjectures: to get any further improvement we have to break these conjectures, as we will discuss in \Cref{sec:lowerBounds}.

%% file: preliminaries.tex
\input{applicationstable.tex}

\section{Preliminaries}
\label{sec:preliminaries}

In this section we will define our notation and state some simple results about matrix multiplication and inversion.

\paragraph{Notation: Identity and Submatrices}
The identity matrix is denoted by $\I$.

Let $I, J \subset [n] := \{1,...,n\}$ and $A$ be an $n \times n$ matrix, then the term $A_{I,J}$ denotes the submatrix of $A$ consisting of the rows $I$ and columns $J$. For some $i \in [n]$ the term $A_{[n],i}$ can thus be seen as the $i$th column of $A$.

Let $I = (i_1,...i_p) \in [n]^p$, $J = (j_1,...j_q) \in [n]^q$ and $A$ be an $n \times n$ matrix, then the term $A_{I,J}$ denotes a matrix such that $(A_{I,J})_{s,t} = A_{i_s,j_t}$. Specifically for $i_1 < ... < i_p$ and $j_1 < ... < j_t$ the term $A_{I,J}$ is just the submatrix of $A$ when interpreting $I$ and $J$ as sets instead of vectors.

We may also mix the notation e.g. for $I = (i_1,...i_p) \in [n]^p$ and $J \subset [n]$, we can consider $J$ to be an ordered set such that $j_1 < ... <j_q$, then the term $A_{I,J}$ is just the matrix where $(A_{I,J})_{s,t} = A_{i_s,j_t}$.

\paragraph{Inner, Outer and Matrix Products}

Given two vectors $u$ and $v$ we will write $u^\top v$ for the inner product and $u v^\top$ for the outer product. This way inner and outer product are just special cases of matrix multiplication, i.e. inner product is a $1 \times n$ matrix multiplied with an $n \times 1$ matrix, while an outer product is the product of an $n \times 1$ matrix by a $1 \times n$ matrix.

We will also often exploit the fact that each entry of a matrix product is given by an inner product: $(AB)_{i,j} = \sum_{k=1}^n A_{i,k} B_{k,j} = A_{i,[n]}B_{[n],j} = e_i AB e_j$. In other words, to compute entry $(i,j)$ of $AB$ we just multiply the $i$th row of $A$ with the $j$th column of $B$.

\paragraph{Fast Matrix Multiplication}

We denote with $O(n^\omega)$ the complexity of multiplying two $n \times n$ matrices. Note that matrix multiplication, inversion, determinant and rank, all have the same complexity \cite[Chapter 16]{algebraicComplexity}.
Currently the best bound is $\omega < \matrixExponent$ \cite{Gall14a}.

For rectangular matrices we denote the complexity of multiplying an $n^a \times n^b$ matrix with an $n^b \times n^c$ matrix with $O(n^{\omega(a,b,c)})$ for any $0 \le a,b,c$. Note that $\omega( \cdot, \cdot, \cdot)$ is a symmetric function so we are allowed to reorder the arguments. The currently best bounds for $\omega(1,1,c)$ can be found in \cite{GallU18}.

The complexity of the algorithms presented in this paper depend on the complexity of multiplying and inverting matrices. For a more in-depth analysis of how we balance the terms that depend on $\omega$ (e.g. how we compute $\omega(a,b,c)$ for $a,b \neq 1$), we refer to the appendix \ref{app:runtime}.

\paragraph{Transformation Matrices}
Throughout this paper, we will often have matrices of the form $T = \I+C$, where $C$ has few non-zero columns. We will often call these matrices transformation matrices.

Note that any matrix $T=\I+C$, where $C$ has at most $m$ non-zero columns, can be brought in the following form by permuting the rows and columns, which corresponds to permuting the columns and rows of its inverse $T^{-1}$ \cite[Section 5]{Sankowski04}:

$$
T = \left(
\begin{array}{c|c}
C_1 & 0 \\
\hline
C_2 & \I
\end{array}
\right)
$$
Here $C_1$ is of size $m \times m$ and $C_2$ of size $(n-m) \times m$. The inverse is given by
$$
T^{-1} = \left(
\begin{array}{c|c}
C_1^{-1} & 0 \\
\hline
-C_2 C_1^{-1} & \I
\end{array}
\right)
$$
In general, without prior permutation of rows/columns, we can state for $T$ and its inverse the following facts:

\begin{fact}\label{lem:Tinverse}
Let $T$ be an $n \times n$ matrix of the form $\I + C$ and let $J \subset [n]$ be the column indices of the non-zero columns of $C$, and thus for $K := [n] \backslash J$ we have $C_{[n],K} = 0$.

Then:
\begin{itemize}
\item $(T^{-1})_{J,J} = (T_{J,J})^{-1}$, $(T^{-1})_{K, J} = -T_{K, J}(T_{J,J})^{-1}$, $(T^{-1})_{K, K} = \I$ and $(T^{-1})_{J, K} = 0$.
\item For $|J| = n^{\delta}$ the inverse $T^{-1}$ can be computed in $O(n^{\omega(1,\delta,\delta)})$ field operations (\Cref{alg:invert}).
\item If given some set $I \subset [n]$ with $J \subset I$, $|I| = n^{\varepsilon}$, then rows $I$ of $T^{-1}$ can be computed in $O(n^{\omega(\varepsilon,\delta,\delta)})$ and for this we only need to know the rows $I$ of $T$. (\Cref{alg:partialInvert})
\end{itemize}
\end{fact}

\begin{algorithm}
\caption{Invert (\Cref{lem:Tinverse})}\label{alg:invert}
\begin{algorithmic}[1]
\REQUIRE $n \times n$ matrix $A=\I+C$ where $J\subset [n]$ are the indices of the nonzero columns of $C$.
\ENSURE  $A^{-1} = \I+\tilde{C}$


\STATE $\tilde{C}_{J,J} \leftarrow (A_{J,J})^{-1}$ \label{line:invertPolyinversion}
\STATE $\tilde{C}_{[n] \backslash J,J} \leftarrow -C_{[n] \backslash J, J} \tilde{C}_{J,J}$
\RETURN $\I+\tilde{C}$

\end{algorithmic}
\end{algorithm}

\begin{algorithm}
\caption{PartialInvert (\Cref{lem:Tinverse})}\label{alg:partialInvert}
\begin{algorithmic}[1]
\REQUIRE Rows $I \subset [n]$ of a matrix $A=\I+C$ where $J\subset [n]$ are the indices of the nonzero columns of $C$ and $J \subset I$.
\ENSURE  Rows $I$ of $A^{-1} = \I+\tilde{C}$


\STATE $\tilde{C}_{J,J} \leftarrow (A_{J,J})^{-1}$
\STATE $\tilde{C}_{I \backslash J,J} \leftarrow -C_{I \backslash J, J} \tilde{C}_{J,J}$
\RETURN $\I+\tilde{C}$

\end{algorithmic}
\end{algorithm}

We will often multiply matrices of the form $\I+C$ where $C$ has few non-zero columns. The complexity of such multiplications is as follows:

\begin{fact}\label{lem:complexityTproduct}
Let $A,B$ be $n \times n$ matrices of the form $A = \I+C$, $B=\I+N$, where $C$ has $n^a$ non-zero columns and $N$ has $n^b$ non-zero columns.

Then:

\begin{itemize}
\item The product $AB$ can be computed in $O(n^{\omega(1,a,b)})$ operations.
\item If $J_C,J_N \subset [n]$ are the sets of column indices where $C$ (or respectively $N$) is non-zero, then $AB$ is of the form $AB = \I + M$ where $M$ can only be non-zero on columns with index in $J_C \cup J_N$.
\item If we want to compute only a subset of the rows, i.e. for $I \subset [n]$ we want to compute $(AB)_{I,[n]} = A_{I,[n]}B$, then for $|I| = n^{c}$ this requires $O(n^{\omega(c,a,b)})$ operations.

For this we only require the rows with index $I \cup J_C$ of the matrix $N$, so we do not have to know the other entries of $N$ to compute the product.
\end{itemize} 

\end{fact}

This fact is a direct implication of $(I+C)(I+N) = I+C+N+CN$ and $CN_{[n],J_N} = C_{[n],J_C}N_{J_C,J_N}$.

%% file: applicationstable.tex
\newcommand{\element}{\elementColor{$O(n^{\fastExponent})$ $[O(n^{1+1/4})]$}}
\newcommand{\column}{\columnColor{$O(n^{\slowExponent})$ $[O(n^{1.5})]$}}
\newcommand{\elementLB}[1]{\elementColor{$#1\Omega(n^{\fastExponentLB})$ $[\Omega(n^{1+1/4})]$}}
\newcommand{\columnLB}[1]{\columnColor{$#1\Omega(n^{\slowExponentLB})$ $[\Omega(n^{1.5})]$}}
\newcommand{\indentOperation}{\hspace{10pt}}
\newcommand{\AW}{{$\Omega(m^{0.814})^{\ddagger}$}}

\begin{figure}
\tiny

\begin{tabular}{>{\raggedright}p{0.18\textwidth}|>{\raggedright}p{0.2\textwidth}|>{\raggedright}p{0.15\textwidth}|>{\raggedright}p{0.16\textwidth}|>{\raggedright}p{0.2\textwidth}}
Problem  & Known upper bound  & New upper bound  & Known lower bound  & New lower bound\tabularnewline
\hline 
\hline 
Bipartite maximum matching  &  &  & -  & - \tabularnewline
(online, total time)  & $O(m\sqrt{n})$ \cite{BosekLSZ14}  & $O(n^{\omega})$  & -  & - \tabularnewline
\hline 
Bipartite maximum matching  &  &  &  & \tabularnewline
(fully dynamic)  &  &  &  & \tabularnewline
\indentOperation edge update  & $O(n^{1.447})$ $[O(n^{1+1/3})]$ \cite{Sankowski07}  & \element{}  & $\Omega(n)$ \cite{HenzingerKNS15}

\AW{} \cite{AbboudW14}  & \elementLB{} \Cref{cor:lowerBoundTransitiveClosure} \tabularnewline
\hline 
\indentOperation right side node update & $O(n^{2})$ \cite{Sankowski07}  & \column{}  & \multicolumn{2}{c}{same as above}\tabularnewline
\hline 
Maximum matching  &  &  &  & \tabularnewline
(general graphs)  &  &  & $\Omega(n)$ \cite{HenzingerKNS15}  & \tabularnewline
\indentOperation edge update  & $O(n^{1.447})$ $[O(n^{1+1/3})]$ \cite{Sankowski07}  & \element{}  & \AW{} \cite{AbboudW14}  & \elementLB{} \Cref{cor:lowerBoundTransitiveClosure} \tabularnewline
\hline 
\hline 
DAG path counting$^\dagger$ and

Transitive Closure  &  &  &  & \tabularnewline
\indentOperation edge update  & $O(n^{\OLDfastExponent})$ $[O(n^{1+1/3})]$ \cite{Sankowski04}  & \element{}  & $u+q=\Omega(n)$ \cite{HenzingerKNS15}  & \elementLB{u+q=} \tabularnewline
\indentOperation pair query  & $O(n^{\OLDfastExponent})$ $[O(n^{1+1/3})]$ \cite{Sankowski04}  & \element{}  & $u+q=$\AW{} \cite{AbboudW14}  & \Cref{cor:lowerBoundTransitiveClosure} \tabularnewline
\hline 
\indentOperation same as above  & $O(n^{\slowExponent})$ $[O(n^{1.5})]$ \cite{Sankowski04}  &  &  & \mixedColor{$u=\Omega(n^{\slowExponentLB})$ $[\Omega(n^{1.5})]$}
or \tabularnewline
 & $O(n^{\slowExponentQuery})$ $[O(n^{0.5})]$ \cite{Sankowski04}  &  &  & \mixedColor{$q=\Omega(n^{\slowExponentQueryLB})$ $[\Omega(n^{0.5})]$}
\Cref{cor:lowerBoundTransitiveClosureElementSource} \tabularnewline
\hline 
\indentOperation node update

\indentOperation (incoming edges)  & $O(n^{2})$ \cite{Sankowski04}  & \column{}  & $u+q=\Omega(n)$ \cite{HenzingerKNS15}

$u+q=$\AW{}  & \columnLB{u+q=} \Cref{thm:lowerBoundTransClosure} \tabularnewline
\indentOperation source query  & $O(n)$ \cite{Sankowski04}  & \column{}  & \cite{AbboudW14}  & \tabularnewline
\hline 
\indentOperation edge update  & $O(n^{\slowExponent})$ $[O(n^{1.5})]$ \cite{Sankowski04}  & -  & $n\cdot u+q=\Omega(n^{2})$  & \mixedColor{$u+q=\Omega(n^{\slowExponentLB})$ $[\Omega(n^{1.5})]$} \tabularnewline
\indentOperation source query  & $O(n^{\slowExponent})$ $[O(n^{1.5})]$ \cite{Sankowski04}  &  & \cite{HenzingerKNS15}  & \Cref{cor:lowerBoundTransitiveClosureElementSource} \tabularnewline
\hline 
\hline 
All-pair-distances (unweighted) &  &  & \multicolumn{1}{>{\raggedright}p{0.16\textwidth}}{} & \tabularnewline
\indentOperation edge update & $O(n^{1.897})$ $[\tilde{O}(n^{2-1/8})]$ \cite{Sankowski05}  & \elementColor{$O(n^{1.724})$ $[\tilde{O}(n^{1+2/3})]$}  & \multicolumn{2}{c}{same as edge update/pair query transitive closure}\tabularnewline
\indentOperation pair query & $O(n^{1.265})$ $[\tilde{O}(n^{1+1/4})]$ \cite{Sankowski05}  & \elementColor{$O(n^{1.724})$ $[\tilde{O}(n^{1+2/3})]$}  & \multicolumn{1}{>{\raggedright}p{0.16\textwidth}}{} & \tabularnewline
\hline 
\hline 
Strong connectivity  &  &  &  & \tabularnewline
\indentOperation edge update & $O(n^{\slowExponent})${*} $[O(n^{1+1/2})]$ \cite{Sankowski04}  & \column{}  & $\Omega(n)$ \cite{HenzingerKNS15}

\AW \cite{AbboudW14}  & \elementLB{} \Cref{cor:lowerBoundTransitiveClosure} \tabularnewline
\hline 
\indentOperation node update

\indentOperation (incoming edges)  & $O(n^{2})${*} \cite{Sankowski04}  & \column{}  & \multicolumn{2}{c}{same as above}\tabularnewline
\hline 
\hline 
Counting node disjoint $ST$-paths  &  &  &  & \tabularnewline
\indentOperation edge update  & $O(n^{\OLDfastExponent})$ \cite{Sankowski07}  & \element{}  & $\Omega(n)$ \cite{HenzingerKNS15}

\AW \cite{AbboudW14}  & \elementLB{}

(via transitive closure) \tabularnewline
\hline 
\hline 
Counting spanning trees$\dagger$  &  &  &  & \tabularnewline
\indentOperation edge update  & $O(n^{\OLDfastExponent})$ $[O(n^{1+1/3})]$ \cite{Sankowski04}  & \element{}  & -  & -\tabularnewline
\hline 
\hline 
Triangle detection  &  &  &  & \tabularnewline
\indentOperation node update  & $O(n^{2})$ {[}trivial{]}  & -  & $\Omega(n^{2})$ \cite{HenzingerKNS15}  & - \tabularnewline
\hline 
\indentOperation node update

\indentOperation (incoming edges)  & $O(n^{2})$ {[}trivial{]}  & \column{}  & $\Omega(n)$ \cite{HenzingerKNS15}  & -\tabularnewline
\hline 
\indentOperation node update

\indentOperation (turn node on/off)  & $O(n^{2})$ {[}trivial{]}  & \element{}  & - & -\tabularnewline
\hline 
\hline 
Cycle detection and

$k$-cycle (constant $k$)  &  &  &  & \tabularnewline
\indentOperation edge update  & $O(n^{\OLDfastExponent})${*} $[O(n^{1+1/3})]$ \cite{Sankowski05}  & \element{}  & $\Omega(n)$ ($k\ge3$) \cite{HenzingerKNS15} & \elementLB{} ($k\ge6$) \Cref{cor:lowerBoundTransitiveClosure} \tabularnewline
\hline 
\indentOperation node update

\indentOperation (incoming edges)  & $O(n^{2})${*} \cite{Sankowski05}  & \column{}  & \multicolumn{2}{c}{same as above }\tabularnewline
\hline 
\hline 
$k$-path (constant $k$)  &  &  &  & \tabularnewline
\indentOperation edge update  & $O(n^{\OLDfastExponent})${*} $[O(n^{1+1/3})]$ \cite{Sankowski04}  & \element{}  & same as transitive closure & same as transitive closure\tabularnewline
\indentOperation pair query  & $O(n^{\OLDfastExponent})${*} $[O(n^{1+1/3})]$ \cite{Sankowski04}  & \element{}  & for $k\ge3$ & for $k\ge5$\tabularnewline
\hline 
\indentOperation node update

\indentOperation (incoming edges)  & $O(n^{2})${*} \cite{Sankowski04}  & \column{}  & same as transitive closure

for $k\ge3$ & same as transitive closure

for $k\ge3$\tabularnewline
\indentOperation source query  & $O(n)${*} \cite{Sankowski04}  & \column{}  &  & \tabularnewline
\end{tabular}

\caption{
The table displays the previous best upper/lower bounds and our results for dynamic gaph problems.
Bounds inside brackets $[ \: \cdot \: ]$ are valid assuming that a linear-time matrix multiplication algorithm exists.
Lower bounds marked with $\ddagger$ only hold for sparse graphs $m = O(n^{1.67})$ and when assuming $O(m^{1.407})$ pre-processing time.
The complexities for problems marked with $\dagger$ are measured in the number of arithmetic operations.
All other complexities measure the time.
Bounds marked with * are new applications of dynamic matrix inverse that were not previously stated. 
Colors of upper bounds indicate which bound in \Cref{tbl:dynamic inverse} each bound in this table follows from. 
Colors of lower bounds indicate which conjecture in  \Cref{tbl:dynamic inverse} each bound in this table follows from. 
For details, see \Cref{sec:applications}.
}
\label{tbl:applications}
\end{figure}

\begin{figure}
\tiny
\begin{tabular}{>{\raggedright}p{0.18\textwidth}|>{\raggedright}p{0.2\textwidth}|>{\raggedright}p{0.15\textwidth}|>{\raggedright}p{0.16\textwidth}|>{\raggedright}p{0.2\textwidth}}
Problem  & Known upper bound  & New upper bound  & Known lower bound  & New lower bound\tabularnewline
\hline 
\hline 
Largest Eigenvalue  &  &  &  & \tabularnewline
\indentOperation entry update  & -  & \element{}  & -  & - \tabularnewline
\hline 
\indentOperation column update

\indentOperation (output eigenvector+value)  & $\tilde{O}(n^{2})$ \cite{FrandsenS11}

(supports rank 1 updates)  & \column{}  & -  & - \tabularnewline
\hline 
\hline 
Pseudo-inverse  &  &  &  & \tabularnewline
\indentOperation row scaling  & $O(n^{\slowExponent})${*} $[O(n^{1.5})]$ \cite{Sankowski04}  & -  & -  & - \tabularnewline
\indentOperation column query  & $O(n^{\slowExponent})${*} $[O(n^{1.5})]$ \cite{Sankowski04}  &  &  & \tabularnewline
\hline 
\indentOperation row scaling  & $O(n^{\OLDfastExponent})${*} $[O(n^{1+1/3})]$ \cite{Sankowski04}  & \element{}  & -  & - \tabularnewline
\indentOperation element query  & $O(n^{\OLDfastExponent})${*} $[O(n^{1+1/3})]$ \cite{Sankowski04}  & \element{}  &  & \tabularnewline
\hline 
\hline 
Linear system  &  &  &  & \tabularnewline
\indentOperation element update  & $O(n^{\OLDfastExponent})$ $[O(n^{1+1/3})]$ \cite{Sankowski04}  & \element{}  & $u+q=\Omega(n)$ \cite{HenzingerKNS15} & \elementLB{u+q=}\tabularnewline
\indentOperation element query  & $O(n^{\OLDfastExponent})$ $[O(n^{1+1/3})]$ \cite{Sankowski04}  & \element{}  &  & \Cref{cor:elementUpdateLB} \tabularnewline
\hline 
\indentOperation constraint update  & $O(n^{2})$  & \column{}  & -  & -\tabularnewline
\hline 
\indentOperation row+column update  & $O(n^{2})$  & -  & -  & $u+q=\Omega(n^{2})$\tabularnewline
\indentOperation element query  & $O(1)$  &  &  & \Cref{thm:columnRowUpdateLB} \tabularnewline
\hline 
\hline 
2-matrix product  &  &  &  & \tabularnewline
\indentOperation element update  & $O(n)$ {[}trivial{]}  & -  & $u+q=\Omega(n)$ \cite{HenzingerKNS15} & -\tabularnewline
\indentOperation element query  & $O(1)$ {[}trivial{]}  &  &  & \tabularnewline
\hline 
\indentOperation column update  & $O(n^{2})$ {[}trivial{]}  & \column{}  & -  & \columnLB{u+q=}\tabularnewline
\indentOperation row query  & $O(1)$ {[}trivial{]}  & \column{}  &  & \Cref{thm:lowerBound2MatrixProd} \tabularnewline
\hline 
\hline 
$k$-matrix product (constant $k$)  &  &  &  & \tabularnewline
\indentOperation element update  & $O(n^{\OLDfastExponent})${*} $[O(n^{1+1/3})]$ \cite{Sankowski04}  & \elementColor{$O(n^{\fastExponent})$ $[O(n^{1+1/4})]$}  & $u+q=\Omega(n)$ \cite{HenzingerKNS15} & \elementLB{u+q=}\tabularnewline
\indentOperation element query  & $O(n^{\OLDfastExponent})${*} $[O(n^{1+1/3})]$ \cite{Sankowski04}  & \elementColor{$O(n^{\fastExponent})$ $[O(n^{1+1/4})]$}  &  & for $k\ge5$ \Cref{thm:lowerBound5MatrixProd} \tabularnewline
\hline 
\indentOperation column update  & $O(n^{2})$  & \columnColor{$O(n^{\slowExponent})$ $[O(n^{1.5})]$}  & -  & \columnLB{u+q=} \tabularnewline
\indentOperation row query  & $O(n)$  & \columnColor{$O(n^{\slowExponent})$ $[O(n^{1.5})]$}  &  & \Cref{thm:lowerBound2MatrixProd} \tabularnewline
\hline 
\indentOperation row+column update  & $O(n^{2})$  & -  & $u+q=\Omega(n^{2})$ \cite{HenzingerKNS15} & - \tabularnewline
\indentOperation element query  & $O(1)$  &  &  & \tabularnewline
\hline 
\hline 
Determinant  &  &  &  & \tabularnewline
\indentOperation element update  & $O(n^{\OLDfastExponent})$ $[O(n^{1+1/3})]$ \cite{Sankowski04}  & \element{}  & $\Omega(n)$ \cite{HenzingerKNS15}  & \elementLB{} \Cref{cor:elementUpdateLB} \tabularnewline
\hline 
\indentOperation column update  & $O(n^{2})$  & \column{}  & -  & \tabularnewline
\hline 
\indentOperation row+column update  & $O(n^{2})$  & -  & -  & $\Omega(n^{2})$ \Cref{thm:columnRowUpdateLB} \tabularnewline
\hline 
\hline 
Adjoint  &  &  &  & \tabularnewline
\indentOperation element update  & $O(n^{\OLDfastExponent})$ $[O(n^{1+1/3})]$ \cite{Sankowski04}  & \element{}  & $u+q=\Omega(n)$ \cite{HenzingerKNS15} & \elementLB{u+q=}\tabularnewline
\indentOperation element query  & $O(n^{\OLDfastExponent})$ $[O(n^{1+1/3})]$ \cite{Sankowski04}  & \element{}  &  & \Cref{cor:elementUpdateLB} \tabularnewline
\hline 
\indentOperation column update  & $O(n^{2})$  & \column{}  & -  & \columnLB{u+q=} \tabularnewline
\indentOperation row query  & $O(n)$  & \column{}  &  & \Cref{cor:columnUpdateLB} \tabularnewline
\hline 
\indentOperation row+column update  & $O(n^{2})$  & -  & $u+q=\Omega(n^{2})$ \cite{HenzingerKNS15} & - \tabularnewline
\indentOperation element query  & $O(1)$  &  &  & \tabularnewline
\hline 
\hline 
Rank  &  &  &  & \tabularnewline
\indentOperation element update  & $O(n^{\OLDfastExponent})$ $[O(n^{1+1/3})]$ \cite{Sankowski07}  & \element{}  & $\Omega(n)$ \cite{HenzingerKNS15}  & \elementLB{} \Cref{cor:elementUpdateLB} \tabularnewline
\hline 
\indentOperation column update  & $O(n^{2})$ \cite{Sankowski07}  & \column{}  & -  & -\tabularnewline
\hline 
\indentOperation row+column update  & $O(n^{2})$ \cite{Sankowski07}  & -  & -  & $\Omega(n^{2})$ \Cref{thm:columnRowUpdateLB} \tabularnewline
\hline 
\hline 
Interpolation polynomial  &  &  &  & \tabularnewline
\indentOperation point update  & $O(n^{2})$  & \column{}  & $\Omega(n)$ {[}trivial{]}  & - \tabularnewline
\end{tabular}

\caption{The table displays the previous best upper/lower bounds and our results for dynamic algebraic problems. Bounds inside brackets $[ \: \cdot \: ]$ are valid assuming that a linear-time rectangular matrix multiplication algorithm exists. 
The complexities are all measured in the number of arithmetic operations, because the algorithms work over any field.
Bounds marked with * are new applications of dynamic matrix inverse that were not previously stated. 
Colors of upper bounds indicate which bound in \Cref{tbl:dynamic inverse} each bound in this table follows from. 
Colors of lower bounds indicate which conjecture in  \Cref{tbl:dynamic inverse} each bound in this table follows from. 
For details, see \Cref{sec:applications}.
}
\label{tbl:applications2}

\end{figure}

%% file: inverse.tex
\section{Dynamic Matrix Inverse}
\label{sec:dynamicInverse}

In this section we show the main algorithmic result, which are two algorithms for dynamic matrix inverse. The first one supports column updates and row queries, while the second one supports element updates and element queries. These two algorithms imply more than ten faster dynamic algorithms, see \Cref{tbl:applications,tbl:applications2,sec:applications} for applications. 

\begin{theorem}\label{thm:columnUpdate}
For every $0 \le \varepsilon \le 1$ there exists a dynamic algorithm for maintaining the inverse of an $n \times n$ matrix $A$, requiring $O(n^\omega)$ field operations during the pre-processing. The algorithm supports changing any column of $A$ in $O(n^{1+\varepsilon}+n^{\omega(1,1,\varepsilon)-\varepsilon})$ field operations and querying any row of $A^{-1}$ in $O(n^{1+\varepsilon})$ field operations.

For current bounds on $\omega$ this implies a $O(n^{\slowExponent})$ upper bound on the update and query cost ($\varepsilon \approx 0.723$), see Appendix \ref{app:runtime}. For $\omega = 2$ the update and query time become $O(n^{1.5})$ ($\varepsilon = 0.5$).
\end{theorem}

\begin{theorem}\label{thm:elementUpdate}
For every $0 \le \varepsilon_1 \le \varepsilon_2 \le 1$ there exists a dynamic algorithm for maintaining the inverse of an $n \times n$ matrix $A$, requiring $O(n^\omega)$ field operations during the pre-processing, The algorithm supports changing any entry of $A$ in $O(n^{\varepsilon_2 +\varepsilon_1} + n^{\omega(1,\varepsilon_1,\varepsilon_2)-\varepsilon_1} + n^{\omega(1,1,\varepsilon_2)-\varepsilon_2})$ field operations and querying any entry of $A^{-1}$ in $O(n^{\varepsilon_2 +\varepsilon_1})$ field operations.

When balancing the terms for current values of $\omega$, the update and query cost are $O(n^\fastExponent)$ (for $\varepsilon_1 \approx 0.551$, $\varepsilon_2 \approx 0.855$), see Appendix \ref{app:runtime}. For $\omega = 2$ the update and query time become $O(n^{1.25})$ (for $\varepsilon_1 = 0.5$, $\varepsilon_2 = 0.75$).
\end{theorem}

Throughout this section, we will write $A^{(t)}$ to denote the matrix $A$ after $t$ updates. The algorithms from both \Cref{thm:columnUpdate} and \Cref{thm:elementUpdate} are based on Sankowski's idea \cite{Sankowski04} of expressing the change of some matrix $A^{(t-1)}$ to $A^{(t)}$ via a linear transformation $T^{(0,t)}$, such that $A^{(t)} = A^{(0)}T^{(0,t)}$ and thus $(A^{(t)})^{-1} = (T^{(0,t)})^{-1}(A^{(0)})^{-1}$.
The task of maintaining the inverse of $A^{(t)}$ thus becomes a task about maintaining the inverse of $T^{(0,t)}$. We will call this problem \emph{transformation maintenance} and the properties for this task will be properly defined in \Cref{sub:TimpliesInverse}. We note that proofs in \Cref{sub:TimpliesInverse} essentially follow ideas from \cite{Sankowski04}, but Sankowski did not state his result in exactly the form that we need.

In the following two subsections \ref{sub:Texplicit} and \ref{sub:Timplicit}, we describe two algorithms for this transformation maintenance problem. We are able to combine these two algorithms to get an even faster transformation maintenance algorithm in subsection \ref{sub:combineT}, where we will also prove the main results \Cref{thm:columnUpdate} and \Cref{thm:elementUpdate}.

Throughout this section we will assume that $A^{(t)}$ is invertible for every $t$. An extension to the case where $A^{(t)}$ is allowed to become singular is given by \Cref{thm:singularUpdates}.

\subsection{Transformation Maintenance implies Dynamic Matrix Inverse}
\label{sub:TimpliesInverse}

In the overview \Cref{sec:overview} we outlined that maintaining the inverse for some transformation matrix $T^{(0,t)}$ implies an algorithm for maintaining the inverse of matrix $A^{(t)}$.
In this section we will formalize and prove this claim in the setting where $A^{(t)}$ receives entry updates.

\begin{theorem}\label{thm:TimpliesInverseElement}
Assume there exists a dynamic algorithm $\mathcal{T}$ that maintains the inverse of an $n \times n$ matrix $M^{(t)}$ where $M^{(0)} = \I$, supporting the following operations:

\begin{itemize}
	\item \emph{update($j_1,...,j_k$,$c_1,...,c_k$)} Set the $j_l$th column of $M^{(t)}$ to be the vector $c_l$ for $l=1...k$ in $O(u(k,m))$ field operations, where $m$ is the number of so far changed columns.

	\item \emph{query($I$)} Output the rows of $(M^{(t)})^{-1}$ specified by the set $I \subset [n]$ in $O(q(|I|,m))$ field operations, where $m$ is the number of so far changed columns.
\end{itemize}

Also assume the pre-processing of this algorithm requires $O(p)$ field operations.

Let $k \le n^\varepsilon$ for $0 \le \varepsilon \le 1$, then there exists a dynamic algorithm $\mathcal{A}$ that maintains the inverse of any (non-singular) matrix $A$ supporting the following operations:
\begin{itemize}
\item \emph{update($(i_1,j_1)...(i_k,j_k),c_1...c_k$)} Set $A^{(t)}_{i_l,j_l}$ to be $c_l$ for $l=1...k$ in $O(u(k,n^{\varepsilon}) + (kn^{-\varepsilon}) \cdot (p + n^{\omega(1,1,\varepsilon)}))$ field operations.
\item \emph{query($I,J$)} Output the sub-matrix $(A^{(t)})^{-1}_{I,J}$ specified by the sets $I,J \subset [n]$ with $|I|=n^{\delta_1}, |J| = n^{\delta_2}$ in $O(q(n^{\delta_1},n^{\varepsilon}) + n^{\omega(\delta_1,\varepsilon,\delta_2)})$ field operations.
\end{itemize}
The pre-processing requires $O(p + n^\omega)$ field operations.
\end{theorem}

The high level idea of the algorithm $\mathcal{A}$ is to maintain $T^{(0,t)}$ such that $A^{(t)} = A^{(0)} T^{(0,t)}$, which allows us to express the inverse of $A^{(t)}$ via $(T^{(0,t)})^{-1}(A^{(0)})^{-1}$. Here the matrix $(A^{(0)})^{-1}$ is computed during the pre-processing and $(T^{(0,t)})^{-1}$ is maintained via the assumed algorithm $\mathcal{T}$. After changing $n^{\varepsilon}$ entries of $A$, we reset the algorithm by computing $(A^{(t)})^{-1}$ explicitly and resetting $T^{(t,t)} = \I$.
We will first prove that element updates to $A$ correspond to column updates to $T$.

\begin{lemma}\label{lem:formulaForT}
Let $A^{(t_1)}$ and $A^{(t_2)}$ be two non-singular matrices, then there exists a matrix $T^{(t_1, t_2)} := \I + (A^{(t_1)})^{-1}(A^{(t_2)}-A^{(t_1)})$ such that $A^{(t_2)} = A^{(t_1)} T^{(t_1,t_2)}$.
\end{lemma}

\begin{proof}
We have $T^{(t_1,t_2)} = \I + (A^{(t_1)})^{-1}(A^{(t_2)}-A^{(t_1)})$, because:
\begin{align*}
A^{(t_1)} \left[ \I + (A^{(t_1)})^{-1}(A^{(t_2)}-A^{(t_1)})\right] 
=
A^{(t_1)} + (A^{(t_2)}-A^{(t_1)}) = A^{(t_2)}
\end{align*}
\end{proof}

\begin{corollary}\label{cor:elementIsColumnUpdate}
Let $0 \le t' \le t$ and $A^{(t)} = A^{(t')} T^{(t',t)}$, where $A^{(t')}$ and $A^{(t)}$ differ in at most $k$ columns. Then

\begin{itemize}
\item An entry update to $A^{(t)}$ corresponds to a column update to $T^{(t',t)}$, where the column update is given by a column of $(A^{(t')})^{-1}$, multiplied by some scalar.
\item The matrix $T^{(t',t)}$ is of the form $\I + C$, where $C$ has at most $k$ non-zero columns.
\end{itemize}
\end{corollary}

\begin{proof}
The first property comes from the fact that
\begin{align*}
T^{(t', t)}
&= \I + (A^{(t')})^{-1}(A^{(t)}-A^{(t')}) \\
&= \I + (A^{(t')})^{-1}(A^{(t)}-A^{(t-1)}+A^{(t-1)}-A^{(t')})\\
&= (A^{(t')})^{-1}(A^{(t)}-A^{(t-1)}) + \I + (A^{(t')})^{-1}(A^{(t-1)}-A^{(t')})\\
&= (A^{(t')})^{-1}(A^{(t)}-A^{(t-1)}) + T^{(t',t-1)}
\end{align*}
and $(A^{(t)}-A^{(t-1)})$ is a zero matrix except for a single entry.
Thus $(A^{(t')})^{-1}(A^{(t)}-A^{(t-1)})$ is just one column of $(A^{(t')})^{-1}$ multiplied by the non-zero entry of $(A^{(t)}-A^{(t-1)})$.

The second property is a direct implication of $T^{(t', t)} = \I + (A^{(t')})^{-1}(A^{(t)}-A^{(t')})$ as $(A^{(t)}-A^{(t')})$ is non-zero in at most $k$ columns.
\end{proof}

\begin{proof}[Proof of \Cref{thm:TimpliesInverseElement}]
We are given a dynamic algorithm $\mathcal{T}$ that maintains the inverse of an $n \times n$ matrix $M^{(t)}$ where $M^{(0)} = \I$, supporting column updates to $M^{(t)}$ and row queries to $(M^{(t)})^{-1}$. We now want to use this algorithm to maintain $(A^{(t)})^{-1}$.

\paragraph{Pre-processing} During the pre-processing we compute $(A^{(0)})^{-1}$ explicitly in $O(n^\omega)$ field operations and initialize the algorithm $\mathcal{T}$ in $O(p)$ operations.

\paragraph{Updates}
We use algorithm $\mathcal{T}$ to maintain the inverse of $M^{(t)} := T^{(0,t)}$, where $T^{(0,t)}$ is the linear transformation transforming $A^{(0)}$ to $A^{(t)}$.
Via \Cref{cor:elementIsColumnUpdate} we know the updates to $A^{(t)}$ imply column updates to $T^{(0,t)}$, so we can use algorithm $\mathcal{T}$ for this task.
\Cref{cor:elementIsColumnUpdate} also tells us that the update performed to $M^{(t)} = T^{(0,t)}$ is simply given by a scaled column of $(A^{(0)})^{-1}$, so it is easy to obtain the change we have to perform to $M^{(t)}$.

\paragraph{Reset and average update cost}

For the first $n^{\varepsilon}$ columns that are changed in $T^{(0,t)}$, each update requires at most $O(u(k,n^{\varepsilon}))$ field operations.
After changing $n^{\varepsilon}$ columns we reset our algorithm, but instead of computing the inverse of $A^{(t)}$ explicitly in $O(n^\omega)$ as in the pre-processing, we compute it by first computing $(T^{(0,t)})^{-1}$ and then multiplying $(T^{(0,t)})^{-1}(A^{(0)})^{-1}$.
Note that $T^{(0,t)}$ is of the form $I+C$ where $C$ has at most $n^{\varepsilon}$ nonzero columns, so its inverse $(T^{(0,t)})^{-1}$ can be computed explicitly in $O(n^{\omega(1,\varepsilon,\varepsilon)})$ operations (see \Cref{lem:Tinverse}). This inverse $(T^{(0,t)})^{-1}$ is of the same form $I+C'$, hence the multiplication of $(T^{(0,t)})^{-1}(A^{(0)})^{-1}$ costs only $O(n^{\omega(1,\varepsilon,1)})$ field operations (via \Cref{lem:complexityTproduct}) and the average update time becomes $O(u(k,n^{\varepsilon}) + kn^{-\varepsilon}(p + n^{\omega(1,1,\varepsilon)}))$, which for a fixed batch-size $k$ (i.e. all updates are of the same size) can be made worst-case via standard techniques (see Appendix \Cref{thm:standardTechnique}).

\paragraph{Queries}

When querying a submatrix $(A^{(t)})^{-1}_{I,J}$ we simply have to compute the product of the rows $I$ of $(T^{(0,t)})^{-1}$ and columns $J$ of $(A^{(0)})^{-1}$. To get the required rows of $(T^{(0,t)})^{-1} = (M^{(t)})^{-1}$ we need $O(q(n^{\delta_1},n^{\varepsilon}))$ time via algorithm $\mathcal{T}$. Because of the structure $(T^{(0,t)})^{-1} = \I+C$, where $C$ has only upto $n^\varepsilon$ nonzero columns, the product of the rows $I$ of $(T^{(0,t)})^{-1}$ and the columns $J$ of $(A^{(0)})^{-1}$ needs $O(n^{\omega(\delta_1,\varepsilon,\delta_2)})$ field operations (\Cref{lem:complexityTproduct}).
\end{proof}

\subsection{Explicit Transformation Maintenance}
\label{sub:Texplicit}

In the previous subsection we motivated that a \emph{dynamic matrix inverse} algorithm can be constructed from a \emph{transformation maintenance} algorithm.

The following algorithm allows us to quickly compute the inverse of a transformation matrix, if only a few are columns changed.
The algorithm is identical to \cite[Theorem 2]{Sankowski04} by Sankowski, for maintaining the inverse of any matrix.
Here we analyze the complexity of his algorithm for the setting that the algorithm is applied to a transformation matrix instead.

\begin{lemma}\label{lem:Texplicit}
Let $0 \le \varepsilon_0 \le \varepsilon_1 \le 1$ and let $T = \I + N$ be an $n \times n$ matrix where $N$ has at most $n^{\varepsilon_1}$ non-zero columns. Let $C$ be a matrix with at most $n^{\varepsilon_0}$ non-zero columns. If the inverse $T^{-1}$ is already known, then we can compute the inverse of $T' = T + C$ in $O(n^{\omega(1,\varepsilon_1,\varepsilon_0)})$ field operations.
\end{lemma}

For the special case $\varepsilon_1= 1$, $\varepsilon_0 = 0$ this result is identical to \cite[Theorem 1]{Sankowski04}, while for $\varepsilon_1 = 1$ this result is identical to \cite[Theorem 2]{Sankowski04}. For $\varepsilon_0 = 0, \varepsilon_1 < 1$ this result is implicitly proven inside the proof of \cite[Theorem 3]{Sankowski04}. Thus \Cref{lem:Texplicit} unifies half the results of \cite{Sankowski04}.

Note that for $\varepsilon_0 = 0$ the complexity simplifies to $O(n^{1+\varepsilon_1})$ field operations.

\begin{algorithm}
\caption{UpdateColumnsInverse (\Cref{lem:Texplicit})}\label{alg:UpdateColumnsInverse}
\begin{algorithmic}[1]
\REQUIRE $n \times n$ matrices $T^{-1}$ and $C$
\ENSURE  $(T+C)^{-1}$

\STATE $M \leftarrow \I+T^{-1}C$
\STATE $M^{-1} \leftarrow \textsc{Invert}(M)$ (\Cref{alg:invert}) \label{line:UpdateColumnsInverseInversion}
\RETURN $M^{-1}T^{-1}$

\end{algorithmic}
\end{algorithm}

\begin{proof}[Proof of \Cref{lem:Texplicit}]
The change from $T$ to $T+C$ can be expressed as some linear transformation $M$:
$$T+C = T (\underbrace{\I + T^{-1}C}_{=:M})$$
Here $C$ has at most $n^{\varepsilon_0}$ non-zero columns and $T^{-1}$ is of the form $\I+N$, where $N$ has at most $n^{\varepsilon_1}$ non-zero columns, so the matrix $M = \I + T^{-1}C$ can be computed in $O(n^{\omega(1,\varepsilon_1,\varepsilon_0)})$ field operations, see \Cref{lem:complexityTproduct}.

The new inverse $(T+C)^{-1}$ is given by $(TM)^{-1} = M^{-1}T^{-1}$. Note that $M$ is of form $\I + N$, where $N$ has $n^{\varepsilon_0}$ non-zero columns, so using \Cref{lem:Tinverse} (\Cref{alg:invert}) we can compute $M^{-1}$ in $O(n^{\omega(1,\varepsilon_1, \varepsilon_0)})$ field operations. Via \Cref{lem:Tinverse} we also know that $M^{-1}$ is again of the form $I+N$, where $N$ has at most $n^{\varepsilon_0}$ non-zero columns, thus the product $(T+C)^{-1} = M^{-1}T^{-1}$ requires $O(n^{\omega(1,\varepsilon_1, \varepsilon_0)})$ field operations (\Cref{lem:complexityTproduct}).
In total we require $O(n^{\omega(1,\varepsilon_1, \varepsilon_0)})$ operations.
\end{proof}

\subsection{Implicit Transformation Maintenance}
\label{sub:Timplicit}

In this section we will describe an algorithm for maintaining the inverse of a transformation matrix $T^{(t_1,t_2)}$ in an implicit form, that is, the entries of $(T^{(t_1,t_2)})^{-1}$ are not computed explicitly, but they can be queried.

We state this result in a more general way: Let $B^{(t)}$ be a matrix that receives column updates and where initially $B^{(0)} = \I$. Thus $B^{(t)}$ is a matrix that differs from $\I$ in only a few columns. As seen in equation \eqref{eq:TinverseImplicit} and \Cref{lem:Tinverse} such a matrix allows us to compute rows of its inverse $(B^{(t)})^{-1}$ without knowing the entire matrix $B^{(t)}$.
Thus we do not require the matrix $B^{(t)}$ to be given in an explicit way, instead it is enough to give the matrix $B^{(t)}$ via some pointer to a data-structure $D_B$. Our algorithm will then query this data-structure $D_B$ to obtain entries of $B^{(t)}$.

\begin{lemma}\label{lem:Timplicit}
Let $B^{(t)}$ be a matrix receiving column updates
where initially $B^{(0)} = \I$
and let $0 \le \varepsilon_0 \le \varepsilon_1 \le 1$.
Here $n^{\varepsilon_0}$ is an upper bound on the number of columns changed per update
and $n^{\varepsilon_1}$ is an upper bound on the number of columns where $B^{(t)}$ differs from the identity
(e.g. via restricting $t \le n^{\varepsilon_1-\varepsilon_0}$).
Assume matrix $B^{(t)}$ is given via some data-structure $D$ that supports the method $D\textsc{.Query}(I,J)$ to obtain any submatrix $B^{(t)}_{I,J}$.

Then there exists a transformation maintenance algorithm which maintains $(B^{(t)})^{-1}$ supporting the following operations:
\begin{itemize}
\item \emph{update($J^{(t)}$):} The set $J^{(t)} \subset [n]$ specifies the column indices where $B^{(t)}$ and $B^{(t-1)}$ differ.

The algorithm updates its internal data-structure using at most $O(n^{\omega(\varepsilon_1, \varepsilon_1, \varepsilon_0)})$ field operations.
To perform this update, the algorithm has to query $D$ to obtain two submatrices of $B^{(t)}$ of size $n^{\varepsilon_0}\times n^{\varepsilon_1}$ and $n^{\varepsilon_1}\times n^{\varepsilon_0}$.
\item \emph{query($I$)} The algorithm outputs the rows of $(B^{(t)})^{-1}$, specified by $I \subset [n]$, $|I| = n^{\delta}$ in $O(n^{\omega(\delta,\varepsilon_1,\varepsilon_1)})$ field operations.
To perform the query, the algorithm has to query $D$ to obtain a submatrix of $B^{(t)}$ of size $n^{\delta}\times n^{\varepsilon_1}$.
\end{itemize}
The algorithm requires no pre-processing.
\end{lemma}

While our algorithm of \Cref{lem:Timplicit} is new, in a restricted setting it has the same complexity as the transformation maintenance algorithm used in \cite[Theorem 4]{Sankowski04}.
When restricting to the setting where the matrix $B^{(t)}$ is given explicitly and no batch updates/queries are performed (i.e. $\varepsilon_0  = \delta = 0$), then the complexity of \Cref{lem:Timplicit} is the same as the transformation maintenance algorithm used in \cite[Theorem 4]{Sankowski04}.
\footnote{Our algorithm is slightly faster for the setting of batch updates and batch queries (i.e. more than one column is changed per update or more than one row is queried at once).
When considering batch updates and batch queries, Sankowski's variant of \Cref{lem:Timplicit} can be extended to have the complexity $\Omega(n^{\omega(\varepsilon_1, \varepsilon_0, \varepsilon_0)+\varepsilon_1 - \varepsilon_0})$ and $\Omega(n^{\omega(\varepsilon_1, \delta, \varepsilon_0)+\varepsilon_1 - \varepsilon_0})$ operations, because all internal computations are successive and can not be properly combined/batched using fast-matrix-multiplication.}

Before we prove \Cref{lem:Timplicit}, we will prove the following lemma, which is implied by \Cref{lem:Tinverse}.
This \namecref{lem:Timplicit} allows us to quickly invert matrices when the matrix is obtained from changing few rows \emph{and} columns.

\begin{algorithm}
\caption{UpdateInverse (\Cref{cor:updateRect})}\label{alg:UpdateInverse}
\begin{algorithmic}[1]
\REQUIRE $n \times n$ matrices $M,C,R$.
\ENSURE  $(M+C+R)^{-1}$

\STATE $M^{-1} \leftarrow \textsc{UpdateColumnsInverse}(M, M^{-1}, C)$ (\Cref{alg:UpdateColumnsInverse}) \label{line:UpdateInverseInversion1}
\STATE $M \leftarrow M+C$
\STATE $M^{-1} \leftarrow (\textsc{UpdateColumnsInverse}(M^\top , M^{\top-1},R^\top ))^\top $ (\Cref{alg:UpdateColumnsInverse}) \label{line:UpdateInverseInversion2}
\RETURN $M^{-1}$

\end{algorithmic}
\end{algorithm}

\begin{lemma}\label{cor:updateRect}
Let $0 \le \varepsilon_0 \le \varepsilon_1 \le 1$ and let $M, C, R$ be square matrices of size at most $n^{\varepsilon_1} \times n^{\varepsilon_1}$.

If $C$ has at most $n^{\varepsilon_0}$ non-zero columns, $R$ has at most $n^{\varepsilon_0}$ non-zero rows and we already know the inverse $M^{-1}$, then we can compute the inverse $(M+C+R)^{-1}$ in $O(n^{\omega(\varepsilon_1,\varepsilon_1,\varepsilon_0)})$ field operations.
\end{lemma}

\begin{proof}
We want to compute $(M+C+R)^{-1}$, where $R$ has at most $n^{\varepsilon_0}$ non-zero rows and $C$ has at most $n^{\varepsilon_0}$ non-zero columns.

Let $m = n^{\varepsilon_1}$, $\delta_1 = 1$ and $\delta_0 = \varepsilon_0/\varepsilon_1$, then $M, C, R$ are $m \times m$ matrices.

We can compute $(M+C)^{-1}$ via \Cref{lem:Texplicit} (\Cref{alg:UpdateColumnsInverse}) in $O(m^{\omega(1,\delta_1,\delta_0)}) = O(n^{\omega(\varepsilon_1,\varepsilon_1,\varepsilon_0)})$ operations.

Let $B = M+C$, then $(M+C+R)^\top  = B^\top +R^\top $ so $M^\top $ is obtained from $B^\top $ by changing at most $n^{\varepsilon_0}$ columns and we can use \Cref{lem:Texplicit} (\Cref{alg:UpdateColumnsInverse}) again to obtain $(M+C+R)^{-1} = ((B^\top +R^\top )^{-1})^\top $ using $O(n^{\omega(\varepsilon_1,\varepsilon_1,\varepsilon_0)})$ operations.
\end{proof}

With the help of \Cref{cor:updateRect} (\Cref{alg:UpdateInverse}), we can now prove \Cref{lem:Timplicit}. The high level idea is to see the matrix $B^{(t)}$ to be of the form $I+C$ similar to equation \eqref{eq:TinverseImplicit} in the overview (\Cref{sec:overview}), then we only maintain the $C_1^{-1}$ block during the updates. When performing queries, we then may have to compute some rows of the product $-C_2 C_1^{-1}$.

\begin{algorithm}
\caption{MaintainTransform (\Cref{lem:Timplicit})}\label{alg:MaintainTransform}
\begin{algorithmic}[1]
\REQUIRE Data-structure $D$ representing the matrix $B^{(t)}$ throughout all updates. We can call some function $D\textsc{.Query}(I,J)$ to receive $B^{(t)}_{I,J}$. In each update we also receive a set $J \subset [n]$ specifying the column indices where $B^{(t)}$ and $B^{(t-1)}$ differ.
\renewcommand{\algorithmicensure}{\textbf{Maintain:}}
\ENSURE  $(B^{(t)})^{-1}$ in an implicit form (rows can be queried).
Internally we maintain:
\begin{itemize}
\item $I^{(t)} = \bigcup_{i=1}^t J^{(i)} \subset [n]$, where $J^{(i)}$ is the set we received at the $i$th update.
\item An $n \times n$ matrix $M^{(t)}$ s.t. $M^{(t)}_{I^{(t)},I^{(t)}} = (B^{(t)})_{I^{(t)},I^{(t)}}$ and $M^{(t)}_{i,j} = \I_{i,j}$ for all other entries $(i,j)$
\item The inverse $(M^{(t)})^{-1}$.
\end{itemize}

\end{algorithmic}

\begin{algorithmic}[1]
\renewcommand{\algorithmicensure}{\textbf{initialization:}}
\ENSURE (We receive the data-structure $D$ )
\renewcommand{\algorithmicensure}{\textsc{Initialize($D$):}} 
\ENSURE
\STATE $t \leftarrow 0, M^{(0)} \leftarrow \I$, $J^{(0)} \leftarrow \emptyset$.
\STATE Remember $D$
\end{algorithmic}

\begin{algorithmic}[1]
\renewcommand{\algorithmicensure}{\textbf{update operation:}} 
\ENSURE (We receive $J \subset [n]$ )
\renewcommand{\algorithmicensure}{\textsc{Update($J$):}} 
\ENSURE
\STATE $t \leftarrow t+1$
\STATE $I^{(t)} \leftarrow I^{(t-1)} \cup J$
\STATE Update $M^{(t)}$ \COMMENT{This requires to query $B^{(t)}_{J,I^{(t)}}$ and $B^{(t)}_{I^{(t)},J}$ by calling $D\textsc{.Query}(J,I^{(t)})$ and $D\textsc{.Query}(I^{(t)},J)$}
\STATE \COMMENT{We will now compute two matrices $R,C$ s.t. $M^{(t)} = M^{(t-1)} +R+C$}
\STATE \hspace{10pt} $R, C \leftarrow 0$-matrices 
\STATE \hspace{10pt} $R_{J,I^{(t)}} \leftarrow M^{(t)}_{J,I^{(t)}} - M^{(t-1)}_{J,I^{(t)}}$
\STATE \hspace{10pt} $C_{I^{(t)} \setminus J, J} \leftarrow M^{(t)}_{I^{(t)} \setminus J, J} - M^{(t-1)}_{I^{(t)} \setminus J, J}$
\STATE $(M^{(t)})^{-1} \leftarrow \I$
\STATE $((M^{(t)})^{-1})_{I^{(t)},I^{(t)}} \leftarrow \textsc{updateInverse}($\\
\hspace{100pt} $M^{(t-1)}_{I^{(t)},I^{(t)}},$\\
\hspace{100pt} $((M^{(t-1)})^{-1})_{I^{(t)},I^{(t)}},$\\
\hspace{100pt} $R_{I^{(t)},I^{(t)}}, C_{I^{(t)},I^{(t)}}))$ (\Cref{alg:UpdateInverse}) \label{line:MaintainTransformInversion}\\
\COMMENT{Note that $((M^{(t)})^{-1})_{I^{(t)},I^{(t)}} = (M^{(t)}_{I^{(t)},I^{(t)}})^{-1} = (B^{(t)}_{I^{(t)},I^{(t)}})^{-1} = ((B^{(t)})^{-1})_{I^{(t)},I^{(t)}}$, because the matrices differ only in columns $I^{(t)}$ from the identity matrix, see \Cref{lem:Tinverse}.}
\end{algorithmic}

\begin{algorithmic}[1]
\renewcommand{\algorithmicensure}{\textbf{query operation:}}
\ENSURE (Querying some rows with index $J \subset [n]$ of $(B^{(t)})^{-1}$)
\renewcommand{\algorithmicensure}{\textsc{Query($J$):}} 
\ENSURE
\STATE $N \leftarrow \I$
\STATE Obtain $B^{(t)}_{J \backslash I^{(t)},I^{(t)}}$ by calling $D\textsc{.Query}(J \backslash I^{(t)},I^{(t)})$.
\STATE $N_{J \backslash I^{(t)},I^{(t)}} \leftarrow -B^{(t)}_{J \backslash I^{(t)},I^{(t)}}(M^{(t)})^{-1}_{I^{(t)},I^{(t)}}$
\STATE $N_{I^{(t)},I^{(t)}} \leftarrow (M^{(t)})^{-1}_{I^{(t)},I^{(t)}}$
\RETURN rows $J$ of $N$.

\end{algorithmic}
\renewcommand{\algorithmicensure}{\textbf{Output:}}
\end{algorithm}

\begin{proof}[Proof of \Cref{lem:Timplicit}]
Let $J^{(i)}$ be the set we received at the $i$th update. At time $t$ let $I^{(t)} = \bigcup_{i=1}^t J^{(i)}$ be the set of column indices of all so far changed columns, and let $M^{(t)}$ be the matrix s.t. $M^{(t)}_{I^{(t)},I^{(t)}} = (B^{(t)})_{I^{(t)},I^{(t)}}$ and $M^{(t)}_{i,j} = \I_{i,j}$ otherwise. We will maintain $I^{(t)}$, $M^{(t)}$ and $(M^{(t)})^{-1}$ explicitly throughout all updates.

For $t=0$ we have $B^{(0)} = \I$ and $I^{(0)} = \emptyset$, $M^{(0)} = \I = (M^{(0)})^{-1}$, so no pre-processing is required.

\paragraph{Updating $I^{(t)}$ and $M^{(t)}$}
When $B^{(t)}$ is "updated" (i.e. we receive a new set $J$), we set $I^{(t)} = I^{(t-1)} \cup J$.
As $J$ specifies the columns in which $B^{(t)}$ differs to $B^{(t-1)}$, we query $D$ to obtain the entries $B^{(t)}_{I^{(t)},J}$ and $B^{(t)}_{J,I^{(t)}}$ and update these entries in $M^{(t)}$ accordingly. Thus we now have $M^{(t)}_{I^{(t)},I^{(t)}} = B^{(t)}_{I^{(t)},I^{(t)}}$.

The size of the queried submatrices $B^{(t)}$ is at most $n^{\varepsilon_0} \times n^{\varepsilon_1}$ and $n^{\varepsilon_1} \times n^{\varepsilon_0}$, because by assumption at most $n^{\varepsilon_1}$ columns are changed in total (so $|I^{(t)}| \le n^{\varepsilon_1}$) and at most $n^{\varepsilon_0}$ columns are changed per update (so $|J| \le n^{\varepsilon_0}$).

\paragraph{Updating $(M^{(t)})^{-1}$}

Next, we have to compute $(M^{(t)})^{-1}$ from $(M^{(t-1)})^{-1}$. Note that the matrix $M^{(t)}$ is equal to the identity except for the submatrix $M^{(t)}_{I^{(t)}, I^{(t)}}$, i.e. without loss of generality (after reordering rows/columns) $M^{(t)}$ and its inverse look like this:
\begin{align*}
M^{(t)} = \begin{pmatrix}
\I & 0 \\
0 & M^{(t)}_{I^{(t)},I^{(t)}}
\end{pmatrix}
\hspace{30pt}
(M^{(t)})^{-1} = \begin{pmatrix}
\I & 0 \\
0 & (M^{(t)}_{I^{(t)},I^{(t)}})^{-1}
\end{pmatrix}
\end{align*}
So we have $(M^{(t)}_{I^{(t)},I^{(t)}})^{-1} = ((M^{(t)})^{-1})_{I^{(t)},I^{(t)}}$, and most importantly $M^{(t)}_{I^{(t)},I^{(t)}}$ is obtained from $M^{(t-1)}_{I^{(t)},I^{(t)}}$ by changing upto $n^{\varepsilon_0}$ rows and columns. We already know the inverse $(M^{(t-1)}_{I^{(t)},I^{(t)}})^{-1} = ((M^{(t-1)})^{-1})_{I^{(t)},I^{(t)}}$, hence we can compute $(M^{(t)})^{-1}$ via \Cref{cor:updateRect} (\Cref{alg:UpdateInverse}) using $O(n^{\omega(\varepsilon_1,\varepsilon_1,\varepsilon_0)})$ operations.

This concludes all performed computations during an update. The total cost is $O(n^{\omega(\varepsilon_1,\varepsilon_1,\varepsilon_0)})$ field operations.

\paragraph{Queries}

Next we will explain the query routine, when trying to query rows with index $J \subset [n]$ of the inverse. Remember that $B^{(t)}$ is of the form of \Cref{lem:Tinverse}, i.e. $B^{(t)} = \I + C$, where the non-zero columns of $C$ have their indices in $I^{(t)}$. Thus we have $((B^{(t)})^{-1})_{I^{(t)},I^{(t)}} = ((B^{(t)}_{I^{(t)},I^{(t)}})^{-1})$ and $((B^{(t)})^{-1})_{[n] \backslash I^{(t)},I^{(t)}} = -B^{(t)}_{[n] \backslash I^{(t)},I^{(t)}}(B^{(t)}_{I^{(t)},I^{(t)}})^{-1}$.

This means by setting some matrix $N = \I$ except for the submatrix $N_{J,I^{(t)}}$, where $N_{J\backslash I^{(t)},I^{(t)}} := -B^{(t)}_{J\backslash I^{(t)},I^{(t)}} (B^{(t)}_{I^{(t)},I^{(t)}})^{-1} = -B^{(t)}_{J\backslash I^{(t)},I^{(t)}} (M^{(t)})^{-1}_{I^{(t)},I^{(t)}}$ and $N_{J\cup I^{(t)}, I^{(t)}} := (M^{(t)})^{-1}_{J \cup I^{(t)}, I^{(t)}}$, then rows $J$ of $N$ and rows $J$ of $(B^{(t)})^{-1}$ are identical, so we can simply return these rows of $N$.

The query complexity is as follows:

The required submatrix $B^{(t)}_{I \backslash I^{(t)},I^{(t)}}$ is queried via $D_{B^{(t)}}^{(t)}$ and is of size at most $n^{\delta} \times n^{\varepsilon_1}$. The product $-B^{(t)}_{J,I^{(t)}} (M^{(t)})^{-1}_{I^{(t)},I^{(t)}}$ requires $O(n^{\omega(\delta,\varepsilon_1,\varepsilon_2)})$ field operations via \Cref{lem:complexityTproduct}.
\end{proof}

\subsection{Combining the Transformation Maintenance Algorithms}
\label{sub:combineT}

The task of maintaining the transformation matrix can itself be interpreted as a dynamic matrix inverse algorithm, where updates change columns of some matrix $T^{(0,t)}$ and queries return rows of $(T^{(0,t)})^{-1}$. This means the trick of maintaining $(A^{(t)})^{-1} = (T^{(t', t)})^{-1} (A^{(t')})^{-1}$ for $t \ge t'$ can also be used to maintain $(T^{(0,t)})^{-1}$ in the form $(T^{(0,t)})^{-1} = (T^{(t',t)})^{-1} (T^{(0,t')})^{-1}$ instead.

This is the high-level idea of how we obtain the following \Cref{lem:combineT} via \Cref{lem:Texplicit} and \Cref{lem:Timplicit}. We use \Cref{lem:Texplicit} to maintain $(T^{(0,t')})^{-1}$ and \Cref{lem:Timplicit} to maintain $(T^{(t',t)})^{-1}$.

We will state the new algorithm as maintaining the inverse of some matrix $B^{(t)}$ where $B^{(t)}$ receives column updates. Note that the following result is slightly more general than maintaining the inverse of some $T^{(0,t)}$ as we do \emph{not} require $B^{(t)} = \I$.

\begin{lemma}\label{lem:combineT}
Let $0 \le \varepsilon_0 \le \varepsilon_1 \le \varepsilon_2 \le 1$ and $k = n^{\varepsilon_0}$.

There exists a transformation maintenance algorithm that maintains the inverse of $B^{(t)}$, supporting column updates to $B^{(t)}$ and submatrix queries to the inverse $(B^{(t)})^{-1}$. Assume that throughout the future updates the form of $B^{(t)}$ is $\I + C^{(t)}$, where $C^{(t)}$ has always at most $n^{\varepsilon_2}$ nonzero columns (e.g. by restricting the number of updates $t \le n^{\varepsilon_2-\varepsilon_0}$). The complexities are:
\begin{itemize}
\item \emph{update($j_1,...,j_k,c_1,...,c_k$)}: Set the columns $j_l$ of $B^{(t)}$ to be $c_l$ for $l=1,...,k$ in $O(n^{\omega(\varepsilon_2, \varepsilon_1, \varepsilon_0)} + n^{\omega(1,\varepsilon_2, \varepsilon_1) - \varepsilon_1+ \varepsilon_0})$ field operations.
\item \emph{query($I$,$J$)}: Output the submatrix $(B^{(t)})^{-1}_{I,J}$ where $I,J \subset [n]$, $|I| = n^{\delta_1}$, $|J|=n^{\delta_2}$ in $O(n^{\omega(\delta_1,\varepsilon_2,\varepsilon_1)}+n^{\omega(\delta_1,\min\{\varepsilon_2, \delta_2\},\varepsilon_1)})$ field operations.
\end{itemize}
The pre-processing requires at most $O(n^\omega)$ operations, though if $B^{(0)} = \I$ the algorithm requires no pre-processing.

\end{lemma}

Before proving \Cref{lem:combineT} we want to point out that both \Cref{thm:columnUpdate,thm:elementUpdate} are direct implications of \Cref{lem:combineT}:

\begin{proof}[Proof of \Cref{thm:elementUpdate} and \Cref{thm:columnUpdate}]

The column update algorithm from \Cref{thm:columnUpdate} is obtained by letting $\varepsilon_0 = 0$, $\varepsilon_1 = \varepsilon$, $\varepsilon_2 = \delta_2 = 1$ and $\delta_1 = 0$ in \Cref{lem:combineT}.

\Cref{thm:elementUpdate} is obtained by combining \Cref{thm:TimpliesInverseElement} and \Cref{lem:combineT}:
\Cref{thm:TimpliesInverseElement} explains how a transformation maintenance algorithm can be used to obtain an element update dynamic matrix inverse algorithm and we use the algorithm from \Cref{lem:combineT} as the transformation maintenance algorithm.

To summarize \Cref{thm:TimpliesInverseElement}, it says that:
Assume there exists an algorithm for maintaining $T^{-1}$ where $T=\I$ initially then $T$ receives $n^{\varepsilon_0}$ column changes per update such that $T$ stays of the form $\I+C$ where $C$ has at most $n^{\varepsilon_2}$ columns. If the update time is $u(n^{\varepsilon_0},n^{\varepsilon_2})$ and the query time (for querying $n^{\delta_1}$ rows) is $q(n^{\delta_1}, n^{\varepsilon_2})$,
then there exists an element update dynamic matrix inverse algorithm that supports changing $n^{\varepsilon_0}$ elements per update and update time $O(u(n^{\varepsilon_0},n^{\varepsilon_2}) + (n^{-\varepsilon_2+\varepsilon_0}) \cdot (p + n^{\omega(1,1,\varepsilon_2)}))$.

For $u(n^{\varepsilon_0},n^{\varepsilon_2}) = O(n^{\omega(\varepsilon_2, \varepsilon_1, \varepsilon_0)} + n^{\omega(1,\varepsilon_2, \varepsilon_1) - \varepsilon_1+ \varepsilon_0})$ and no pre-processing time $p$ as in \Cref{lem:combineT}, we obtain with $\varepsilon_0 = 0$ the update complexity of \Cref{thm:elementUpdate} $O(n^{\varepsilon_2 +\varepsilon_1} + n^{\omega(1,\varepsilon_1,\varepsilon_2)-\varepsilon_1} + n^{\omega(1,1,\varepsilon_2)-\varepsilon_2})$.

The query time of \Cref{thm:TimpliesInverseElement} for querying an element of $T^{-1}$ is with $q(n^{\delta_1}, n^{\varepsilon_2}) = O(n^{\omega(\delta_1,\varepsilon_2,\varepsilon_1)})$ given via $O(q(1,n^{\varepsilon_2}) + n^{\omega(0,\varepsilon_2,0)}) = O(n^{\varepsilon_2+
\varepsilon_1})$.

\end{proof}

Next, we will prove \Cref{lem:combineT}.

\begin{proof}[Proof of \Cref{lem:combineT}]

Let $B^{(t)}$ be the matrix at round $t$, i.e. $B^{(0)}$ is what the matrix looks like at the time of the initialization/pre-processing. As pre-processing we compute $(B^{(0)})^{-1}$, which can be done in $O(n^\omega)$ operations, though for $B^{(0)} = \I$ this can be skipped since $(B^{(0)})^{-1} = \I$.

We implicitly maintain $B^{(t)}$ by maintaining another matrix $T^{(t',t)}$ such that $B^{(t)} = B^{(t')} T^{(t',t)}$ for some $t' \le t$, so $(B^{(t)})^{-1} = (T^{(t',t)})^{-1} (B^{(t')})^{-1}$. The matrix $(B^{(t')})^{-1}$ is maintained via \Cref{lem:Texplicit} while $(T^{(t',t)})^{-1}$ is maintained via \Cref{lem:Timplicit}. After a total of $n^{\varepsilon_1}$ columns were changed (e.g. when $t$ is a multiple of $n^{\varepsilon_1-\varepsilon_0}$), we set $t' = t$, which means $B^{(t')}$ receives an update that changes upto $n^{\varepsilon_1}$ columns. Additionally the matrix $T^{(t',t)}$ is reset to be the identity matrix and the algorithm from \Cref{lem:Timplicit} is reset as well.

\paragraph{Maintaining $(B^{(t')})^{-1}$}

The matrix $(B^{(t')})^{-1}$ is maintained in an explicit form via \Cref{lem:Texplicit}, which requires $O(n^{\omega(1,\varepsilon_2,\varepsilon_1)})$ operations. As this happens every $n^{-\varepsilon_1+\varepsilon_0}$ rounds, the cost for this is $O(n^{\omega(1,\varepsilon_2,\varepsilon_1)-\varepsilon_1+\varepsilon_0})$ operations on average per update. (This can be made worst case via Appendix \Cref{thm:standardTechnique}.)

\paragraph{Maintaining $(T^{(t',t)})^{-1}$}

We now explain how $(T^{(t',t)})^{-1}$ is maintained via \Cref{lem:Timplicit}. We have $B^{(t)} = B^{(t')} T^{(t',t)}$ which means the matrix $T^{(t',t)}$ is of the following form (this can be seen by multiplying both sides with $B^{(t')}$):
\begin{align*}
T^{(t',t)}
&= \I + (B^{(t')})^{-1}(B^{(t)}-B^{(t')})
\end{align*}
We do not want to compute this product explicitly, instead we construct a simple data-structure $D$ (\Cref{alg:ProductDataStructure}) to represent $T^{(t',t)} = \I + (B^{(t')})^{-1}(B^{(t)}-B^{(t')})$. (Note that in the algorithmic description \Cref{alg:CombinedTransformation} the matrix $S^{(t)} = B^{(t)}-B^{(t')}$.)
This data-structure allows queries to submatrices of $T^{(t',t)}$, by computing a small matrix product.
More accurately, for any set $I,J \subset [n]$ calling $D\textsc{.Query}(I,J)$ to obtain $T^{(t',t)}_{I,J}$ requires to compute the product $((B^{(t')})^{-1})_{I,[n]}(B^{(t)}-B^{(t')})_{[n],J}$.
Since $B^{(t')}$ (and thus also $(B^{(t')})^{-1})$, see \Cref{lem:Tinverse}) is promised to be of the form $I+C$, where $C$ has at most $n^{\varepsilon_2}$ non-zero columns, querying this new data-structure for $|I|=n^a,|J|=n^b$ requires $O(n^{\omega(a,\varepsilon_2,b)})$ field operations for any $0 \le a,b \le 1$.

When applying \Cref{lem:Timplicit} to maintain $(T^{(t',t)})^{-1}$, the update complexity is bounded by $O(n^{\omega(\varepsilon_1,\varepsilon_1,\varepsilon_0)}+n^{\omega(\varepsilon_0,\varepsilon_2,\varepsilon_1)}+n^{\omega(\varepsilon_1,\varepsilon_2,\varepsilon_0)}) = O(n^{\omega(\varepsilon_2,\varepsilon_1,\varepsilon_0)})$.

The average update complexity for updating both $(T^{(t',t)})^{-1}$ and $(B^{(t')})^{-1}$ thus becomes $O(n^{\omega(\varepsilon_2,\varepsilon_1,\varepsilon_0)} + n^{\omega(1,\varepsilon_2,\varepsilon_1)-\varepsilon_1+\varepsilon_0})$.

\paragraph{Queries}

Next, we will analyze the complexity of querying a submatrix $(B^{(t)})^{-1}_{I,J}$. To query such a submatrix, we need to multiply the rows $I$ of $(T^{(t',t)})^{-1}$ with the columns $J$ of $(B^{(t')})^{-1}$. Querying the rows of $(T^{(t',t)})^{-1}$ requires at most $O(n^{\omega(\delta,\varepsilon_2,\varepsilon_1)})$ operations according to \Cref{lem:Timplicit} (querying entries of $T^{(t',t)}$ via data-structure $D$ is the bottleneck). Note that $(B^{(t')})^{-1}$ is of the form $\I+C$ where $C$ has at most $n^{\varepsilon_2}$ nonzero columns, so for $|J| = n^{\delta_2}$ multiplying the rows and columns requires at most $O(n^{\omega(\delta_1,\varepsilon_1,\min\{\delta_2,\varepsilon_2\})})$ operations (see \Cref{lem:complexityTproduct}).

\end{proof}

\begin{algorithm}[H]
\caption{ColumnUpdateRowQuery (\Cref{lem:combineT})}\label{alg:CombinedTransformation}
\begin{algorithmic}[1]
\REQUIRE An $n \times n$ matrix $B^{(0)} = \I+C^{(0)}$ and inputs of the form $B^{(t)} = T^{(t-1)}+C^{(t)}$ where $C^{(t)}$ has nonzero columns $J^{(t)}$. 
\ENSURE  Maintain $(B^{(t)})^{-1}$ in an implicit form s.t. submatrices can be queried.

\end{algorithmic}

\begin{algorithmic}[1]
\renewcommand{\algorithmicensure}{\textbf{initialization:}}
\ENSURE
\renewcommand{\algorithmicensure}{\textsc{Initialize($B^{(0)}$)}}
\ENSURE
\STATE Compute $(B^{(0)})^{-1}$ (or just set $(B^{(0)})^{-1} \leftarrow \I$ in case of $B^{(0)} = \I$) \label{line:CombinedTransformationInversion1}
\STATE $S^{(0)} \leftarrow $ zero-matrix
\STATE $t' \leftarrow 0$, $t \leftarrow 0$
\STATE $D\textsc{.Update}((B^{(0)})^{-1}, S^{(0)})$ (Initialize data-structure $D$ \Cref{alg:ProductDataStructure})
\STATE \textsc{MaintainTransform.Initialize}$(D)$ (Initialize \Cref{alg:MaintainTransform} for $T^{(t',t)} = \I$)
\end{algorithmic}

\begin{algorithmic}[1]
\renewcommand{\algorithmicensure}{\textbf{update operation:}}
\ENSURE
\renewcommand{\algorithmicensure}{\textsc{Update($C$)}}
\ENSURE
\STATE $t \leftarrow t+1$,
\STATE $S^{(t)} \leftarrow S^{(t-1)}+ C$, $J^{(t)} \leftarrow$ indices of non-zero columns of $C$.
\IF{$|\bigcup_{i=1}^{t} J^{(i)}| \ge n^\varepsilon$}
\STATE $(B^{(t)})^{-1} \leftarrow \textsc{updateColumnsInverse}(B^{(t')},(B^{(t')})^{-1}, S^{(t)})$ (\Cref{alg:UpdateColumnsInverse}) \label{line:CombinedTransformationInversion2}
\STATE $t' \leftarrow t$
\STATE $S^{(t)} \leftarrow$ zero-matrix
\STATE $D\textsc{.Update}((B^{(t)})^{-1}, S^{(t)})$ (\Cref{alg:ProductDataStructure})
\STATE \textsc{MaintainTransform.Initialize}$(D)$ (Reinitialize \textsc{MaintainTransform} \Cref{alg:MaintainTransform} for $T^{(t',t)} = \I$)
\ELSE 
\STATE $D.\textsc{Update}((B^{(t')})^{-1}, S^{(t)})$ (\Cref{alg:ProductDataStructure})
\STATE \textsc{MaintainTransform.Update}$(J^{(t)})$ (\Cref{alg:MaintainTransform}). \label{line:CombinedTransformationInversion3}
\ENDIF
\end{algorithmic}

\begin{algorithmic}[1]
\renewcommand{\algorithmicensure}{\textbf{query operation:}}
\ENSURE (Querying some submatrix $(B^{(t)})^{-1})_{I,J}$)
\renewcommand{\algorithmicensure}{\textsc{Query($I,J$)}}
\ENSURE
\STATE Obtain rows $I$ of $(T^{(t',t)})^{-1}$ by calling\\
\textsc{MaintainTransform.Query}$(I, [n])$ (\Cref{alg:MaintainTransform}).
\RETURN $((T^{(t',t)})^{-1})_{I,[n]}((B^{(t')})^{-1})_{[n],J}$
\end{algorithmic}

\end{algorithm}

\begin{algorithm}
\caption{ProductDataStructure (\Cref{lem:combineT})}\label{alg:ProductDataStructure} (Used inside \Cref{alg:MaintainTransform} and \Cref{alg:CombinedTransformation})
\begin{algorithmic}[1]
\REQUIRE Two $n \times n$ matrices $A$ and $B$ given via pointers.
\ENSURE  Maintain  $P := \I + AB$ in an implicit form s.t. submatrices can be queried.
\end{algorithmic}

\begin{algorithmic}[1]
\renewcommand{\algorithmicensure}{\textbf{Update operation:}}
\ENSURE Called if matrix $A$ or $B$ change.
\renewcommand{\algorithmicensure}{\textsc{Update}$(A,B)$:}
\ENSURE
\STATE Remember the pointers to matrices $A$ and $B$.
\end{algorithmic}

\begin{algorithmic}[1]
\renewcommand{\algorithmicensure}{\textbf{Query operation:} Returns the submatrix $P_{I,J}$ for $I,J \subset [n]$.}
\ENSURE
\renewcommand{\algorithmicensure}{\textsc{Query}$(I,J)$:}
\ENSURE
\STATE $N \leftarrow \I$
\STATE $N_{I,J} \leftarrow N_{I,J} + A_{I,[n]} B_{[n], J}$
\RETURN $N_{I,J}$
\end{algorithmic}
\end{algorithm}

\begin{algorithm}
\caption{ElementUpdate (\Cref{thm:elementUpdate,thm:TimpliesInverseElement})}\label{alg:ElementUpdate}
\begin{algorithmic}[1]
\REQUIRE An $n \times n$ matrix $A^{(0)}$ and inputs of the form $A^{(t)} = A^{(t-1)}+C^{(t)}$ where $C^{(t)}$ has nonzero entries at $(i^{(t)}_1,j^{(t)}_1)...(i^{(t)}_k,j^{(t)}_k)$.
\ENSURE  Maintain $(A^{(t)})^{-1}$ in an implicit form s.t. submatrices can be queried.

\end{algorithmic}

\begin{algorithmic}[1]
\renewcommand{\algorithmicensure}{\textbf{initialization:}}
\ENSURE
\renewcommand{\algorithmicensure}{\textsc{Initialize}$(A^{(0)})$}
\ENSURE
\STATE Compute $(A^{(0)})^{-1}$ \label{line:ElementUpdateInversion1}
\STATE $S^{(0)} \leftarrow$ zero-matrix
\STATE \textsc{CombinedTransformation.Initialize$(\I)$} (Initialize \Cref{alg:CombinedTransformation} for $B^{(0)} = \I$) \label{line:ElementUpdateInversion2}
\STATE $t \leftarrow 0$
\end{algorithmic}

\begin{algorithmic}[1]
\renewcommand{\algorithmicensure}{\textbf{update operation:}}
\ENSURE
\renewcommand{\algorithmicensure}{\textsc{Update}$(C)$}
\ENSURE
\STATE $t \leftarrow t+1$
\STATE $S^{(t)} \leftarrow S^{(t-1)}+ C$, $J^{(t)} \leftarrow$ indices of the non-zero columns of $C$.
\IF{$|\bigcup_{i=1}^{t} J^{(i)}| \ge n^\varepsilon$}
\STATE $(A^{(0)})^{-1} \leftarrow \textsc{updateColumnsInverse}(A^{(0)},(A^{(0)})^{-1}, S^{(t)})$
\STATE $S^{(0)} \leftarrow$ zero-matrix
\STATE $t \leftarrow 0$
\STATE \textsc{ColumnUpdateRowQuery.Initialize$(\I$)} (Reinitialize \Cref{alg:CombinedTransformation} where $B^{(0)} := T^{(0,0)} = \I$) \label{line:ElementUpdateInversion3}
\ELSE 
\STATE $\tilde{C} \leftarrow (A^{(0)})^{-1} C$ (i.e. we select some columns of $(A^{(0)})^{-1}$)
\STATE \textsc{ColumnUpdateRowQuery.Update$(\tilde{C})$} (Update \Cref{alg:CombinedTransformation} where $B^{(t)} := T^{(0,t)}$) \label{line:ElementUpdateInversion4}
\ENDIF
\end{algorithmic}

\begin{algorithmic}[1]
\renewcommand{\algorithmicensure}{\textbf{query operation:}}
\ENSURE (Querying some submatrix $(A^{(t)})^{-1})_{I,J}$)
\renewcommand{\algorithmicensure}{\textsc{Query}$(I,J)$}
\ENSURE
\STATE Query rows $I$ of $(T^{(0,t)})^{-1}$ by calling\\
\textsc{ColumnUpdateRowQuery.Query$(I, [n])$} (\Cref{alg:CombinedTransformation})
\RETURN $((T^{(0,t)})^{-1})_{I,[n]}((A^{(0)})^{-1})_{[n],J}$

\end{algorithmic}
\end{algorithm}

\subsection{Applications}

There is a wide range of applications, which we summarized in \Cref{tbl:applications,tbl:applications2}.
The reductions are moved to \Cref{sec:applications}, because some have already been stated before (e.g. \cite{Sankowski04,Sankowski05,Sankowski07,MulmuleyVV87}) while many others are just known reductions for static problems applied to the dynamic setting.

In this section we want to highlight the most interesting applications, for an extensive list of all applications we refer to \Cref{sec:applications}.

\paragraph{Algebraic black box reductions}

The dynamic matrix inverse algorithms from \cite{Sankowski04} can
also be used to maintain the determinant, adjoint or solution of a
linear system. However, these reductions are white box.
In the static setting we already know, that determinant, adjoint
and matrix inverse are equivalent and that we can solve a linear
system via matrix inversion.
However, not all static reductions can be translated to work in the
dynamic setting. For example the Baur-Strassen theorem \cite{BaurS83,Morgenstern85}
used to show the hardness of the determinant in the static setting
can not be used in the dynamic setting.
Likewise the typical reduction of linear system to matrix inversion
does not work in the dynamic setting either.
Usually one would solve $Ax = b$ by inverting $A$ and computing the product $A^{-1}b$.
However, in the dynamic setting the matrix $A^{-1}$ is not explicitly given,
one would first have to query all entries of the inverse.
Thus it is an interesting question, what the relationship of the
dynamic versions of matrix inverse, determinant, adjoint, linear system is.

Can any dynamic matrix inverse algorithm be used to maintain determinant,
adjoint, solution to a linear system, or was this a special property
of the algorithms in \cite{Sankowski04}?
Is the dynamic determinant easier in the dynamic setting, or is it
as hard as the dynamic matrix inverse problem?

In \Cref{sub:algebraicApplications} we are able to confirm the
equivalence: dynamic matrix inverse, adjoint, determinant and linear
system are all equivalent in the dynamic setting, i.e. there exist
black box reductions that result in the same update time.
This is also an interesting difference to the static setting, where
there is no reduction from matrix inverse, determinant etc. to
solving a linear system.

\paragraph{Results based on column updates}

For many dynamic graph problems (e.g. bipartite matching, triangle detection,
$st$-reachability) there exist $\Omega(n^2)$ lower bounds for dense
graphs, when we allow node updates \cite{HenzingerKNS15}.
Thanks to the new column update dynamic matrix inverse algorithm we
are able to achieve sub-$O(n^2)$ update times, even though we allow
(restricted) node updates.
For example the size of a maximum bipartite matching can be maintained
in $O(n^\slowExponent)$, if we restrict the node updates to be only
on the left or only on the right side. Likewise triangle detection and
$st$-reachability can be maintained in $O(n^\slowExponent)$, if we
restrict the node updates to change only outgoing edges.
Especially for the dynamic bipartite matching problem this is a very
interesting result, because often one side is fixed: Consider for
example the setting where users have to be matched with servers, then
the server infra-structure is rarely updated, but there are constantly
users that will login/logout.
Previously only for the incremental setting
(i.e. no user will logout) there existed (amortized) sub-$O(n^2)$
algorithms \cite{BosekLSZ14}.
The total time of \cite{BosekLSZ14} for $n$ node insertions is $O(\sqrt{n}m)$, 
so $O(m/\sqrt{n}) = O(n^{1.5})$ amortized update time for dense graphs.
In \Cref{sec:lookAhead} we improve this to $O(n^{\omega-1})$.

%% file: lowerbounds.tex
\section{Conditional Lower Bounds}
\label{sec:lowerBounds}

In this section we will formalize the current barrier for dynamic matrix algorithms. 
We obtain conditional lower bounds for the trade-off between update and query time for column update/row query dynamic matrix inverse, which is the main tool of all currently known element update/element query algorithms by using them as transformation maintenance algorithms, see \Cref{thm:TimpliesInverseElement}.
The lower bounds we obtain (\Cref{cor:columnUpdateLB}) are tight with our upper bounds when the query time is not larger than the update time.
We also obtain worst-case lower bounds for element update and element query (\Cref{cor:elementUpdateLB} and \Cref{cor:elementColumnLB}), which are tight with our result \Cref{thm:elementUpdate} and Sankowski's result \cite[Theorem 3]{Sankowski04}.
The lower bounds are formalized in terms of dynamic matrix products over the boolean semi-ring and thus they also give lower bounds for dynamic transitive closure and related graph problems.

The conditional problems and conjectures defined in this section 
should be understood as questions. The presented problems are a formalization of the current barriers and the trade-off between using fast-matrix multiplication to pre-compute lots of information vs using slower matrix-vector multiplication to compute only required information in an online fashion. Our conjectures ask: Is there a better third option?

We will start this lower bound section with a short discussion of past lower bound results.
Then we follow with three subsections, each giving tight bounds for a different type of dynamic matrix inverse algorithm.
In last subsection \ref{sub:lowerBoundDiscussion} we will discuss, why other popular conjectures for dynamic algorithms are not able to capture the current barrier for dynamic matrix inverse algorithms.

\paragraph{Previous lower bounds}

All known lower bounds for the dynamic matrix inverse are based on matrix-matrix or matrix-vector products.
In \cite{FrandsenHM01} an unconditional linear lower bound is proven in the restricted computational model of algebraic circuits (\emph{history dependent algebraic computation trees}) for the task of dynamically maintaining the product of two matrices supporting element updates and element queries.
Via a reduction similar to our \Cref{thm:matrixProductReduction}), the lower bound then also holds for the dynamic matrix inverse.
Using \Cref{thm:matrixProductReduction}, a similar conditional lower bound $\Omega(n^{1-\varepsilon})$ in the RAM-model for all constants $\varepsilon > 0$ can be obtained from the OMv conjecture \cite{HenzingerKNS15}.

We can also obtain a $\Omega(n^{2-\varepsilon})$ lower bound via OMv for dynamic matrix inverse with column updates and column queries and (when reducing from OuMv) for an algorithm supporting both column and row updates and only element queries (which then gives hardness to column+row update dynamic determinant via \Cref{thm:inverseDeterminantEquivalence}).

\input{columnUpdateRowQuery.tex}

\input{elementUpdateRowQuery.tex}

\input{elementUpdateElementQuery.tex}

\subsection{Discussion on Super-Linear Bounds for Dynamic Matrix Inverse}
\label{sub:lowerBoundDiscussion}

The high-level idea of our lower bounds can be summarized as \emph{precomputing everything} vs \emph{waiting what information is going to be required}, i.e. if we do not know which information is going to be required in the next phase, then we can either do nothing and compute a vector-matrix product or we can precompute all possibilities using fast-matrix multiplication. Both of these option are a bit slow, one one hand using many vector-matrix products is slower than using fast matrix-multiplication, on the other hand pre-computing everything will compute never needed information. The trade-off between these two options forms the barrier for all currently known techniques for the dynamic matrix inverse. Our conjectures ask, if there is some better third option available.

This nature of \emph{precomputing everything} vs \emph{waiting} can also be seen in P\v{a}tra\c{s}cu's multiphase problem \cite{Patrascu10}. Unfortunately the multiphase problem, like other popular problems for lower bounds such as OMv, triangle detection, orthogonal vectors, SETH or 3-orthogonal vectors, are all unable to give super-linear lower bounds for the dynamic matrix inverse.

\paragraph{Online matrix-vector \cite{HenzingerKNS15}} 
The OMv conjecture states, that given a boolean matrix $M$ and polynomial time pre-processing of that matrix, computing $n$ products $Mv_i$, $i=1,...,n$, requires $\Omega(n^{3-\varepsilon})$ time, if the algorithm has to output $Mv_i$ before receiving the next vector $v_{i+1}$.

We now explain why this conjecture can not give super-linear bounds for dynamic matrix inverse.
\begin{itemize}
\item[Element updates] In worst-case the $n$ vectors $(v_i)_{1\le i\le n}$ contains $\Theta(n^2)$ bits of information. An element update contains only $O(\text{polylog}(n))$ bits of information, unless we use some really large field, which would result in slow field operations. Thus for all $n$ vectors, we have to perform $\Omega(n^{2-\varepsilon})$ updates in total for every constant $\varepsilon > 0$, which means no super-linear lower bound for element updates is possible.

\item[Column updates] For the setting column update and column query, the OMv conjecture is able to give a $\Omega(n^{2-\varepsilon})$ lower bound for the dynamic matrix inverse and the related OuMv conjecture can give the same lower bound for column and row update, element query dynamic inverse (both lower bounds are a result via the reduction \Cref{thm:matrixProductReduction}). However, for column update/row query we again have the problem that we would have to perform $n$ updates (or one update and $n$ queries), yielding no super-linear lower bound, when using the reduction from \Cref{thm:matrixProductReduction}.
\end{itemize}

\paragraph{Multiphase problem \cite{Patrascu10}} In the multiphase problem we have three phases: First, we are given a $k \times n$ matrix $M$, then a vector $v$ and lastly we have to answer whether some $(Mv)_i$ is 0 or 1. From all other presented problems, this problem captures best the issue of \emph{online computation} vs \emph{precomputation}, however, the conjectured time ($\Omega(kn)$ when given $v$ or $\Omega(k)$ when given $i$) is not large enough for a reduction, because we have to perform $O(n^{1-\varepsilon})$ element updates for every constant $\varepsilon > 0$, just to insert the information of $v$, so we can not get super-linear lower bounds.

\paragraph{Static problems conjectured to require $\Omega(n^\omega)$ time (e.g. triangle detection)}

An intuitive approach to obtain lower bounds is to reduce some static problem\footnote{Static problem refer to problems that do not have several phases. For example triangle detection or APSP are typical static problems used for dynamic lower bounds. \cite{AbboudW14}} to a dynamic one.
These type of lower bounds have the issue, that they require a pre-processing time that is lower than the time required to solve the problem in the static way.
For example one could get a super-linear $n^{\omega-1}$ (amortized) lower bound for dynamic matrix inverse via $st$-reachability by reducing from triangle detection as in \cite{AbboudW14}. However, this lower bound only holds, if one assumes $o(n^{\omega})$ pre-processing time.

Assuming $o(n^{\omega})$ pre-processing has two problems:
\begin{itemize}
\item It does not rule out algorithms with fast update time that have $\Omega(n^\omega)$ pre-processing.
Our dynamic algorithms and the ones from \cite{Sankowski04,Sankowski07} are of this type. 
We are interested in understanding why these algorithm can not achieve linear update time, 
even though their pre-processing is larger than $\Omega(n^\omega)$.
\item For some problems the $o(n^\omega)$ pre-processing requirement will refute any non-trivial algorithms (see \Cref{sub:trivialLowerBounds}). For example any dynamic matrix determinant algorithm with $o(n^{\omega})$ pre-processing must have $\Omega(n^\omega)$ update time.
This is why algorithm with larger pre-processing are interesting. 
\end{itemize}

\paragraph{Static problem conjectured to require $\Omega(n^{\omega+\varepsilon})$}

In the previous paragraph we highlighted the problems of using static problems that are conjectured to take $\Omega(n^\omega)$ time. Here we want to discuss static problems that are conjectured to have a higher complexity.

\begin{itemize}
\item[APSP] Computing all-pairs-shortest-paths with polynomialy bounded edge weights (i.e. $n^c$) is conjectured to require $\Omega(n^{3-\varepsilon})$ for every constant $\varepsilon,c > 0$.
So far APSP seems unsuited for algebraic algorithms since these algorithms always incur a pseudo-polynomial dependency on the edge weights.
\item[BMM] The Boolean-Matrix-Multiplication conjecture forbids the use of fast matrix multiplication, so it can not be used to bound the complexity of algebraic algorithms.
\item[$k$-clique] It is conjectured that detecting a $k$-clique in a graph requires $\Omega(n^{k\omega/3})$. For $k=3$ the $k$-clique problem is triangle detection, which was covered in the previous paragraph. For $k > 3$ there is no known reduction to matrix inverse without increasing the dimension to $n^{k/3}$ in which case we have the same problem as in the previous paragraph.
\item[$k$-orthogonal] In the $k$-orthogonal vectors problem we are given $k$ sets $S^1,...,S^k$, each containing $n$ vectors of dimension $d = n^{o(1)}$. The task is to find a $k$-tuple $(i_1,...,i_k)$ such that $\sum_{j=1}^n S^1_{i_1,j} \cdot ... \cdot S^k_{i_k,j} = 0$. This is conjectured to require $\Omega(n^{k-\varepsilon})$ time for every constant $\varepsilon > 0$.
For $k > 2$ no reduction to dynamic matrix inverse is known, while for $k=2$ we again have the same issue as with multiphase and OMv: We require to perform too many updates, just to insert the sets $S^1,S^2$.
\end{itemize}

%% file: columnUpdateRowQuery.tex
\subsection{Column Update, Row Query}
\label{sub:columnUpdate}

In this subsection we will present a new conditional lower bound for the dynamic matrix inverse with column updates and row queries, based on the dynamic product of two matrices. The new problem for the column update setting can be seen as an extension of the OMv conjecture. Instead of having online vectors, a set of possible vectors is given first and then one vector is selected from this list. We call this problem \emph{\hintedOMv{}} as it is similar to the OMv problem when provided a hint for the vectors.

\begin{definition}[\hintedOMv{}]\label{def:hintedOMv}

Let the computations be performed over the boolean semi-ring and let $t = n^\tau$, $0 < \tau < 1$. The \hintedOMv{} problem consists of the following phases:
\begin{enumerate}
\item Input an $n \times t$ matrix $M$ \label{phase:preOMV}
\item Input a $t \times n$ matrix $V$ \label{phase:batchUpdateOMV}
\item For an input index $i \in [n]$ output $MV_{[n],i}$ (i.e. multiply $M$ with the $i$th column of $V$). \label{phase:queryOMV}
\end{enumerate}
\end{definition}

The definition of the \hintedOMv{} problem is based on boolean matrix operations, so it can also be interpreted as a graph problem, i.e. the transitive closure problem displayed in \Cref{fig:transitiveClosure}. For this interpretation, the matrices $M$ and $V$ can be seen as a tripartite graph, where $M$ lists the directed edges between the first layer of $n$ nodes and the second layer of $n^\tau$ nodes. The matrix $V$ specifies the edges between the second layer and the third layer of $n$ nodes. All edges are oriented in the direction: first layer $\leftarrow$ second layer $\leftarrow$ third layer. The last phase of the \hintedOMv{} problem consists of queries, where we have to answer which nodes of the first layer can be reached by some node $i$ in the third layer, i.e. we perform a \emph{source query}.

\begin{figure}
\centering
\begin{tikzpicture}

\draw[fill=lightgray] (0.1,0) -- (0.1,4) -- (3-0.1,3) -- (3-0.1,1) -- (0.1,0); 
\node at (1.5,2) {M};
\draw[<-] (1,1.5) -- (2,1.5);

\draw[fill=lightgray] (6-0.1,0) -- (6-0.1,4) -- (3+0.1,3) -- (3+0.1,1) -- (6-0.1,0);
\node at (3+1.5,2) {V};
\draw[<-] (3+1,1.5) -- (3+2,1.5);

\node at (0,-0.5) {First layer};
\node at (0,-1) {$n$ nodes};
\foreach \y in {1,2,...,5}
    \draw [fill=white] (0,\y-1) circle (0.2);  
    
\node at (3,0) {Second layer};
\node at (3,-0.5) {$t$ nodes};
\foreach \y in {1,2,...,3}
    \draw [fill=white] (3,\y) circle (0.2);

\node at (6,-0.5) {Third layer};
\node at (6,-1) {$n$ nodes};
\foreach \y in {1,2,...,5}
    \draw [fill=white] (6,\y-1) circle (0.2);

\end{tikzpicture}

\caption{Graphical representation of the matrices $M$ and $V$.}
\label{fig:transitiveClosure}
\end{figure}
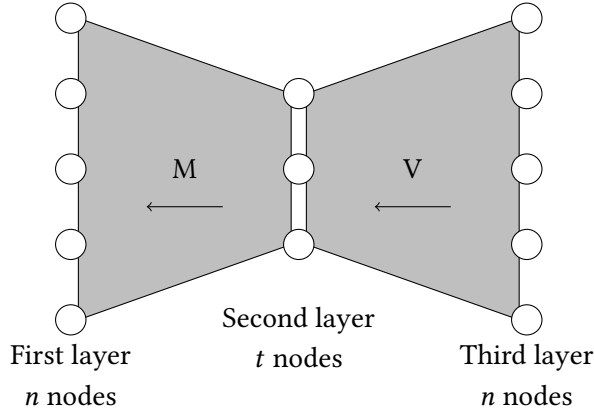

To motivate a lower bound, let us show two simple algorithms for solving the \hintedOMv{} problem:
\begin{itemize}
\item Precompute the product $MV$ in phase \ref{phase:batchUpdateOMV} using $O(n^{\omega(1,1,\tau)})$ operations, and output the $i$th column of the product in phase \ref{phase:queryOMV}.
\item Do not compute anything in phase \ref{phase:batchUpdateOMV} and compute $MV_{i,[n]}$ in phase \ref{phase:queryOMV} using a matrix-vector product in $O(n^{1+\tau})$ operations.
\end{itemize}

Currently no polynomially better way than these two options are known.\footnote{One can, however, improve the time requirement of phase \ref{phase:queryOMV} by a factor of $\log n$ using the technique from \cite{Williams07}, but no $O(n^{1+\tau-\varepsilon})$ algorithm is known for some constant $\varepsilon > 0$. For further discussion what previous results for the Mv- and OMv-problem imply for our conjectures/problems, we refer to \Cref{sub:implicationOfPreviousResults}.} We ask if there is another third option with a substantially different complexity and formalize this via the following conjecture:
We conjecture that the trivial algorithm is essentially optimal, i.e. we cannot do better than to decide between precomputing everything in phase \ref{phase:batchUpdateOMV} or to compute a matrix-vector product in phase \ref{phase:queryOMV}. The conjecture can be seen as formalizing the trade-off between \emph{pre-computing everything via fast matrix multiplication} vs \emph{computing only required information online via vector-matrix product}.

\begin{conjecture}[\hintedOMv{} conjecture]\label{con:hintedOMv}
Any algorithm solving \hintedOMv{} with polynomial pre-processing time in phase \ref{phase:preOMV} requires $\Omega(n^{\omega(1,1,\tau)-\varepsilon})$ operations for phases \ref{phase:batchUpdateOMV} or $\Omega(n^{1+\tau-\varepsilon})$ operations for phase \ref{phase:queryOMV} for all constant $\varepsilon > 0$.
\end{conjecture}

\begin{theorem}\label{thm:lowerBound2MatrixProd}
Assuming the \hintedOMv{} \Cref{con:hintedOMv}, the dynamic matrix-product with row updates and column queries requires $\Omega(n^{\omega(1,1,\tau)-\tau-\varepsilon})$ update time (worst-case), if the query time (worst-case) is $O(n^{1+\tau-\varepsilon})$ for some constant $\varepsilon > 0$.

The same lower bound holds for any column update, row query algorithm, as we can just maintain the transposed product.

For current $\omega$ when balancing update and query time, this implies lower bound of $\Omega(n^\slowExponentLB)$.
\end{theorem}

\begin{proof}
Assume there exist a dynamic matrix-product algorithm with update time $O(n^{\omega(1,1,\tau)-\tau-\varepsilon})$ and query time $O(n^{1+t-\varepsilon})$ for some $\varepsilon > 0$, then we can break \Cref{con:hintedOMv}.

We have to maintain the product $MV$, where $V$ is initially the zero matrix. We initialize the dynamic matrix product on the matrix $M$ in phase \ref{phase:preOMV}. In phase \ref{phase:batchUpdateOMV} we perform $t$ row updates to insert the values for $V$. When querying for some index $i$ in phase \ref{phase:queryOMV}, we perform a column query to the $i$th column of $MV$. The total cost for all updates is $O(n^{\omega(1,1,\tau)-\tau-\varepsilon})\cdot n^\tau = O(n^{\omega(1,1,\tau)-\varepsilon})$ and the cost for the queries is $O(n^{1+\tau-\varepsilon})$.
\end{proof}

Note that the lower bound from \Cref{thm:lowerBound2MatrixProd} allows for a trade-off between query and update time. The bound is tight with our upper bound from \Cref{thm:columnUpdate}, if the query time is not larger than the update time. We can also give a more direct lower bound for the dynamic matrix inverse, that captures the algebraic nature of the \hintedOMv{} problem.

By expressing the boolean matrix products as a graph as in \Cref{fig:transitiveClosure}, we obtain the following lower bound for transitive closure.

\begin{corollary}\label{thm:lowerBoundTransClosure}
Assuming the \hintedOMv{} \Cref{con:hintedOMv}, the dynamic transitive closure problem (and DAG-path counting and $k$-path for $k \ge 3$) with polynomial pre-processing time and node updates (restricted to updating only incoming edges) and query operations for obtaining the reachability of any source node, requires $\Omega(n^{\omega(1,1,\tau)-\tau-\varepsilon})$ update time (worst-case), if the query time (worst-case) is bounded by $O(n^{1+\tau-\varepsilon})$ for some constant $\varepsilon > 0$.

For current $\omega$ when balancing update and query time, this implies a lower bound of $\Omega(n^\slowExponentLB)$.
\end{corollary}

\begin{proof}
Assume there exists an algorithm for dynamic transitive closure with query time $O(n^{1+\tau-\varepsilon})$ but $O(n^{\omega(1,1,t)-\tau-\varepsilon})$ update time for some $\varepsilon > 0$. We can use this algorithm to refute the \hintedOMv{} \Cref{con:hintedOMv}.

We start with an empty 3-layered graph, where the first and third layer consist of $n$ nodes and the layer between them has $t$ nodes. During phase \ref{phase:preOMV} we initialize the dynamic transitive closure algorithm on the graph, where the edges going from layer two to layer one are as specified by matrix $M$. In phase \ref{phase:batchUpdateOMV} we perform $t=n^\tau$ updates to add the edges specified by $V$ between third and second later. In phase \ref{phase:queryOMV} we query which nodes in the first layer can be reached by the $i$-th node in the third layer. The total cost for \ref{phase:batchUpdateOMV} is $O(n^{\omega(1,1,t)-\tau-\varepsilon}) \cdot n^\tau = O(n^{\omega(1,1,t)-\varepsilon})$ and the cost for \ref{phase:queryOMV} is $O(n^{1+\tau-\varepsilon})$.
\end{proof}

\Cref{thm:lowerBound2MatrixProd} implies the same lower bound for column update/row query dynamic matrix inverse and adjoint via the reductions from \Cref{thm:matrixProductReduction} and \Cref{cor:inverseAdjointEquivalence}. 

\begin{corollary}\label{cor:columnUpdateLB}

Assuming the \hintedOMv{} \Cref{con:hintedOMv}, the dynamic matrix inverse (and dynamic adjoint) with column updates and row queries requires $\Omega(n^{\omega(1,1,\tau)-\tau-\varepsilon})$ update time (worst-case), if the query time (worst-case) is $O(n^{1+\tau-\varepsilon})$ for some constant $\varepsilon > 0$.

For current $\omega$ when balancing update and query time, this implies lower bound of $\Omega(n^\slowExponentLB)$.

\end{corollary}

%% file: elementUpdateRowQuery.tex
\subsection{Element Update, Row Query}
\label{sub:elementUpdateRowQuery}

Next, we want to define a problem which is very similar to the \hintedOMv{} problem, but allows for a lower bound for the weaker setting of element updates and row query dynamic matrix inverse.

First, remember the high-level idea of the \hintedOMv{} problem: We are given a matrix $M$ and a set of possible vectors (i.e. a matrix) $V$ and have to output only one matrix vector product $Mv$ for some $v$ in $V$, but since we don't know which vector is going to be chosen our only choices are pre-computing everything or waiting for the choice of $v$. When trying to extend this problem to element updates, then we obviously can not insert the matrix $V$ via element updates one by one, as that would cause a too high overhead in the reduction and thus a very low lower bound. So instead we will give $V$ already during the pre-processing, but the matrix $M$ is not fully known. Instead, the matrix $M$ is created from building blocks, which are selected by the element updates. Formally the problem is defined as follows:

\begin{definition}[\doubleHintedOMv{}]\label{def:doubleHintedOMv}

Let all operations be performed over the boolean semi-ring and let $t = n^{\tau}$ for $0 < \tau < 1$. The \emph{\doubleHintedOMv{}} problem consists of the following phases:
\begin{enumerate}
\item Input matrices $N \in R^{n \times n}, V \in R^{t \times n}$ \label{phase:pre}
\item Input $I \in [n]^t$. \label{phase:setV}
\item Input index $j \in [n]$ and output $N_{[n],I} V_{[t],j}$. \label{phase:queryVector}
\end{enumerate}
\end{definition}

This problem has three different interpretations. One is to consider this a variant of the \hintedOMv{} problem, but with two hints: one for $M$ and one for the vector $v$ (hence the name \doubleHintedOMv{}). First in phase \ref{phase:pre}, we are given a matrix $N$ and a matrix $V$ (i.e. a set of vectors) as a hint for $M$ and $v$. In phase \ref{phase:setV} the hint for $M$ concretized by constructing $M$ from columns of $N$.

Another interpretation for this problem is as some \emph{dynamic 3-matrix product} $NRV$, where $R$ is a rectangular $n \times t$ matrix. Here the phase \ref{phase:setV} can be seen as updates to the $R$ matrix, where for $I = (i_1,...,i_t)$ the entries $R_{i_j,j}$ are set to $1$.

The third interpretation for the problem is graph theoretic and considers the problem to be a dynamic transitive closure problem with edge updates and source query. This graphical representation is displayed in \Cref{fig:transitiveClosureEdgeSource}. We are given 4 groups of nodes, the first, second and fourth group are of size $n$ while the third group is only of size $t$. There exist directed edges from the second to the first group given by the non-zero entries of $N$ and the edges from the fourth to the third group are given by $V$. In phase \ref{phase:setV}, $t$ edges are inserted from the third to the second layer. In phase \ref{phase:queryVector} we have to answer which nodes in the first layer can be reached from source-node $j$ in the third group.

\begin{figure}
\centering
\begin{tikzpicture}

\draw[fill=lightgray] (4+0.1,0) -- (4+0.1,4) -- (6-0.1,4) -- (6-0.1,0) -- (4+0.1,0);
\node at (5,2) {$N$};
\draw[<-] (4.5,1.5) -- (5.5,1.5);

\draw[line cap=round,line width=1pt,dash pattern=on 0pt off 2\pgflinewidth,fill=white] (6+0.1,0) -- (6+0.1,4) -- (8-0.1,3) -- (8-0.1,1) -- (6+0.1,0);
\node at (7,2) {$R$ or $I$};
\draw[<-] (6.5,1.5) -- (7.5,1.5);

\draw[fill=lightgray] (10-0.1,0) -- (10-0.1,4) -- (8+0.1,3) -- (8+0.1,1) -- (10-0.1,0);
\node at (9,2) {$V$};
\draw[<-] (8.5,1.5) -- (9.5,1.5);
    
\node at (4,-0.5) {First layer};
\node at (4,-1) {$n$ nodes};
\foreach \y in {1,2,...,5}
    \draw [fill=white] (4,\y-1) circle (0.2);      

\node at (6,-0.5) {Second layer};
\node at (6,-1) {$n$ nodes};
\foreach \y in {1,2,...,5}
    \draw [fill=white] (6,\y-1) circle (0.2); 
    
\node at (8,0) {Third layer};
\node at (8,-0.5) {$t$ nodes};
\foreach \y in {1,2,...,3}
    \draw [fill=white] (8,\y) circle (0.2);

\node at (10,-0.5) {Fourth layer};
\node at (10,-1) {$n$ nodes};
\foreach \y in {1,2,...,5}
    \draw [fill=white] (10,\y-1) circle (0.2); 
    
\end{tikzpicture}

\caption{Graphical representation of the \doubleHintedOMv{} problem (\Cref{def:doubleHintedOMv}).}
\label{fig:transitiveClosureEdgeSource}
\end{figure}
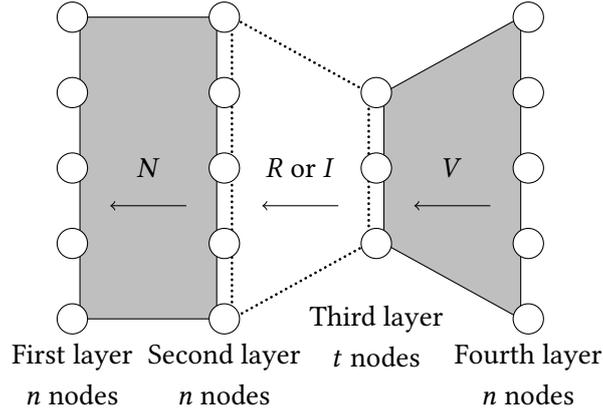

The \emph{\doubleHintedOMv{}} problem can be solved by the following trivial algorithm: Assuming polynomial pre-processing time, we do not know how to exploit the given information $N$ and $V$. We have no idea which entries of $N$ will be multiplied with which entries of $V$ and we can not try all exponentially many possible combinations for the vector $I$, so we do not know how to compute anything useful. For the next phases we have the following options:
\begin{itemize}
\item In phase \ref{phase:setV}, compute the product $N_{[n],I}V$ using $O(n^{\omega(1,\tau,1)})$ operations. In phase \ref{phase:queryVector} we simply output the $j$th column of that product.
\item We do not compute anything in phase \ref{phase:setV}, but remember the set $I$. In phase \ref{phase:queryVector} we compute $N_{[n],I}V_{[t],j}$ as a vector-matrix-vector product in $O(n^{1+\tau})$ time.
\end{itemize}
Again we ask, if there is a better third option than trade-off between \emph{pre-computing everything} vs \emph{waiting and computing only required information}. We formalize this question as the following conjecture: The two options are essentially optimal, meaning there is no better way than to pre-compute everything in phase \ref{phase:setV} or to wait and perform a matrix vector product in phase \ref{phase:queryVector}.

\begin{conjecture}\label{con:doubleHintedOMv}
Any algorithm solving \emph{\doubleHintedOMv{}} with polynomial pre-processing time in Phase \ref{phase:pre} satisfies one of the following:
\begin{itemize}
\item Phase \ref{phase:setV} requires $\Omega(n^{\omega(1,\tau,1)-\varepsilon})$.
\item Phase \ref{phase:queryVector} requires $\Omega(n^{1+\tau-\varepsilon})$.
\end{itemize}
For every $\varepsilon > 0$.
\end{conjecture}

Since the \doubleHintedOMv{} problem can be represented as a product of three matrices $NRV$, we obtain the following lower bound for dynamic 3-matrix product algorithms with element updates and row queries.

\begin{theorem}\label{thm:lowerBound3MatrixProd}
Assuming the \doubleHintedOMv{} \Cref{con:doubleHintedOMv}, any dynamic matrix product algorithm with polynomial pre-processing time, element updates and column queries requires $\Omega(n^{\omega(1,1,\tau)-\tau-\varepsilon})$ worst-case update time, if the worst-case query time is $O(n^{1+\tau-\varepsilon})$ for some constant $\varepsilon > 0$.

The same lower bound holds for any element update, row query algorithm, as we can just maintain the transposed product.

For current $\omega$ when balancing update and query time, this implies lower bound of $\Omega(n^\slowExponent)$.
\end{theorem}

\begin{proof}
Assume there exist a dynamic 3-matrix-product algorithm with update time $O(n^{\omega(1,1,\tau)-\tau-\varepsilon})$ and query time $O(n^{1+\tau-\varepsilon})$ for some $\varepsilon > 0$, then we can break \Cref{con:doubleHintedOMv}.

We have to maintain the product $NRV$, where $R$ is initially the zero matrices. Let $I = (i_1,...i_t)$, then in phase \ref{phase:setV} we perform $t$ element updates to set $R_{i_j, j}=1$ for $j=1,...,t$. We now have $N_{[n],I} = NR$, so in phase \ref{phase:queryVector} we simply query the column $j$ of $NRV$ to obtain $N_{[n],I}V_{[t],j}$.
\end{proof}

Note that \Cref{thm:lowerBound3MatrixProd} also implies a lower bound on element update and element query, as we could query $n$ elements to get an entire column of the product. The same lower bounds hold for dynamic matrix inverse and adjoint via the reduction from \Cref{thm:matrixProductReduction} and \Cref{cor:inverseAdjointEquivalence}.
Hence we get a lower bound of $\Omega(n^{\omega(1,\tau,1)-\tau-\varepsilon})$ per element update or $\Omega(n^{\tau-\varepsilon})$ per element query for every $\varepsilon > 0$, which is tight via Sankowski's result presented in \cite[Theorem 3]{Sankowski04}, when $n$ queries are not slower than one update.

\begin{corollary}\label{cor:elementColumnLB}

Assuming the \doubleHintedOMv{} \Cref{con:doubleHintedOMv}, any dynamic matrix inverse (or dynamic ajoint) algorithm with element updates and row queries requires $\Omega(n^{\omega(1,1,\tau)-\tau-\varepsilon})$ update time (worst-case), if the query time (worst-case) is $O(n^{1+\tau-\varepsilon})$ for some constant $\varepsilon > 0$.

For current $\omega$ when balancing update and query time, this implies lower bound of $\Omega(n^\slowExponentLB)$.

Additionally, any dynamic matrix inverse (or dynamic adjoint) algorithm with element updates and element queries requires $\Omega(n^{\omega(1,1,\tau)-\tau-\varepsilon})$ update time (worst-case), if the query time (worst-case) is $O(n^{\tau-\varepsilon})$ for some constant $\varepsilon > 0$.

\end{corollary}

The graph theoretic representation of the problem, yield the same lower bound for the dynamic transitive closure and DAG path counting problem with edge updates and source queries (and thus also edge updates and $n$ pair queries).

\begin{corollary}\label{cor:lowerBoundTransitiveClosureElementSource}
Assuming the \doubleHintedOMv{} \Cref{con:2hintedOuMv}, dynamic transitive-closure and DAG path counting with polynomial pre-processing time, edge updates and source queries, requires $\Omega(n^{\omega(1,1,\tau)-\tau-\varepsilon})$ worst-case update time, if the worst-case query time is $O(n^{1+\tau-\varepsilon})$ for some $\varepsilon > 0$.

For current $\omega$, this implies a lower bound of $\Omega(n^\slowExponentLB)$.

Additionally, any dynamic transitive-closure or DAG path counting algorithm with edge updates and pair queries requires $\Omega(n^{\omega(1,1,\tau)-\tau-\varepsilon})$ update time (worst-case), if the query time (worst-case) is $O(n^{\tau-\varepsilon})$ for some constant $\varepsilon > 0$.
\end{corollary}

Note that these lower bounds specify a trade-off between update and query time, i.e. we can query faster, if we are willing to pay a higher update time. This trade-off is tight unless the query time for querying a row exceeds the update time.

%% file: elementUpdateElementQuery.tex
\subsection{Element Update, Element Query}
\label{sub:elementUpdate}

The \doubleHintedOMv{} problem allowed us to specify a lower bound for dynamic matrix inverse algorithms with slow update and fast query time (i.e. \cite[Theorem 3]{Sankowski04}).
We also want to obtain a lower bound for the case that update and query time are balanced. The \doubleHintedOMv{} problem does not properly capture the hardness of single element queries, because the problem asks for an entire column to be queried.
When querying a column via $O(n)$ element queries, we would have some kind of look-ahead information, because after the first element query, the next $O(n)$ positions for the queries are known, since they are in the same column.
The next problem we define is a variation of the \doubleHintedOMv{} problem, where we try to fix this issue by querying only a single value.
The new problem can be considered a hinted variant of the OuMv problem \cite{HenzingerKNS15}, where we repeat the idea of restricting some matrix $N$ to obtain a matrix $M$ and giving hints for the vectors $u$ and $v$.

\begin{definition}[\doubleHintedOuMv{}]\label{def:2hintedOuMv}

Let all operations be performed over the boolean semi-ring and let $t_1 = n^{\tau_1}, t_2 = n^{\tau_2}$, $0 < \tau_1, \tau_2 < 1$. The \doubleHintedOuMv{} problem consists of the following phases:
\begin{enumerate}
\item Input matrices $U \in R^{n \times t_1}, N \in R^{n \times n}, V \in R^{t_2 \times n}$ \label{phase:preOuMv}
\item Input $I \in [n]^{t_1}$. \label{phase:setI}
\item Input $J \in [n]^{t_2}$. \label{phase:setJ}
\item Input indices $i,j \in [n]$ and output $(U N_{I,J} V)_{i,j}$. \label{phase:query}
\end{enumerate}
\end{definition}

This problem can be considered a dynamic 5-matrix product $UR_1NR_2V$, where $R_1$ and $R_2$ are rectangular matrices. Phase \ref{phase:setI} and \ref{phase:setJ} can be seen as updates to the $R_{1}$ and $R_2$ matrices. A more intuitive illustration for this problem is to see the problem as a vector-matrix-vector product $u^\top Mv$, where during the pre-processing we are given a hint what the vectors $u,v$ and the matrix $M$ could be. This interpretation of the problem is also the source for its name \doubleHintedOuMv{}. The two phases \ref{phase:setI} and \ref{phase:setJ} concretize the hint for $M$ by constructing $M$ from rows and columns of $N$ via $M = N_{I,J}$. (Note that $I$ and $J$ are vectors, not sets, so $N_{I,J}$ is not a typical submatrix, see submatrix notation in the preliminaries \Cref{sec:preliminaries}. Instead, rows and columns can be repeated and re-ordered.) During the last phase one row $u$ of $U$ and one column $v$ of $V$ are selected and the product $u^\top Mv$ has to be computed.

Similar to the \doubleHintedOMv{} problem (\Cref{def:hintedOMv}), we can specify the \doubleHintedOuMv{} problem as a transitive closure problem.
This graphical representation is displayed in \Cref{fig:transitiveClosureElement}. We are given 6 groups of nodes, the first group is of size $n$ and the second group is of size $t_1$. There exist directed edges from the second to the first group, specified by $U$. The third and fourth group are of size $n$ and have directed edges from the fourth to third group, specified by $N$. The fifth group is of size $t_2$, while the 6th group is of size $n$. There also exists directed edges from the sixth to the fifth group, specified by $V$. In phase \ref{phase:setI} each node in the second group gets a directed edge from a node in the third group. In phase \ref{phase:setJ} each node in the fifth group gets a directed edge to a node in the fourth group. In phase \ref{phase:query} we have to answer whether the $j$th node in the last group can reach the $i$th node in the first group.

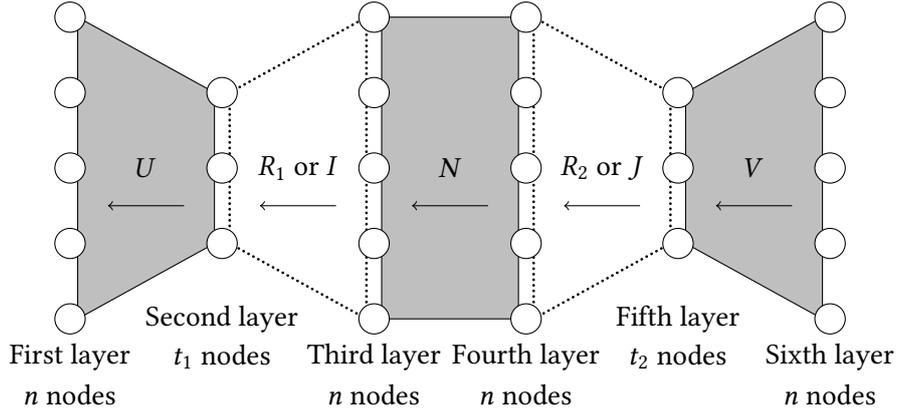
\begin{figure}
\centering
\begin{tikzpicture}

\draw[fill=lightgray] (0.1,0) -- (0.1,4) -- (2-0.1,3) -- (2-0.1,1) -- (0.1,0); 
\node at (1,2) {$U$};
\draw[<-] (0.5,1.5) -- (1.5,1.5);

\draw[line cap=round,line width=1pt,dash pattern=on 0pt off 2\pgflinewidth,fill=white] (4-0.1,0) -- (4-0.1,4) -- (2+0.1,3) -- (2+0.1,1) -- (4-0.1,0);
\node at (3,2) {$R_1$ or $I$};
\draw[<-] (2.5,1.5) -- (3.5,1.5);

\draw[fill=lightgray] (4+0.1,0) -- (4+0.1,4) -- (6-0.1,4) -- (6-0.1,0) -- (4+0.1,0);
\node at (5,2) {$N$};
\draw[<-] (4.5,1.5) -- (5.5,1.5);

\draw[line cap=round,line width=1pt,dash pattern=on 0pt off 2\pgflinewidth,fill=white] (6+0.1,0) -- (6+0.1,4) -- (8-0.1,3) -- (8-0.1,1) -- (6+0.1,0);
\node at (7,2) {$R_2$ or $J$};
\draw[<-] (6.5,1.5) -- (7.5,1.5);

\draw[fill=lightgray] (10-0.1,0) -- (10-0.1,4) -- (8+0.1,3) -- (8+0.1,1) -- (10-0.1,0);
\node at (9,2) {$V$};
\draw[<-] (8.5,1.5) -- (9.5,1.5);

\node at (0,-0.5) {First layer};
\node at (0,-1) {$n$ nodes};
\foreach \y in {1,2,...,5}
    \draw [fill=white] (0,\y-1) circle (0.2);  
    
\node at (2,0) {Second layer};
\node at (2,-0.5) {$t_1$ nodes};
\foreach \y in {1,2,...,3}
    \draw [fill=white] (2,\y) circle (0.2);
    
\node at (4,-0.5) {Third layer};
\node at (4,-1) {$n$ nodes};
\foreach \y in {1,2,...,5}
    \draw [fill=white] (4,\y-1) circle (0.2);      

\node at (6,-0.5) {Fourth layer};
\node at (6,-1) {$n$ nodes};
\foreach \y in {1,2,...,5}
    \draw [fill=white] (6,\y-1) circle (0.2); 
    
\node at (8,0) {Fifth layer};
\node at (8,-0.5) {$t_2$ nodes};
\foreach \y in {1,2,...,3}
    \draw [fill=white] (8,\y) circle (0.2);

\node at (10,-0.5) {Sixth layer};
\node at (10,-1) {$n$ nodes};
\foreach \y in {1,2,...,5}
    \draw [fill=white] (10,\y-1) circle (0.2); 
    
\end{tikzpicture}

\caption{Graphical representation of the \doubleHintedOuMv{} problem (\Cref{def:2hintedOuMv}).}
\label{fig:transitiveClosureElement}
\end{figure}

The \doubleHintedOuMv{} problem can be solved by the following trivial algorithm. Depending on the values for $t_1$ and $t_2$ we have the following three options:
\begin{itemize}
\item Compute the product $UN_{I,[n]}$ in phase \ref{phase:setI}, using $O(n^{\omega(1,\tau_1,1)})$ operations. In phase \ref{phase:query} we compute the product of the $i$th row of $UN_{I,[n]}$ and the $j$th column of $V$.
\item Compute the product $N_{I,J}V$ in phase \ref{phase:setJ}, using $O(n^{\omega(\tau_1,\tau_2,1)})$ operations.  In phase \ref{phase:query} we compute the product of the $i$th row of $U$ and the $j$th column of $M_{I,J}V$.
\item Compute the product $(UN_{I,J}V)_{i,j}$ as a vector-matrix-vector product in $O(n^{\tau_1 \tau_2})$
\end{itemize}
(Note that computing $UN_{I,J}$ needs as much time as computing $N_{I,J}V$, so this would be the same as the second variant.)
Again we ask, if there is a better option than pre-computing everything via fast-matrix multiplication or to wait which information is going to be required, or maybe there exists some clever pre-processing even though we do not know which entries of $N$ will be multiplied with which entries of $V$ or $U$.
The question is formalized via the conjecture that the three options of the trivial algorithm are essentially optimal.
So similar to the \hintedOMv{} and \doubleHintedOMv{} conjecture, the \doubleHintedOuMv{} conjecture can be seen as a trade-off between pre-computing everything vs waiting for the next phase and computing only required information, which forms the fundamental barrier for all currently known techniques for the dynamic matrix inverse algorithms.

\begin{conjecture}\label{con:2hintedOuMv}
Any algorithm solving \doubleHintedOuMv{} with polynomial pre-processing time in Phase \ref{phase:preOuMv} satisfies one of the following:
\begin{itemize}
\item Phase \ref{phase:setI} requires $\Omega(n^{\omega(1,\tau_1,1)-\varepsilon})$.
\item Phase \ref{phase:setJ} requires $\Omega(n^{\omega(\tau_2,\tau_1,1)-\varepsilon})$.
\item Phase \ref{phase:query} requires $\Omega(n^{\tau_1+\tau_2-\varepsilon})$.
\end{itemize}
For every $\varepsilon > 0$.
\end{conjecture}

Since the \doubleHintedOuMv{} problem can be represented as a 5-matrix product, we obtain the following lower bound for dynamic 5-matrix product algorithms with element updates and element queries.

\begin{theorem}\label{thm:lowerBound5MatrixProd}
Assuming the \doubleHintedOuMv{} \Cref{con:2hintedOuMv}, the dynamic 5-matrix-product with polynomial time pre-processing, element updates and element queries requires $\Omega(\min_{\tau_1,\tau_2}(n^{\tau_1+\tau_2}+n^{\omega(1,\tau_1,\tau_2)-\tau_2}+n^{\omega(1,1,\tau_1)-\tau_1})n^{-\varepsilon})$ worst-case time for all $\varepsilon > 0$ for updates or queries.

For current $\omega$, this implies a lower bound of $\Omega(n^\fastExponent)$.
\end{theorem}

\begin{proof}
Assume there exist a dynamic 5-matrix-product algorithm with worst-case update and query time $O((n^{\tau_1+\tau_2}+n^{\omega(1,\tau_1,\tau_2)-\tau_2}+n^{\omega(1,1,\tau_1)-\tau_1})n^{-\varepsilon})$ for some $\varepsilon > 0$, then we can break \Cref{con:2hintedOuMv}.

We have to maintain the product $UR_1NR_2V$, where $R_1$ and $R_2$ are initially the zero matrices. Let $I = (i_1,...i_{t_1})$, in phase \ref{phase:setI} we perform $t_1$ element updates to $R_1$ where the $k$th update sets $(R_1)_{k,i_k}=1$. In phase \ref{phase:setJ} we analogously set $(R_2)_{j_k,k}=1$ by performing $t_2$ element updates, where $J = (j_1,...,j_{t_2})$. We now have $UN_{I,J}V = UR_1NR_2V$, so in phase \ref{phase:query} we simply query the entry $(i,j)$.
\end{proof}

From the graph theoretic representation of the \doubleHintedOuMv{} problem, we obtain the same lower bound for the dynamic transitive closure problem (and dynamic DAG path counting as well as dynamic $k$-path for $k \ge 5$) with edge updates and pair queries. The same lower bound can also be obtained for cycle detection (and $k$-cycle detection for $k \ge 6$) by adding an edge from the first to the last layer during the query phase. The reduction from transitive closure to strong connectivity is done as in \cite[Lemma 6.4]{AbboudW14}.

\begin{corollary}\label{cor:lowerBoundTransitiveClosure}
Assuming the \doubleHintedOuMv{} \Cref{con:2hintedOuMv}, dynamic transitive-closure (and DAG path counting, strong connectivity, $k$-path, cycle detection and $k$-cycle detection) with polynomial pre-processing time, element updates and element queries requires $\Omega(\min_{\tau_1,\tau_2}(n^{\tau_1+\tau_2}+n^{\omega(1,\tau_1,\tau_2)-\tau_2}+n^{\omega(1,1,\tau_1)-\tau_1})n^{-\varepsilon})$ worst-case time for all $\varepsilon > 0$ for updates or queries.

For current $\omega$, this implies a lower bound of $\Omega(n^\fastExponent)$.
\end{corollary}

Since dynamic matrix inverse (and adjoint) can be used to maintain a 5-matrix product (see \Cref{thm:matrixProductReduction} and \Cref{cor:inverseAdjointEquivalence}), we obtain the same lower bound for dynamic matrix inverse with element updates and queries. The lower bound extends even to determinant and rank. For determinant this is because element update/query determinant is equivalent to element update/query inverse (\Cref{thm:inverseDeterminantEquivalence}). For rank the reduction is a bit longer using graph problems:

For element queries, transitive closure can be solved via \streach{}. For the reduction we only have to prove that even though $s$ and $t$ are fixed, the reachibility between any pair $(u,v)$ can be queried. For this we simply add edges $(s,u)$ and $(v,t)$ and check if $s$ can reach $t$. Afterward we remove these two edges again.
Thanks to \cite{AbboudW14}, we know bipartite perfect matching can solve \streach{}, which in turn can be solved by rank via \Cref{thm:bipartiteMatchingReduction}. Thus we obtain the following corollaries:

\begin{corollary}\label{cor:elementUpdateLB}
Assuming the \doubleHintedOuMv{} \Cref{con:2hintedOuMv}, any dynamic matrix inverse (and dynamic adjoint, determinant, rank) algorithm with polynomial time pre-processing, element updates and element queries requires $\Omega((n^{\tau_1+\tau_2}+n^{\omega(1,\tau_1,\tau_2)-\tau_2}+n^{\omega(1,1,\tau_1)-\tau_1})n^{-\varepsilon})$ worst-case time for all $\varepsilon > 0$ for updates or queries.

For current $\omega$, this implies a lower bound of $\Omega(n^\fastExponent)$.
\end{corollary}

\begin{corollary}\label{cor:lowerBoundBipartiteMatching}
Assuming the \doubleHintedOuMv{} \Cref{con:2hintedOuMv}, dynamic bipartite perfect matching with polynomial pre-processing time, element updates requires $\Omega((n^{\tau_1+\tau_2}+n^{\omega(1,\tau_1,\tau_2)-\tau_2}+n^{\omega(1,1,\tau_1)-\tau_1})n^{-\varepsilon})$ worst-case time for all $\varepsilon > 0$.

For current $\omega$, this implies a lower bound of $\Omega(n^\fastExponent)$.
\end{corollary}

%% file: lookahead.tex
\section{Look-Ahead Setting}
\label{sec:lookAhead}

In this section we will present our dynamic matrix inverse algorithm for the look-ahead setting. This setting can informally be described as follows: We know ahead of time, in which columns the future updates will be performed, i.e. we know the column indices, but not the column values ahead of time. A common setting where this look-ahead assumption is satisfied, are online/incremental problems where the input is revealed in order, one-by-one.

\begin{theorem}\label{thm:onlineRank}
There exists an algorithm that maintains the determinant and rank of an $n \times n$ matrix $A$, when the columns of $A$ are changed in order from left to right and the values of the next column are only given after answering the new determinant/rank of the matrix.

The algorithm runs in $O(n^\omega)$ total time.

\end{theorem}

\Cref{thm:onlineRank} directly implies an $O(n^\omega)$ upper bound on online bipartite matching via the reduction from \Cref{thm:bipartiteMatchingReduction}.

\begin{corollary}\label{cor:onlineMatching}
There exists an algorithm that maintains the size of the maximum cardinality matching in a bipartite graph, when the nodes on the right hand side are added one-by-one. The next node is only added after answering the size of the current maximum cardinality matching.

The algorithm runs in $O(n^\omega)$ total time.
\end{corollary}

Note that the reduction from dynamic rank (\Cref{thm:onlineRank}) to dynamic matrix inverse/determinant is an adaptive reduction, i.e. the position of future updates depends on the results of previous updates. As such, the reduction usually does not work in the look-ahead setting. We want to point out, that our look-ahead assumption is very weak: We actually do not require the exact column position of the updates. Instead, a rough estimate for the update position is enough, and this estimate is allowed to become more and more inaccurate, the further we look ahead into the future. 
The exact definition of our look-ahead setting will be given in \Cref{def:lookahead}, but first we will give a high-level idea of how our look-ahead algorithm works.

\paragraph{High-Level idea:}

Similar to the dynamic matrix inverse algorithms from \Cref{sec:dynamicInverse}, we maintain a transformation matrix of the form $T=\I+N$, where $N$ has few non-zero columns.
In section \Cref{sec:dynamicInverse} we presented two approaches on how to maintain the inverse of $\I+N$:
When $J \subset [n]$ are the column indices of the non-zero columns of $N$, let $C_1 := (I+N)_{J,J}$ and $C_2 = N_{[n] \backslash J, J}$, then the inverse of $\I+N$ can be obtained by computing $(C_1)^{-1}$ and $-C_2(C_1)^{-1}$. For this, we have two choices: Either compute the matrix product $-C_2 (C_1)^{-1}$ explicitly, using fast matrix multiplication, or compute one row of the product via matrix-vector multiplication, whenever a row is required/queried.

The algorithms from \Cref{sec:dynamicInverse} have in common that the required rows of $(\I+N)^{-1}$ (i.e. rows of product $-C_2 (C_1)^{-1}$) are exactly the rows with the same index as the columns in which we perform the updates.
So if we know the location of the updates ahead of time, we obtain the following trivial speedup:
Instead of computing the required rows of $-C_2 (C_1)^{-1}$ slowly via matrix-vector product, or computing the entire $-C_2 (C_1)^{-1}$ ahead of time using fast matrix-multiplication, we can simply use fast matrix multiplication to compute \emph{only the required rows} of $-C_2 (C_1)^{-1}$ ahead of time.

\paragraph{Structure of the Section}

In subsection \ref{sub:shortSequence} we will show how the idea of the previous paragraph can be used for a fast algorithm. When our look-ahead contains information about future $n^\varepsilon$ rounds, then our algorithm presented in subsection \ref{sub:shortSequence}, will be able to maintain the inverse for $n^\varepsilon$ rounds. In the succeeding subsection \ref{sub:longSequence} we show how the algorithm can be extended to support any number of rounds, so the number of updates that can be performed are no longer bounded by the size of the look-ahead.

\paragraph{Look-Ahead Assumption}

We will now define the exact type of look-ahead that we require. Note that this look-ahead does not require the actual values of the updates, as was the case in \cite{SankowskiM10}. Instead, we only require the column indices (and row indices) where the updates (and queries) will occur. Additionally, even these indices do not have to be known exactly, instead some rough estimate is enough, which allows our look-ahead algorithm to be used in adaptive reductions such as \Cref{thm:rankReduction}.

\begin{definition}
\label{def:lookahead}
Let $s,l \in \N$. Assume we have the task of maintaining the inverse of some matrix $A$ under column updates and row queries.

Let $(F_i^{(t)})_{0\le t, 0 \le i \le l}$ be a sequence with $F_i^{(t)} \subset [n]$, then we call $(F^{(t)})_{0\le t, 0 \le i \le l}$ an $s$-accurate $l$-look-ahead, if every round $t$ the set $F_i^{(t)}$ represents the possible column indices (and rows indices) for the updates to $A$ (and queries to $A^{-1}$) in the future rounds $t,t+1,...,t+i$. Additionally this sequence of sets has to satisfy:
\begin{align*}
\forall 0 \le i \le l, 0\le t :& \: |F_i^{(t)}| \le s\cdot i \\
\forall t_1,t_2,i_1,i_2 \text{ such that } t_1 \le t_2 \le t_2+i_2 \le t_1+i_2 :& \: F_{i_2}^{(t_2)} \subset F_{i_1}^{(t_1)}
\end{align*}

Note that the sets $(F_i^{(t)})_{0\le i \le l}$ do not have to be known in rounds $t'<t$, but they must be known during round $t$.
\end{definition}

If we know the position of the future updates/queries exactly, then we have a 1-accurate look-ahead. However (especially for $s > 1$), the look-ahead is allowed to be more inaccurate the further we look into the future, or conversely, the look-ahead becomes more accurate for rounds that will occur very soon. The second condition of \Cref{def:lookahead} is just to make the look-ahead well-defined, i.e. if we are currently in round $t$ and suspect some columns (rows) to be used in the future, then we must have suspected them already in earlier rounds $t' < t$.

Note that at time $t$ the sets $(F_i^{(t)})_{0 \le i \le l}$ can be represented using $O(n)$ memory, via $\tilde{F}_i^{(t)} = F_i^{(t)} \backslash \bigcup_{k = 0}^{i-1} F_k^{(t)}$. This way we can give $(\tilde{F}_i^{(t)})_{0\le i\le l}$ as part of our input at time $t$ without bloating the runtime by the size of this input. When the algorithm needs to use some $F_i^{(t)}$, then it can be constructed via $F_i^{(t)} = \bigcup_{k=0}^i \tilde{F}_k^{(t)}$ in $O(i)$ time.

\subsection{Short Sequence of Updates}
\label{sub:shortSequence}

As outlined before, the look-ahead algorithm is based on the idea to maintain only a subset of the rows of some inverse. For this we require the following lemma:

\begin{lemma}\label{lem:partialUpdateInverse}
Let $T$ be an $n^\delta \times n^\delta$ matrix. Let $C$ be a matrix with non-zero columns in $J_C$, where $|J_C| \le n^{\varepsilon_C}$. Then we can compute rows $I \subset [n^\delta]$ of $(T+C)^{-1}$ in $O(n^{\omega(\varepsilon_I,\varepsilon_C,\delta)})$ field operations, if $J_C \subset I$, $|I| = n^{\varepsilon_I}$.

The computation can be performed, if we only know $C$ and the values of the rows $I \subset [n]$ of $T^{-1}$. All other values do not need to be known to compute rows $I$ of $(T+C)^{-1}$.
\end{lemma}

\begin{algorithm}
\caption{PartialUpdateInverse (\Cref{lem:partialUpdateInverse})}\label{alg:partialUpdateInverse}
\begin{algorithmic}[1]
\REQUIRE Let $T$ be an $n^\delta \times n^\delta$ matrix and let $I \subset [n^\delta]$. As input we are given set $I$ and rows $I$ of $T^{-1}$ as well as the matrix $C$.
\ENSURE  Rows $I$ of $(T+C)^{-1}$.

\STATE $M \leftarrow \I$
\STATE $M_{I,[n^\delta]} \leftarrow M_{I,[n^\delta]} + (T^{-1}C)_{I,[n^\delta]}$
\STATE $(M^{-1})_{I,[n^\delta]} \leftarrow \textsc{PartialInvert}(M_{I,[n^\delta]},I)$ (\Cref{alg:partialInvert}) \label{line:partialUpdateInverseInversion}
\RETURN $(M^{-1})_{I,[n^\delta]}T^{-1}$

\end{algorithmic}
\end{algorithm}

\begin{proof}
We write the change to $T$ as a linear transformation $T+C = TM$, where $M = \I+T^{-1}C$. Thus we have $(T+C)^{-1} = M^{-1}T^{-1}$.

First we compute rows $I$ of $M$, which requires to compute the product $(T^{-1})_{I,[n^\delta]}C$, which can be done in $O(n^{\omega(\varepsilon_I,\delta,\varepsilon_C)})$ operations. Note that by assumption, the required rows $I$ of $T^{-1}$ are known.

Next, we have to compute rows $I$ of $M^{-1}$. The matrix $M$ is of the form $I+N$ where $N$ is non-zero in the columns given by $J_C \subset I$. Rows $I$ of $M^{-1}$ can thus be computed in $O(n^{\omega(\varepsilon_I,\varepsilon_C,\varepsilon_C)})$ operations via \Cref{lem:Tinverse}.

Since $M^{-1}$ is again of the form $\I+N$, where $N$ has non-zero column in $J_C \subset I$, we can compute $((T+C)^{-1})_{I,[n^\delta]} = (M^{-1})_{I,[n^\delta]}T^{-1}$. Note that $(M^{-1})_{I,[n^\delta]}T^{-1} = (M^{-1})_{I,I}(T^{-1})_{I,[n^\delta]}$, so all rows of $T^{-1}$, which are required to compute the product, are known and the product can be computed in $O(n^{\omega(\varepsilon_I,\varepsilon_C,\delta)})$ operations (via \Cref{lem:complexityTproduct}).
\end{proof}

\begin{lemma}\label{lem:Tlookahead}
Let $0 \le \varepsilon\le\delta \le 1$ and $s \in \mathbb{N}$ be a constant. Let $(F^{(t)}_i)_{0 \le i \le n^\varepsilon, t\ge 0}$ be an $s$-accurate $n^\varepsilon$-look-ahead of both column updates (and queries) to some $n^\delta \times n^\delta$ matrix $T^{(t)}$ (and its inverse $(T^{(t)})^{-1}$).

Then there exists a transformation maintenance algorithm, that maintains the inverse of $T^{(t)}=\I+\sum_{k=0}^t C^{(t)}$ for $0 \le t < n^\varepsilon$, supporting column updates to $T^{(t)}$ (i.e. $T^{(t)} = T^{(t-1)}+C^{(t)}$, where $C^{(t)}$ is zero everywhere except for one column with index $J^{(t)}$) and row queries to $(T^{(t)})^{-1}$. The complexity for both updates and queries is $O(n^{\omega(\varepsilon,\delta,\varepsilon)-\varepsilon})$.

For the initialization, the algorithm requires some matrix $M$ such that the rows with index in $F^{(0)}_{n^\varepsilon}$ of $M$ and $(T^{(0)})^{-1}$ are identical.

\end{lemma}

\begin{algorithm}
\caption{LookAhead (\Cref{lem:Tlookahead})}\label{alg:lookAhead}
\begin{algorithmic}[1]
\REQUIRE Let $(F_i^{(t)})_{0\le i \le l,t \ge 0}$ be an $s$-accurate $n^{\varepsilon}$-look-ahead as in \Cref{def:lookahead}. We are given an $n^\delta \times n^\delta$ matrix $T^{(0)} = I+C^{(0)}$, where the non-zero columns of $C^{(0)}$ have their index in $J^{(0)} \subset [n^\delta]$. We are also given the inverse $(T^{(0)})^{-1}$.
\renewcommand{\algorithmicensure}{\textbf{Maintain:}}
\ENSURE For every $1 \le i \le \log n^{\varepsilon}$ we maintain a copy $(T^{(a2^{i-1})}_{(i)})^{-1}$ of $(T^{(a2^{i-1})})^{-1}$, that is updated every $2^{i-1}$ updates and for which we only compute the columns corresponding to the position of the past $2^{i-1}$ updates/queries and the future $s \cdot 2^i$ possible positions for updates/queries.
\end{algorithmic}

\begin{algorithmic}[1]
\renewcommand{\algorithmicensure}{\textbf{initialization:}}
\ENSURE
\renewcommand{\algorithmicensure}{\textsc{Initialize}$(M)$}
\ENSURE
\FOR{$i = 1,...,\log n^\varepsilon + 1$}
\STATE $(T^{(0)}_{(i)})^{-1} \leftarrow M$.
\ENDFOR
\end{algorithmic}

\begin{algorithmic}[1]
\renewcommand{\algorithmicensure}{\textbf{update operation:}}
\ENSURE (We receive $C^{(t)}$, which is zero everywhere except of the column with index $J^{(t)}$, and $(F_i^{(t)})_{0 \le i \le n^\varepsilon}$.)
\renewcommand{\algorithmicensure}{\textsc{Update}$((F_i)_{0 \le i \le n^\varepsilon}, C)$}
\ENSURE
\STATE $t \leftarrow t+1$, $(F_i^{(t)})_{0 \le i \le n^\varepsilon} \leftarrow (F_i)_{0 \le i \le n^\varepsilon}$
\STATE Let $J^{(t)}$ be the column index of the non-zero column of $C$.
\FOR{$i = 1,...,\log n^\varepsilon$}
\IF{$t$ is of the form $t = a2^{i} + 2^{i-1}$ for some $a \in \mathbb{N}$}
\STATE $I_{(i)} \leftarrow F^{(t)}_{2^i} \cup \bigcup_{k=t-2^{i-1}+1}^t J^{(k)}$ the set of the last $2^{i-1}$ and $s \cdot 2^i$ future update/query positions, though we will never perform more than $n^\varepsilon-1$ updates in total.
\STATE $(T^{(a2^{i}+2^{i-1})}_{(i)})^{-1}_{I_{(i)},[n^\delta]} \leftarrow \textsc{PartialUpdateInverse}(I_{(i)},\: (T^{(a2^{i})}_{(i+1)})^{-1}_{I_{(i)},[n^\delta]}, \sum_{k=t-2^{i-1}+1}^{t} C^{(k)})$ (\Cref{alg:partialUpdateInverse}) \label{line:lookaheadInvervion}
\ELSIF{$t$ is of the form $t = a2^{i}$ for some $a \in \mathbb{N}$}
\STATE $I_{(i)} \leftarrow F^{(t)}_{2^i} \cup \bigcup_{k=t-2^{i}+1}^t J^{(k)}$ the set of the last $2^{i}$ and $s \cdot 2^i$ future update/query positions, though we will never perform more than $n^\varepsilon-1$ updates in total.
\STATE $(T^{(a2^{i})}_{(i)})^{-1}_{I_{(i)},[n^\delta]} \leftarrow \textsc{PartialUpdateInverse}(I_{(i)},\: (T^{(a2^{i}-2^i)}_{(i+1)})^{-1}_{I_{(i)},[n^\delta]}, \sum_{k=t-2^{i}+1}^{t} C^{(k)})$ (\Cref{alg:partialUpdateInverse}) \label{line:lookaheadInvervion2}
\ENDIF
\ENDFOR
\end{algorithmic}

\begin{algorithmic}[1]
\renewcommand{\algorithmicensure}{\textbf{query operation:}}
\ENSURE (Querying some row with index $j$ of $(T^{(t)})^{-1}$)
\renewcommand{\algorithmicensure}{\textsc{Query}$((F_i)_{0 \le i \le n^\varepsilon}, j)$}
\ENSURE
\STATE \textsc{LookAhead.Update$((F_i)_{0 \le i \le n^\varepsilon}, \text{zero-matrix})$} \COMMENT{Increase $t$ because queries are empty updates. Inside \textsc{LookAhead.Update} consider $j$ to be a non-zero column of $C$, even though it is actually zero. }
\RETURN row $j$ of $(T^{(t)}_{(1)})^{-1}$.

\end{algorithmic}
\renewcommand{\algorithmicensure}{\textbf{Output:}}
\end{algorithm}

\begin{proof}
For simplicity assume $n^\varepsilon$ to be a power of two and we will first focus on the case where all operations are column updates to $T^{(t)}$ and there are no row queries to $(T^{(t)})^{-1}$. We will later extend the algorithm to support query operations.

The idea behind the algorithm is to maintain $\log(n^\varepsilon)$ copies of $(T^{(t)})^{-1}$ denoted with $(T^{(t)}_{(i)})^{-1}$ for $i=1,...,\log n^\varepsilon$, where the $i$th copy is only updated after every $2^{i-1}$ updates. Each copy will only maintain a subset of the rows of $(T^{(t)})^{-1}$. Whenever the $i$th copy is updated, the computation will be spread over the next $2^{i-2}$ rounds (for $j \le 2$, we just perform the computation directly). We will see, that this helps us to obtain a good worst-case bound, instead of just amortized bounds.

To prove the correctness, we will prove the following statement:

\paragraph{Claim}
For every $0 \le t < n^\varepsilon$, every $i = 1,...,\log n^\varepsilon$ and every $a$ such that $a 2^{i-1} \le t$ ( i.e. $a 2^{i-1}$ is some point in time $t'$ in the past, where the the matrix $(T^{(t')}_{(i)})^{-1}$ was updated) let $I_{(i)} := F^{(a 2^{i-1})}_{2^i} \cup \bigcup_{k=a 2^{i-1}-2^{i-1}+1}^{a 2^{i-1}} J^{(k)}$, then we have $(T^{(a 2^{i-1})})^{-1}_{I_{(a 2^{i-1})},[n^\varepsilon]} = ((T^{(a 2^{i-1})})^{-1})_{I_{(a 2^{i-1})},[n^\varepsilon]}$.

To clarify: the set $I_{(i)}$ are the indices corresponding to the column indices of the past $2^{i-1}$ updates and the possible locations of the future $2^i$ updates, relative to some time $t' = a 2^{i-1}$. When at round $t'$ the matrix $(T^{(t')}_{(i)})^{-1}$ is updated, the rows given by it $I_{(i)}$ will be exactly the rows of the inverse $(T^{(t')})^{-1}$.

\paragraph{Base case}
At the time of the initialization $t=0$ we have $I_{(i)} \subset F^{(0)}_{n^\varepsilon}$ for all $i=1...\log n^{\varepsilon}$ and we are given a matrix $M$ such that rows $F^{(0)}_{n^\varepsilon}$ of $M$ and $(T^{(0)})^{-1}$ are the same. We set each $(T^{(0)}_{(i)})^{-1}$ be a reference to this matrix $M$, then for every $i=1,...,\log n^\varepsilon$ the matrix $(T^{(0)}_{(i)})^{-1}$ satisfies our claim. Note, that we do not have to create $\log n$ copies of $M$. When the algorithm tries to access some $(T^{(0)}_{(i)})^{-1}$, the algorithm can instead access $M$ via the reference. This way we do not require any pre-processing.

We now want to prove via induction, that \Cref{alg:lookAhead} maintains the matrices $(T^{(t)}_{(i)})^{-1}$ as advertised, i.e. we prove that when some $(T^{(t)}_{(i)})^{-1}$ is updated, we really compute rows $I_{(i)}$ of $(T^{(t)})^{-1}$.
For this, we consider two cases: if $t$ is a multiple of $2^{i-1}$, then it could also be a multiple of $2^i$, so the two cases are $t = a2^i + 2^{i-1}$ and $t = a2^i$ for some $a$.

\paragraph{First case: $t = a2^i + 2^{i-1}$ for some $a$}

First, consider the case where $t = a2^i + 2^{i-1}$, i.e. it is not a multiple of $2^{i}$. We now have to update $(T^{(t)}_{(i)})^{-1}$. 
We can compute rows $I_{(i)}$ of $(T^{(t)})^{-1}$ via \Cref{lem:partialUpdateInverse}. For this we require rows $I_{(i)}$ of $(T^{(a2^i)})^{-1}$ and the change $C := \sum_{i=a2^i}^{a2^i+2^{i-1}+1} C^{(t)}$.

The rows $I_{(i)}$ of $(T^{(a2^i)})^{-1}$ are known via the matrix $(T^{(a2^i)}_{(i+1)})^{-1}$ by inductive assumption, because $F^{(a 2^i + 2^{i-1})}_{2^i} \subset F^{(a 2^{i})}_{2^{i+1}}$ and $\bigcup_{k=a 2^{i}+1}^{a 2^{i}+2^{i-1}} J^{(k)} \subset F^{(a 2^{i})}_{2^{i+1}}$.

Thus, we can compute rows $I_{(i)}$ of $(T^{(t)})^{-1}$ via the algorithm of \Cref{lem:partialUpdateInverse} in $O(2^{\omega(i-1,i+1,\log n^\delta)})$ operations, because $|I_{(i)}| \le s2^i + 2^{i-1} < s2^{i+1}$ and we change upto $2^{i-1}$ columns.

Note that this total cost $O(2^{\omega(i-1,i+1,\log n^\delta)})$ will be spread over the next $2^{i-2}$ updates (if $i > 2$), to get a better worst-case bound. The required matrix $(T^{(a2^i)}_{(i+1)})^{-1}$ was also computed in this delayed fashion, and had their update computation spread over $2^{i-1}$ rounds. This is not an issue, because this means matrix $(T^{(a2^i)}_{(i+1)})^{-1}$ becomes available in round $a2^i + 2^{i-1}$, so we can access it.

\paragraph{Second case: $t = a2^i + 0 \cdot 2^{i-1}$ for some $a$}

Obviously the equality $t = a2^i + 0 \cdot 2^{i-1}$ could be true for more than one $i$. The following proof holds for all such $i$.

We want to update $(T^{(t)}_{(i)})^{-1}$, i.e. compute rows $I_{(i)}$ of $(T^{(a2^i)})^{-1}$, though we will now compute the rows for a slightly larger set $\tilde{I}_{(i)} := I_{(i)} \cup \bigcup_{k=a 2^{i}-2^i+1}^{a 2^{i}} J^{(k)}$ instead.

Computing these rows of the inverse can, similar to the previous case, be done via \Cref{lem:partialUpdateInverse}. For this we use rows $\tilde{I}_{(i)}$ of $(T^{(a2^i-2^{i})})^{-1}$ and the change $C := \sum_{i=a2^i-2^i+1}^{a2^i}C^{(i)}$.

The rows $\tilde{I}_{(i)}$ of $(T^{(a2^i)})^{-1}$ are known via the matrix $(T^{(a2^i)}_{(i+1)})^{-1}$ by inductive assumption, because $F^{(a 2^i)}_{2^i} \subset F^{(a 2^{i}- 2^{i})}_{2^{i+1}}$ and $\bigcup_{k=a 2^{i}-2^i+1}^{a 2^{i}} J^{(k)} \subset F^{(a 2^{i}-2^i)}_{2^{i+1}}$.

Thus, we can compute rows $\tilde{I}_{(i)}$ of $(T^{(t)})^{-1}$ using the algorithm of \Cref{lem:partialUpdateInverse} in $O(2^{\omega(i,i+1,\log n^\delta)})$, because $|\tilde{I}_{(i)}| \le s2^i + 2^i < s2^{i+1}$ and we change upto $2^{i}$ columns.

Note that this total cost $O(2^{\omega(i,i+1,\log n^\delta)})$ will be spread over the next $2^{i-2}$ updates (if $i > 2$), to get a better worst-case bound. The required matrix $(T^{(a2^i-2^i)}_{(i+1)})^{-1}$ had their cost also spread over $2^{i-1}$ rounds, when it was last updated. This means the matrix is available in round $a2^i$, so the delay of computation is not an issue.

\paragraph{Cost}

The worst-case cost is $\sum_{i=1}^{\log n^\varepsilon} O(2^{\omega(i,i+1,\log n^\delta)})/2^i$.
The term $2^{\omega(i,i+1,\log n^\delta)}$ grows faster in $i$ than $2^i$ (or equally fast, if matrix multiplication turns out to have a linear complexity), so we can simplify the worst-case cost to $O(n^{\omega(\varepsilon,\varepsilon,\delta)-\varepsilon})$ (with an additional $\log n$ factor, if matrix multiplication turns out to have a linear complexity).

Note that the sets $(F_i^{(t)})_{0 \le i \le n^\varepsilon}$ are technically given via $(\tilde{F}_i^{(t)})_{0\le i\le n^\varepsilon}$, where $\tilde{F}_i^{(t)} = F_i^{(t)} \backslash \bigcup_{k = 0}^{i-1} F_k^{(t)}$. Each set $F_i^{(t)}$ can be constructed via $F_i^{(t)} = \bigcup_{k=0}^i \tilde{F}_k^{(t)}$ in $O(i)$ time, but since we use $F_{2^i}^{(t)}$ only every $2^{i-1}$ rounds, this does not further affect the complexity of the algorithm, because the cost for constructing one $F_{2^i}^{(t)}$ is $O(2^i) = O(n^\varepsilon)$.

\paragraph{Queries}

It is easy to see that this algorithm does already support row queries to $(T^{(t)})^{-1}$, because an update to the $j$th column of $T^{(t)}$ will compute the $j$th row of $(T^{(t)})^{-1}$. Hence a query to the $j$th row of $(T^{(t)})^{-1}$ can be represented by a column update where we add a zero vector to the $j$ row of $T^{(t)}$. Then, when the update is performed by the algorithm, we simply output the computed row of $(T^{(t)})^{-1}$ that can be found in the matrix $(T^{(t)}_{(1)})^{-1}$.

\end{proof}

\begin{corollary}\label{cor:TlookaheadLarge}
Let $0 \le \varepsilon\le\delta \le 1$ and $s \in \mathbb{N}$ constant. Let $(F^{(t)}_i)_{0 \le i \le n^\varepsilon, t\ge 0}$ be a $s$-accurate $n^\varepsilon$-lookahead of both column updates (and queries) to some $n \times n$ matrix $T^{(t)}$ (and its inverse $(T^{(t)})^{-1}$). Assume $T^{(0)} = \I + N$, where $N$ has at most $n^\delta$ non-zero columns, let $J \subset [n]$ be the column indices of these non-zero columns.

Then there exists a transformation maintenance algorithm, that maintains the inverse of $T^{(t)}=\I+\sum_{k=0}^t C^{(t)}$ for $0 \le t < n^\varepsilon$, supporting column updates to $T^{(t)}$ (i.e. $T^{(t)} = T^{(t-1)}+C^{(t)}$, where $C^{(t)}$ is zero everywhere except for one column) and row queries to $(T^{(t)})^{-1}$. The complexity for both updates and queries is $O(n^{\omega(\varepsilon,\delta,\varepsilon)-\varepsilon})$. 

For the initialization, the algorithm requires some matrix $M$ such that the rows with index in $F^{(0)}_{n^\varepsilon}$ of $M$ and $(T^{(0)})^{-1}$ are indentical.

\end{corollary}

\begin{proof}
Let $I = J \cup F^{(0)}_{n^\varepsilon}$, then we know via \Cref{lem:Tinverse} that $((T^{(t)})^{-1})_{I,I} = (T^{(t)}_{I,I})^{-1}$ and $((T^{(t)})^{-1})_{I,[n] \backslash I} = 0$ for all $0 \le t \le n^\varepsilon$. Hence we can run the algorithm from \Cref{lem:Tlookahead} on the smaller matrix $T^{(t)}_{I,I}$. We have $|I| \le sn^\varepsilon + n^\delta = O(n^\delta)$, so the complexity for updates and queries is $O(n^{\omega(\varepsilon,\delta,\varepsilon)-\varepsilon})$ each.

Note that $F^{(t)}_i$ might return indices outside of $I$ for $t > 0$, but since we only perform $n^\varepsilon-1$ updates, it is enough to use a restricted look-ahead $F'^{(t)}_i := F^{(t)}_i \cap F^{(0)}_{n^\varepsilon}$.
\end{proof}

\subsection{Long Sequence of Updates}
\label{sub:longSequence}

In this subsection we will extend the algorithm from \Cref{cor:TlookaheadLarge}, so it supports more than $n^\varepsilon$ rounds. The high-level idea is similar to the one from \Cref{sec:dynamicInverse}, where we presented different transformation maintenance algorithms and showed in \Cref{lem:combineT} that they could be combined to a faster transformation maintenance algorithm. We can do the same again with the new look-ahead transformation maintenance algorithm from \Cref{cor:TlookaheadLarge}, by combining it with another transformation maintenance algorithm.

\begin{algorithm}
\caption{ColumnUpdateRowQueryLookAhead (\Cref{lem:combineTLookahead})}\label{alg:CombinedTransformationLookAhead}
\begin{algorithmic}[1]
\REQUIRE An $n \times n$ matrix $T^{(0)} = \I+C^{(0)}$ and inputs of the form $T^{(t)} = T^{(t-1)}+C^{(t)}$, where $C^{(t)}$ has a single nonzero column with index $J^{(t)}$, and a $s$-accurate $n^{\varepsilon_0}$-look-ahead $(F_i^{(t)})_{0\le i \le n^{\varepsilon_0}, t \ge 0}$. 
\ENSURE  Maintain $(T^{(t)})^{-1}$ in an implicit form s.t. rows can be queried.

\end{algorithmic}

\begin{algorithmic}[1]
\renewcommand{\algorithmicensure}{\textbf{initialization:}}
\ENSURE
\renewcommand{\algorithmicensure}{\textsc{Initialize}$((F_i)_{0 \le i \le n^\varepsilon}, T^{(0)})$}
\ENSURE
\STATE \textsc{ColumnUpdateRowQuery.Initialize$(T^{(0)})$} (Initialize \Cref{alg:CombinedTransformation} on $T^{(0)}$)
\STATE $S^{(0)} \leftarrow $ zero-matrix
\STATE $t',t \leftarrow 0$
\STATE $M^{(0)} \leftarrow \I$
\IF{$T^{(0)} \neq \I$}
\STATE $M^{(0)}_{F_{n^{\varepsilon_0}},[n]} \leftarrow$ \textsc{ColumnUpdateRowQuery.Query$(F_{n^{\varepsilon_0}},[n])$} 
\ENDIF
\STATE \textsc{LookAhead.Initialize$(M^{(0)})$} (Initialize \Cref{alg:lookAhead} / \Cref{cor:TlookaheadLarge} on $M^{(0)}$).
\end{algorithmic}

\begin{algorithmic}[1]
\renewcommand{\algorithmicensure}{\textbf{update operation:}}
\ENSURE
\renewcommand{\algorithmicensure}{\textsc{Update}$((F_i)_{0 \le i \le n^\varepsilon}, C)$}
\ENSURE
\STATE $t \leftarrow t+1$
\STATE Let $J^{(t)}$ be the column index where $C$ is non-zero.
\IF{$|\bigcup_{i=t'+1}^{t} J^{(i)}| \ge n^{\varepsilon_0}$} 
\STATE \textsc{ColumnUpdateRowQuery.Update$(S^{(t)})$} (Update $(T^{(t)})^{-1}$ via \Cref{alg:CombinedTransformation}) \label{line:CombinedTransformationLookAheadResetInversion}
\STATE $t' \leftarrow t$
\STATE $S^{(t)} \leftarrow $ zero-matrix
\STATE $M^{(t)}_{F_{n^{\varepsilon_0}}^{(t)},[n]} \leftarrow$ \textsc{ColumnUpdateRowQuery.Query$(F_{n^{\varepsilon_0}},[n])$}
\STATE \textsc{LookAhead.Initialize$(M^{(t)})$} (Reinitialize \Cref{alg:lookAhead} / \Cref{cor:TlookaheadLarge}).
\ELSE 
\STATE $S^{(t)} \leftarrow S^{(t-1)}+ C$
\STATE \textsc{LookAhead.Update$((F_i)_{0 \le i \le n^{\varepsilon_0}}, C)$} (Update $(M^{(t)})^{-1}$ via \Cref{alg:lookAhead}). \label{line:CombinedTransformationLookAheadInversion}
\ENDIF
\end{algorithmic}

\begin{algorithmic}[1]
\renewcommand{\algorithmicensure}{\textbf{query operation:}}
\ENSURE (Querying some row $i$ of $(T^{(t)})^{-1})$)
\renewcommand{\algorithmicensure}{\textsc{Query}$((F_i)_{0 \le i \le n^\varepsilon}, i)$}
\ENSURE
\RETURN \textsc{LookAhead.Query$((F_i)_{0 \le i \le n^\varepsilon}, i)$} (Query row $i$ from $(M^{(t)})^{-1}$ via \Cref{alg:lookAhead} / \Cref{cor:TlookaheadLarge}).

\end{algorithmic}
\end{algorithm}

\begin{lemma}\label{lem:combineTLookahead}
Let $0 \le \varepsilon_0 \le \varepsilon_1 \le \varepsilon_2 \le 1$ and $s \in \N$ constant.
Given an $s$-accurate $n^{\varepsilon_0}$-lookahead $(F^{(t)}_i)_{0\le i \le n^{\varepsilon_0}, 0 \le t}$ for the column indices of the future updates and row indices of the future queries, there exists a transformation maintenance algorithm that maintains the inverse of $T^{(t)}$, supporting column updates to $T^{(t)}$ and row queries to the inverse $(T^{(t)})^{-1}$. Assume that throughout the future updates the form of $T^{(t)}$ is $\I + N$, where $N$ has always at most $n^{\varepsilon_2}$ nonzero columns (e.g. by restricting the number of updates $t \le n^{\varepsilon_2}$). Each column update and row query requires $O(
n^{\omega(\varepsilon_2,\varepsilon_1,\varepsilon_0)-\varepsilon_0}
+n^{\omega(1,\varepsilon_2,\varepsilon_2)-\varepsilon_1})$ field operations.

The pre-processing time is bounded by $O(n^\omega)$, though if $T^{(0)} = I$ the algorithm requires no pre-processing.

\end{lemma}

\begin{proof}[Proof of \Cref{lem:combineTLookahead}]
The proof is similar to \Cref{lem:combineT} as we simply combine two existing transformation maintenance algorithms.

Let $T^{(t)}$ be the matrix at round $t$, i.e. $T^{(0)}$ is what the matrix looks like at the time of the initialization/pre-processing. The matrix $(T^{(t)})^{-1}$ will be maintained by two algorithms: We use the algorithm from \Cref{lem:combineT} to perform batch updates of size $n^{\varepsilon_0}$ every $n^{\varepsilon_0}$ rounds and for the rounds in between these batch updates, the algorithm from \Cref{cor:TlookaheadLarge} handles all updates/queries.

As pre-processing we initialize the algorithm from \Cref{lem:combineT} on the matrix $(T^{(0)})^{-1}$, which needs $O(n^\omega)$ time unless $T^{(0)} = \I$. We also initialize the algorithm from \Cref{cor:TlookaheadLarge}, which requires rows $F^{(0)}_{n^{\varepsilon_0}}$ of $(T^{(0)})^{-1}$. If $T^{(0)} \neq \I$, then the requires rows can simply be queried from the algorithm of \Cref{lem:combineT}.

For the next $n^{\varepsilon_0}-1$ rounds, the algorithm from \Cref{cor:TlookaheadLarge} handles all updates and queries, which require $O(n^{\omega(\varepsilon_0,\varepsilon_2,\varepsilon_0)-\varepsilon_0})$ field operations each.
At round $t = a n^{\varepsilon_0}$ for some $a \in \mathbb{N}$, we perform a batch update to the algorithm of \Cref{lem:combineT} in $O(n^{\omega(\varepsilon_2,\varepsilon_1,\varepsilon_0)}+n^{\omega(1,\varepsilon_2,\varepsilon_2)-\varepsilon_1+\varepsilon_0})$ operations.
We also re-initialize the look-ahead algorithm from \Cref{cor:TlookaheadLarge}, which requires us to know rows $F^{(t)}_{n^{\varepsilon_0}}$ of $(T^{(t)})^{-1}$. These rows can be queried from the algorithm of of \Cref{lem:combineT} in $O(n^{\omega(\varepsilon_2,\varepsilon_1,\varepsilon_0)})$ operations.

The average update time is thus $O(
n^{\omega(\varepsilon_2,\varepsilon_1,\varepsilon_0)-\varepsilon_0}
+n^{\omega(1,\varepsilon_2,\varepsilon_2)-\varepsilon_1})$, which can be made worst-case via standard techniques, see for instance the proof of \Cref{lem:combineT}.

The query time is the same as the update time, since the algorithm of \Cref{cor:TlookaheadLarge} performs queries by performing an empty update.

\end{proof}

By setting $\varepsilon_2 = 1$ we obtain the following algorithm for column update, row query in the look-ahead setting.

\begin{theorem}\label{thm:columnLookAhead}
Let $0 \le \varepsilon_0 \le \varepsilon_1 \le 1$ and $s \in \N$.
Given an $s$-accurate $n^{\varepsilon_0}$-lookahead for the column indices of the future updates and row indices of the future queries, there exists a dynamic algorithm for maintaining the inverse of an $n \times n$ matrix $A$.
The algorithm supports column updates and row queries in $O(
n^{\omega(1,\varepsilon_1,\varepsilon_0)-\varepsilon_0}
+n^{\omega(1,1,\varepsilon_1)-\varepsilon_1)})$ operations.

The pre-processing requires $O(n^\omega)$ operations.

\end{theorem}

Note that given a $\Omega(n)$-lookahead, the column update/row query complexity is $O(n^{\omega-1})$, which is optimal given that it can be used to invert an entire matrix. For a $O(1)$-lookahead, the runtime coincides with the algorithm from \Cref{thm:columnUpdate}, for which we present a matching conditional lower bound in \Cref{sec:lowerBounds}. So for the two extreme points of the look-ahead, the algorithm is optimal. 

We also obtain an algorithm for element update and element query dynamic matrix inverse in the look-ahead setting. The proof is identical to \Cref{thm:TimpliesInverseElement}, which also holds in the look-ahead setting, so we can simply apply \Cref{thm:TimpliesInverseElement} to \Cref{lem:combineTLookahead}).

\begin{theorem}\label{thm:elementLookAhead}
Let $0 \le \varepsilon_0 \le \varepsilon_1 \le \varepsilon_2 \le 1$ and $s \in \N$.
Given an $s$-accurate $n^{\varepsilon_0}$-lookahead for the column indices of the future updates and row indices of the future queries, there exists a dynamic algorithm for maintaining the inverse of an $n \times n$ matrix $A$.
The algorithm supports element updates and element queries in 
$O(
n^{\omega(\varepsilon_2,\varepsilon_1,\varepsilon_0)-\varepsilon_0}
+n^{\omega(1,\varepsilon_2,\varepsilon_1)-\varepsilon_1)}
+n^{\omega(1,1,\varepsilon_2)-\varepsilon_1})$ operations.

The pre-processing requires $O(n^\omega)$ operations.

\end{theorem}

This look-ahead algorithm for element updates shares the property with the column update variant (\Cref{thm:columnLookAhead}), that given a $\Omega(n)$-look-ahead, the update/query complexity is $O(n^{\omega-1})$. Note that all current matrix rank algorithms for sparse matrices ($O(n)$ non-zero entries), can still require $\Omega(n^{\omega-1})$ operations in the worst-case \cite{CheungKL13}. So our look-ahead algorithm is optimal in the sense, that an improvement would lead to a faster matrix rank algorithm for sparse matrices.

For a $O(1)$-look-ahead, the runtime of the element update/query look-ahead algorithm coincides with the respective upper and lower bounds for the non-look-ahead variant (see \Cref{thm:columnUpdate} for the upper bound and \Cref{sec:lowerBounds} for lower bounds).

Now we are only left with proving the initial \Cref{thm:onlineRank} of this section.

\begin{proof}[Proof of \Cref{thm:onlineRank}]
If the updates to some $n \times n$ matrix $A$ are performed from left to right, then we obviously have a $1$-accurate $n$-look-ahead. We only have to check that this look-ahead is not lost when using the reduction from determinant/rank.

The reduction from determinant to inverse (\Cref{thm:inverseDeterminantEquivalence}) still satisfies the look-ahead assumption, because we always query exactly the row $i$ after performing an update to column $i$.

Even though the reduction from rank to determinant (\Cref{thm:rankReduction}) is adaptive, i.e. an update is performed depending on whether the determinant changed to 0, the reduction still satisfies our definition of a look-ahead (see the proof of \Cref{thm:rankReduction}).
The reduction embeds the input matrix $A$ in a larger matrix $\tilde{A}$ of size $3n \times 3n$. When the current rank is $r$, an update to the block in $\tilde{A}$, in which $A$ is embedded in, is followed by an update to entry $3n-r-1$, $3n-r$ or $3n-r+1$ in $\tilde{A}$ (depending on how the determinant changed). Hence we know that when performing $i$ updates to $A$, the reduction will actually perform $2i$ updates and every second of these updates has to be in a column with index in the range $3n-r-i,...,3n-r+i$, so if we have an $s$-accurate look-ahead for the updates to $A$, then we have a $O(s)$-accurate look-ahead for the updates to $\tilde{A}$.

\end{proof}

%% file: openproblems.tex
\section{Open Problems}
\label{sec:open}

\paragraph{Amortization}
All the results in this paper focus on worst-case update time. A major open problem is whether one can  get faster update time via amortization or describe reasonable conjectures that hold for amortized update time as well. Curently there only exist amortized lower bounds for sparse graphs and for algorithms with small pre-processing time \cite{AbboudW14}. (We could extend our conjectures to imply lower bounds for amortized update time by repeating some phases, but we do not feel that they are reasonable enough. If interested, \Cref{sec:amortizedLowerBounds} describes how to amortize the lower bounds via repetition of some phases.)
It will be already groundbreaking if amortization can improve the update time for some applications, such as $st$-reachability. 

\paragraph{Refuting  or supporting our conjectures}  In this paper we need to propose new conjectures to capture the power of dynamic matrix multiplication. Since these conjectures are new, they need to be scrutinized. Breaking one of these conjectures would give a hope for improved algorithms for many problems considered in this paper. It might also be possible to support these conjectures with, e.g. via algebraic circuit lower bounds.

\paragraph{Distances}
Many of the upper and lower bounds in \Cref{tbl:applications} are tight. However, for the distance problems ($st$-distance or all-pair-distances) there are no matching upper and lower bounds. The best lower bound so far is obtained via transitive closure/reachability, but the upper bound is far above this lower bound. A major open problem is to close or at least narrow this gap.

Another open problem related to distances would be to extend our results to weighted graphs. The results can easily be extended to support integer weights in $[1,W]$ at the cost of an extra $W$ factor. Thus these algebraic techniques are only suited for small integer weights. We wonder if it is possible to obtain an algorithm with $\log W$ dependency, e.g. via approximation as in \cite{Zwick02}.

\paragraph{Maintaining the object}

Algebraic techniques tend to only return a quantitative answer: The size of the maximum matching, the distance or reachability between two nodes etc.

Consider for instance our online bipartite matching algorithm. We trivially know which nodes on the right are part of the matching, as each newly added right node that increases the matching size must be part of the matching. Yet, we do not know which nodes on the left are matched or which edges are used. Is it possible to obtain the maintained object such as the matching or the path?

\paragraph{Sparse graphs}

So far dynamic matrix inverse is the only technique that returns a non-trivial upper bound for $st$-reachability. However, for sparse graph even this upper bound is slower than just trivially running breath/depth first search in $O(m)$ time. We wonder if it is possible to obtain a $O(m^{1-\varepsilon})$ upper bound or a $\Omega(m)$ conditional lower bound. 
Currently the best lower bound for sparse graphs is $\Omega(m^{0.814})$, if one assumes $O(m^{1.407})$ pre-processing time \cite{AbboudW14}.

\paragraph{Derandomization}

While the dynamic matrix inverse for non-singular matrices is deterministic, we require randomization to extend the result to the setting where the matrix is allowed to temporarily become singular. Likewise most graph application such as  reachability are randomized. Is it possible to derandomize some of these applications or can we make them at least las-vegas instead of monte-carlo?

%% file: acknowledgment.tex
\section*{Acknowledgment}

The authors would like to thank Amir Abboud for comments and pointing out an open problem in \cite{AbboudW14}. 
We thank Adam Karczmarz for pointing out the strong-connectivity application.

This project has received funding from the European Research Council (ERC) under the European
Unions Horizon 2020 research and innovation programme under grant agreement No 715672. Danupon
Nanongkai and Thatchaphol Saranurak were also partially supported by the Swedish Research Council (Reg. No. 2015-04659).

%% file: appendix.tex

\section{Runtime Analysis, Balancing Terms}
\label{app:runtime}

All our complexities depend on the matrix multiplication exponent $\omega$.
The matrix multiplication exponent is defined via:
\begin{align*}
\omega := \inf \{ s \mid \text{we can multiply two $n \times n$ matrices in } O(n^s) \text{ arithmetic operations}\}
\end{align*}
The best current bound is $\omega < \matrixExponent$ \cite{Gall14a}. Given this definition for $\omega$ via some infimum, multiplying two matrices technically requires $O(n^{\omega+\varepsilon})$ operations for every constant $\varepsilon>0$, but the $\varepsilon$ term is typically ignored in literature and instead $O(n^\omega)$ simply refers to the complexity of multiplying two matrices. Given that the exponent is rounded to some finite precision, ignoring the $\varepsilon$ term is typically not an issue.

For rectangular matrix multiplication of an $n^a \times n^b$ matrix and an $n^b \times n^c$ matrix, the term $\omega(a,b,c)$ is defined in the same way, and we have $\omega = \omega(1,1,1)$. The current best upper bounds for $\omega(1,1,k)$ are presented in \cite{GallU18}. These bounds can be extended to $\omega(a,b,c)$ via the following properties of the $\omega(\cdot,\cdot,\cdot)$ function: By splitting the matrices into submatrices, we obtain the bound $\omega(a,b,c+d) \le \omega(a,b,c)+d$. Since $n^a = (n^a)^1$, $n^b = (n^a)^{b/a}$, $n^c = (n^a)^{c/a}$ we have $\omega(a,b,c) = a\cdot \omega(1,b/a,c/a)$. Also the matrix multiplication exponent is a symmetric function, i.e. $\omega(a,b,c) = \omega(a,c,b) = ... = \omega(a,b,c)$.

With the previous observations on how to simplify matrix exponent term, we can use the upper bounds from \cite{GallU18} for $\omega(1,1,k)$ to get upper bounds on any $\omega(a,b,c)$ via the following routine. Without loss of generality $a \ge b \ge c$, otherwise rename/reorder the variables.\footnote{
	An online version for bounds on $\omega(a,b,c)$ is available at \url{https://people.kth.se/~janvdb/matrix.html}.
}
\begin{algorithmic}[1]
\IF{$a = b =c$}
\RETURN $a \cdot \omega$
\ENDIF
\IF{$a=b$}
\RETURN $a \cdot \omega(1,1,c/a)$
\ENDIF
\IF{$b=c$}
\RETURN $b \cdot \omega(1,1,a/b)$ (Note that \cite{GallU18} also proved bounds for $\omega(1,1,k)$ where $k > 1$.)
\ENDIF
\RETURN $\min\{(a-b) + b \cdot \omega(1,1,c/b), (b-c) + c \cdot \omega(1,1,a/c)\}$
\end{algorithmic}

All complexities in this paper are computed using a small optimization program. For instance for the complexity $O(n^{1+\varepsilon}+n^{\omega(1,1,\varepsilon)-\varepsilon})$ of \Cref{thm:columnUpdate}, we have to compute $\min_{0 \le \varepsilon \le 1} \max \{ 1+\varepsilon, \omega(1,1,\varepsilon)-\varepsilon\}$ in order to get the smallest possible update complexity.

\section{Worst-Case Standard Technique}

\begin{theorem}\label{thm:standardTechnique}
Let $\mathcal{A}$ be a dynamic algorithm with reset time $O(r)$, update time $O(u(t))$ and query time $O(q(t))$, where $t$ is the number of past updates, since the last reset or initialization.

For every $\mu \in \N$ there is an algorithm $\mathcal{W}$ with worst-case update complexity $O(u(\mu) + r / \mu)$ and query complexity $O(q(\mu))$.
\end{theorem}

If we were interested in amortized complexity, \Cref{thm:standardTechnique} would be trivial by simply reset the algorithm after every $\mu$ updates.

\begin{proof}[Proof of \Cref{thm:standardTechnique}]
For simplicity assume $\mu$ is a multiple of 4. We maintain two copies of $\mathcal{A}$ in parallel, where each copy will have the following life-cycle:
\begin{enumerate}
\item For the next $\mu / 4$ rounds (i.e. updates), the copy performs its reset operation. This means for every of the next $\mu / 4$ updates the algorithm will perform $O(r) / (\mu / 4)) = O(r / \mu)$ operations of the reset routine.
\item The reset routine was started $\mu / 4$ updates in the past, so the last $\mu / 4$ updates were not yet applied to the copy. To fix this we will always perform two queued updates for each of the next $\mu / 4$ rounds. This means we have $2 \cdot O(u(\mu/2)) = O(u(\mu))$ cost per round and after $\mu / 4$ rounds the copy has caught up with all queued updates.
\item For the next $\mu / 2$ rounds the copy can perform updates as usual and is able to answer queries.
\item The copy now has received a total of $\mu$ updates and needs to reset again, so jump back to step 1.
\end{enumerate}
Note that during this cycle, each copy is alternating between being unavailable (resetting + catching up) and being available (performing current updates/answering current queries) for a sequence of $\mu / 2$ rounds each. Thus we simply need two copies of the algorithm that are phase shifted, then one copy is always available for updates/queries.

Both copies are initialized at the same time, but the phase-shift can easily be obtained by simply resetting the first copy directly after the initialization. So copy 1 starts with phase 1 while copy 2 starts with phase 3.

\end{proof}

\input{applications.tex}

\section{Bit-Length Increase for the Matrix Product Reduction}
\label{app:bitlength}

In this section we will prove, that when using our algorithms for the reduction from \Cref{thm:matrixProductReduction} and input matrices with small (polylog bit-length) entries, our algorithms will internally use only polylog bit-length numbers.

\begin{theorem}
\label{thm:bitlength}
Let $A_1...A_s$ be matrices, where each $A_i$ is of size $n \times n$.

Consider the following $(sn) \times (sn)$ matrix $\tilde{A}$:
\begin{align*}
\tilde{A} := \left(\begin{array}{cccccc}
  \I   &   A_1  &    0   & \cdots & \cdots &    0   \\
   0   &   \I   &   A_2  &    0   &        & \vdots \\
\vdots & \ddots & \ddots & \ddots & \ddots & \vdots \\
\vdots &        &    0   &   \I   & A_{s-1}&    0   \\
\vdots &        &        &    0   &   \I   &   A_s  \\
   0   & \cdots & \cdots & \cdots &    0   &   \I
\end{array}\right)
\end{align*}
When our dynamic matrix inverse algorithms maintain the inverse of $\tilde{A}$ and the updates are only allowed to the matrices $A_1,...,A_s$, then the internal computations of our algorithms will be performed on numbers of length $O(\text{poly}(s) \text{ polylog}(n))$, if the entries of $A_1,...,A_s$ are bounded by polylog$(n)$.
\end{theorem}

To prove \Cref{thm:bitlength}, we will first generalize the structure of the matrix. 

\begin{definition}\label{def:blocktriangular}
We call $\tilde{A}$ a upper triangular block matrix, if there exist matrices $(A^{(i,j)})_{1 \le i \le j \le s}$, where each matrix $A^{(i,j)}$ is of size $n \times n$, such that:
\begin{align*}
\tilde{A} := \left(\begin{array}{cccccc}
  \I   & A^{(1,1)}& A^{(1,2)} & \cdots & \cdots & A^{(1,k)}  \\
   0   &   \I    &A^{(2,2)}&       &        & \vdots \\
\vdots & \ddots  & \ddots  & \ddots &  & \vdots \\
\vdots &         &    0    &   \I   &A^{(s-1,s-1)}&A^{(s-1,s)} \\
\vdots &         &         &    0   &   \I   &A^{(s,s)}\\
   0   & \cdots  & \cdots  & \cdots &    0   &   \I
\end{array}\right)
\end{align*}
\end{definition}

\begin{fact}\label{fact:blockProduct}
Let $\tilde{A}$ and $\tilde{B}$ be upper triangular block matrices. Then $\tilde{C}:=\tilde{A}\tilde{B}$ is also an upper triangular block matrix with $C^{(i,j)} = A^{(i,j)}+B^{(i,j)}+\sum_{t=i}^{j-2} A^{(i,t)}B^{(t+1,j)}$.
\end{fact}

\begin{lemma}\label{lem:upperTriangularInverse}
Let $\tilde{A}$ be a upper triangular block matrix, then the inverse of $\tilde{A}$ is given by $\tilde{B}$, where
\begin{align*}
\tilde{B} := \left(\begin{array}{cccccc}
  \I   & B^{(1,1)}& B^{(1,2)} & \cdots & \cdots & B^{(1,k)}  \\
   0   &   \I    &B^{(2,2)}&       &        & \vdots \\
\vdots & \ddots  & \ddots  & \ddots &  & \vdots \\
\vdots &         &    0    &   \I   &B^{(s-1,s-1)}&B^{(s-1,s)} \\
\vdots &         &         &    0   &   \I   &B^{(s,s)}\\
   0   & \cdots  & \cdots  & \cdots &    0   &   \I
\end{array}\right)
\end{align*}
and $B^{(i,j)} = -A_{(i,j)}-\sum_{t=i}^{j-1} A^{(i, t)}B^{(t+1,j)}$

\end{lemma}

\begin{proof}
When we try to multiply the $j$th block column of $\tilde{B}$ with the $i$th block row of $\tilde{A}$ we get for $j > i$ the term $B^{(i,j-1)} + (\sum_{t=i}^{j-2} A^{(i,t)} B^{(t+1,j-1)}) + A^{(i,j-1)}$. For $i=j$ the result is the identity and for $j<i$ the result is a zero matrix.
Hence by setting $B^{(i,j-1)} = -A^{(i,j-1)}-\sum_{t=i}^{j-2} A^{(i,t)} B^{(t+1,j-1)}$ we obtain the inverse.

Note that we have no circular dependency here, i.e. we can compute $\tilde{B}$ bottom up, starting with the lowest block $B^{(j-1,j-1)}$ for every column $j$.
\end{proof}

Given that a product and an inverse of an upper triangular block matrix can be expressed via matrix products of smaller $n \times n$ matrices, we obtain the following corollary:

\begin{corollary}\label{cor:bitlength}
Let $\tilde{A}$ and $\tilde{B}$ be upper triangular block matrices, then the entries of the inverse of $\tilde{A}$ and the entries of the product $\tilde{A}\tilde{B}$ increase in their bit-length (compared to the bit-length of $\tilde{A}$ and $\tilde{B}$) by a factor of at most $O(\text{poly}(s) \cdot \text{polylog}(n))$.
\end{corollary}

\begin{proof}[Proof of \Cref{thm:bitlength}]
We give the proof of \Cref{thm:bitlength} for the algorithm of \Cref{thm:elementUpdate}, since the proof works analogously for all other algorithms.

The algorithm maintains the inverse of some matrix $A^{(t)}$ via a chain of transformation matrices, i.e. $(A^{(t)})^{-1} = (T^{(t_1,t)})^{-1}(T^{(t_1,t_2)})^{-1} (A^{(t_2)})^{-1}$ for some $t_2 \le t_1 \le t$.

Such a transformation matrix is given via $T^{(t',t)} = \I+(A^{(t')})^{-1}(A^{(t)}-A^{(t')})$, so if $A$ is an upper triangular block matrix consisting of $s$ block per row/column, then $T^{(t',t)}$ and $(T^{(t',t)})^{-1}$ have entries, whose bit-length is larger by a factor of at most poly$(s)\cdot$ polylog$(n)$ compared to $A$, and the transformation matrices and their inverses are also upper triangular block matrices. (\Cref{lem:upperTriangularInverse}, \Cref{cor:bitlength}). Thus, if the entries of the matrix $A$ have polylog$(n)$ bit-length, all transformation matrices have entries with bit-length poly$(s)\cdot$ polylog$(n)$ throughout our computations, and every computation performed by the algorithm uses numbers that have at most $O(\text{poly}(s)\cdot\text{polylog}(n))$ bit-length.

\end{proof}

For the largest eigenvector reduction, we have $s = \text{polylog}(n)$, so the bitlength does not influence the time complexity of the algorithm besides of polylog factors.


\input{appendix_lowerbound.tex}

%% file: applications.tex
\section{Applications}
\label{sec:applications}


\begin{figure}[H]
\center
\includegraphics[trim={250 110 250 100},clip,scale=0.75]{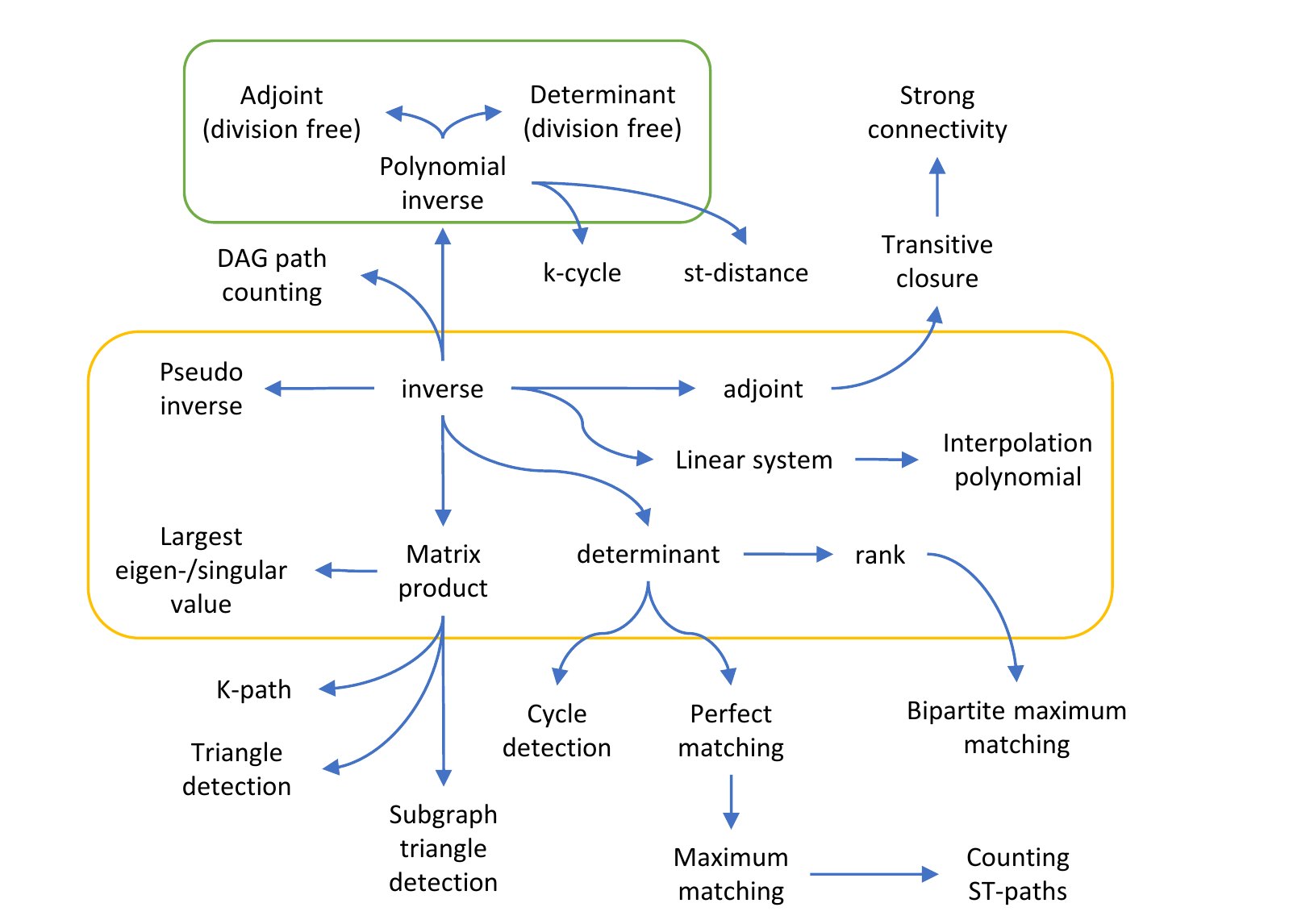}
\caption{\label{fig:reductionOverview}
Problems that can be solved via dynamic matrix inverse. An arrow $A \rightarrow B$ mean that an algorithm for $A$ can be used to solve $B$. The yellow box contains algebraic applications proven in subsection \ref{sub:algebraicApplications} while the green box contains reductions from subsection \ref{sub:polynomialMatrices}. All remaining reductions for graph problems are stated in subsection \ref{sub:graphApplications}. An overview where the reductions can be found in this section is given in \Cref{fig:algebraicReductionsOverview} (yellow and green box) and \Cref{fig:graphReductionsOverview} (remaining reductions).
}
\end{figure}


There is a wide range of applications, which we summarized in \Cref{tbl:applications,tbl:applications2}. The chain of reductions for these applications is displayed in \Cref{fig:reductionOverview}.
%
We split the applications into two categories: algebraic applications (subsection \ref{sub:algebraicApplications}) and graph applications (subsection \ref{sub:graphApplications}). The algebraic applications have their complexities measured in field operations, as they work over any field. For most graph applications we use finite fields of bit-length $O(\log n)$, and in the standard model arithmetic operations for these fields require $O(1)$ time, so we can give the required complexities as required time.

Most of graph applications are done via already existing reductions, but the problem they reduce to are not just the dynamic matrix inverse, but a variety of algebraic problems such as dynamic rank or determinant. Hence we will start this chapter by showing how any dynamic matrix inverse algorithm can be used for a multitude of other dynamic algebraic problems.

\subsection{Algebraic Applications}
\label{sub:algebraicApplications}

\input{algebraic_reductions_table.tex}

First we want to state that a dynamic matrix product of \emph{any} length can be solved using dynamic matrix inverse.

\begin{theorem}[inverse $\leftrightarrow$ matrix product]\label{thm:matrixProductReduction}
Let $A_1...A_s$ be matrices, where $A_i$ is of size $n_i \times m_i$ define $k = \sum_i n_i + m_i$.

Let $\mathcal{I}$ be a dynamic matrix inverse algorithm, that can maintain the inverse of an $n \times n$ matrix, which is promised to stay non-singular throughout the updates. Assume algorithm $\mathcal{I}$ requires $O(p(n))$ field operations for the pre-processing, $O(u(n))$ operations for updates and $O(q(n))$ operations for queries.

Then there exists an algorithm $\mathcal{P}$ for dynamic matrix product supporting the same type of updates/queries (e.g. row, column or element) as $\mathcal{A}$ using $O(u(k))$ and $O(q(k))$ operations respectively. The pre-processing requires $O(p(k))$ field operations.

\end{theorem}

\begin{proof}
The proof is based on the following simple observation. If we combine the matrices $A_1...A_s$ to the following matrix $A$:
\begin{align*}
\left(\begin{array}{cccccc}
  \I   &   A_1  &    0   & \cdots & \cdots &    0   \\
   0   &   \I   &   A_2  &    0   &        & \vdots \\
\vdots & \ddots & \ddots & \ddots & \ddots & \vdots \\
\vdots &        &    0   &   \I   & A_{s-1}&    0   \\
\vdots &        &        &    0   &   \I   &   A_s  \\
   0   & \cdots & \cdots & \cdots &    0   &   \I
\end{array}\right)
\end{align*}
Then the inverse of the matrix $A$ is:
\begin{align*}
\left(\begin{array}{cccccc}
  \I   &   -A_1  & A_1A_2 & -A_1A_2A_3 & \cdots & \prod_{i=1}^s -A_i \\
   0   &   \I   &   -A_2  & A_2A_3 & -A_2A_3A_4       & \vdots \\
\vdots & \ddots & \ddots & \ddots & \ddots & \vdots \\
\vdots &        &    0   &   \I   & -A_{s-1}& A_{s-1}A_s \\
\vdots &        &        &    0   &   \I   &   -A_s  \\
   0   & \cdots & \cdots & \cdots &    0   &   \I
\end{array}\right)
\end{align*}
Which means we can query any element of any consecutive sub-product of $A_1,...,A_s$ by querying elements of the inverse.
\end{proof}

For the static setting we know that matrix inverse and determinant are equivalent. For the dynamic setting we can prove the following analogue. To our knowledge this is the first black-box reduction between dynamic matrix inverse and dynamic determinant. Previously dynamic determinant was obtained via dynamic matrix inverse via a white-box reduction in \cite{Sankowski04}.

\begin{theorem}[inverse $\rightarrow$ determinant]
Let $A$ be an $n \times n$ matrix over some field.
Let $\mathcal{I}$ be a dynamic matrix inverse algorithm, that can maintain the inverse of an $n \times n$ matrix, which is promised to stay non-singular throughout the updates. Assume algorithm $\mathcal{I}$ requires $O(p(n))$ field operations for the pre-processing, $O(u(n))$ operations for updates and $O(q(n))$ operations for queries.

\paragraph{Element update}

If $\mathcal{I}$ supports element updates and element queries,
then there exists a dynamic determinant algorithm that supports element updates in $O(u(n)+q(n))$ operations. The pre-processing requires $O(p(n)+n^\omega)$ operations.

\paragraph{Column update}

If $\mathcal{I}$ supports column updates and row queries,
then there exists a dynamic determinant algorithm that supports column updates in $O(u(n)+q(n))$ operations. The pre-processing requires $O(p(n)+n^\omega)$ operations.

Note: The condition that the matrix stays non-singular can be removed. Using \Cref{thm:singularUpdates} (from \cite{Sankowski07}), one can extend dynamic matrix inverse/determinant algorithms to also work on singular matrices.
\end{theorem}

\begin{proof}
The reduction is based on the determinant lemma:
\begin{align}\det(A+uv^\top ) = \det(A) (1 + v^\top A^{-1}u).\label{eq:determinantLemma}\end{align}
Equation \eqref{eq:determinantLemma} tells us that the determinant after some update changes by a factor of $(1 + v^\top A^{-1}u)$. 
In case of element updates, the vectors $u$ and $v$ are just scaled unit-vectors, i.e. $v^\top A^{-1}u$ can be computed by querying a single entry of $A^{-1}$. Likewise for column updates, we have that $v$ is a unit-vector, so $v^\top A^{-1}u$ cna be computed using a single row query.

During the pre-processing we simply compute the initial determinant, which requires $n^\omega$ field operations and then we track the changes via equation \eqref{eq:determinantLemma}.
\end{proof}

\begin{theorem}[inverse $\leftarrow$ determinant]\label{thm:inverseDeterminantEquivalence}
Let $A$ be an $n \times n$ matrix over some field.
Let $\mathcal{D}$ be a dynamic matrix determinant algorithm, that can maintain the determinant of an $n \times n$ matrix, which is promised to stay non-singular throughout the updates. Assume algorithm $\mathcal{I}$ requires $O(p(n))$ field operations for the pre-processing and $O(u(n))$ operations for element updates.

Then there exists a dynamic matrix inverse algorithm, that can maintain the inverse of an $n \times n$ matrix, which is promised to stay non-singular throughout the updates. The algorithm supports element updates and element queries in $O(u(n))$ and the pre-processing requires $O(p(n))$ operations.

\end{theorem}

\begin{proof}
The reduction is again based on equation \eqref{eq:determinantLemma}: The multiplicative change of the determinant is $(1 + v^\top A^{-1}u)$ which contains information about $A^{-1}$.

Let's say we want to obtain $A^{-1}_{i,j}$, then we add the value $1$ to the entry $A_{i,j}$, i.e. we set $u = e_i$ and $v = e_j$ in equation \eqref{eq:determinantLemma}. Let $d$ be the value of the determinant before this update and $d'$ the determinant after this update, then obtain the entry of the inverse via $A^{-1}_{i,j} = d'/d -1$. By performing this update, the internal state of the dynamic determinant algorithm has changed, but in total at most $O(u(n))$ bit were changed, so we can revert these changes in another extra $O(u(n))$ time.
\end{proof}

Note that this equivalence between inverse and determinant is also true for adjoint and determinant, because for invertible matrixes $\adj(A) = \det(A) A^{-1}$ so equation \eqref{eq:determinantLemma} becomes
$$\det(A+uv^\top ) = \det(A) (1 + v^\top A^{-1}u) = \det(A) + v^\top \adj(A) u .$$
So similar to how the reductions between inverse and determinant require to track the multiplicative change of the determinant, we can also obtain reductions between adjoint and determinant by tracking the additive change.

For invertible matrices this also means that maintaining inverse and adjoint is equivalent. For instance in order to maintain the adjoint using a dynamic matrix inverse algorithm we simply use \Cref{thm:inverseDeterminantEquivalence} to also maintain the determinant and then for a query we simply multiply the output by the maintained determinant. For the converse, when maintaining the inverse via the adjoint, we simply divide the result by the determinant.

\begin{corollary}[inverse $\leftrightarrow$ adjoint]\label{cor:inverseAdjointEquivalence}
Element update/element query (and column update/row query) dynamic matrix inverse and element update/element query (and column update/row query) dynamic matrix adjoint are equivalent, if the matrix stays non-singular throughout the updates.
\end{corollary}

Same as with determinant and adjoint, the following application of solving linear systems was already observed in \cite{Sankowski04} as a white-box reduction. Here we give a black-box reduction from dynamic linear system to dynamic matrix inverse.

\begin{theorem}[inverse $\rightarrow$ linear system]\label{thm:linearSystemReduction}
Let $\mathcal{I}$ be a dynamic matrix inverse algorithm, that can maintain the inverse of an $n \times n$ matrix, which is promised to stay non-singular throughout the updates. Assume algorithm $\mathcal{I}$ requires $O(p(n))$ field operations for the pre-processing, $O(u(n))$ operations for updates and $O(q(n))$ operations for queries.

Then there exists a dynamic linear system algorithm, that can maintain for an $n \times n$ matrix $A$ and a matrix $M$ of size at most $n \times n$ the product $A^{-1}M$, when $A$ is promised to stay non-singular throughout the updates. The algorithm supports the same type of updates to $A$ and $M$ and queries to $A^{-1}M$ as $\mathcal{I}$ in $O(u(n))$ and $O(q(n))$ respectively. The pre-processing requires $O(p(n))$ operations.
\end{theorem}

The converse, i.e. solving dynamic inverse problem via dynamic linear system is trivially given for element queries, by solving $A^{-1}e_i$. Thus element update/query dynamic matrix inverse and dynamic linear system solver are equivalent in the dynamic setting. This is an interesting observation, because the hardness of solving a linear system in the static setting is not yet proven.

\begin{proof}
The proof is a simple implication of the inverse of block matrices:
\begin{align*}
\left(\begin{array}{cc}
Q & R\\
S & T
\end{array}\right)^{-1}
=
\left(\begin{array}{cc}
Q^{-1} + Q^{-1}R(T-SQ^{-1}R)^{-1}SQ^{-1} & -Q^{-1}R(T-SQ^{-1}R)^{-1} \\
-(T-SQ^{-1}R)^{-1}SQ^{-1} & (T-SQ^{-1}R)^{-1}
\end{array}\right)
\end{align*}
By setting $Q = \I$, $R = 0$, $S = M$ and $T = A$, the lower left block becomes $-A^{-1}M$ and we can thus maintain this product via the dynamic matrix inverse algorithm. To maintain $MA^{-1}$ we can use the same trick to have this product in the upper right block.
\end{proof}

Using our column-update row-query algorithm from \Cref{thm:columnUpdate}, the reduction \Cref{thm:linearSystemReduction} implies that we get a dynamic linear system algorithm, which can maintain the solution $x$ of $Ax = b$ explicitly, while updates replace entire constraints of the system. This is useful for other problems that can be solved via linear systems, for example interpolation polynomials can be constructed via linear systems.

\begin{theorem}\label{thm:interpolationPolynomial}
Let $\mathcal{L}$ be a dynamic linear system algorithm, that can maintain the solution $x$ of $Ax = b$, where $A$ is a non-singular $n \times n$ matrix. Assume algorithm $\mathcal{L}$ requires $O(p(n))$ field operations for the pre-processing, $O(u(n))$ operations for updates.

Then there exists a dynamic algorithm $\mathcal{P}$ that can maintain a degree $n-1$ interpolation polynomial (i.e. a vector containing the $n$ coefficients), interpolating upto $n$ points. The supported updates are adding/removing or moving points. The update and query complexity of $\mathcal{P}$ is the same as the complexity of $\mathcal{L}$.
\end{theorem}

\begin{proof}
Let $p := \sum_{i=0}^{n-1} p_i X^i$ be the interpolation polynomial interpolating upto $n$ points $(x_i, y_i)$ for $i=1,...,k \le n$. Each point induces a constraint $\sum_{j=0}^{n-1} p_j \cdot x_i^j = y_j$. When $k < n$ we add $n-k$ constraints of the form $a_j = 0$ for $j=k+1,...,n$. Every polynomial $p$ satisfying these constraints is a valid interpolation polynomial, so we simply maintain the solution to this system. Adding/removing or moving a point means we have to change one of the constraints.
\end{proof}

\begin{theorem}[inverse $\rightarrow$ pseudo-inverse]\label{thm:pseudoInverseReduction}
For $m > n$ a $m\times n$ matrix $A$ of rank $n$ the pseudo inverse $A^+$ is $(A^\top  A)^{-1}A^\top $.

Let $\mathcal{I}$ be a dynamic matrix inverse algorithm, that can maintain the inverse of an $n \times n$ matrix, which is promised to stay non-singular throughout the updates. Assume algorithm $\mathcal{I}$ requires $O(p(n))$ field operations for the pre-processing, $O(u(n))$ operations for element updates and $O(q(n))$ operations for queries.

Then there exists a dynamic pseudo inverse algorithm, that can maintain the pseudo inverse of a $m \times n$ matrix, which is promised to stay of rank $n$ throughout the updates. The algorithm supports the row scaling updates in $O(u(m))$ and the same type of queries as $\mathcal{I}$ in $O(q(m))$. The pre-processing requires $O(p(m))$ operations.

\end{theorem}

\begin{proof}
The proof works in a similar way as for \Cref{thm:linearSystemReduction}. We set $Q = \I$, $R = A$, $S = A^\top $ and $T = 0$. Then the lower left block of the inverse is $(A^\top  Q^{-1} A)^{-1}A^\top  Q^{-1}$. When scaling a row of $A$, we instead change a single entry of $Q$.
\end{proof}

\begin{theorem}[matrix product $\rightarrow$ largest singular value]\label{thm:eigenValueVectorReduction}

If we are only interested in the largest eigenvalue (or largest singular value), we have the following reduction.
\begin{itemize}
\item Let $\mathcal{I}$ be a dynamic matrix product algorithm, that can maintain the product of $k$ many $n \times n$ matrices. Assume algorithm $\mathcal{I}$ requires $O(p(n,k))$ field operations for the pre-processing, $O(u(n,k))$ operations for updates and $O(q(n,k))$ operations for element queries.

Then there exists a dynamic algorithm $\mathcal{E}$, that can maintain the largest absolute value of the eigenvalues with error 
$\varepsilon$
 for symmetric matrices over an ordered field, when $k = O(\varepsilon^{-1} \cdot \log n\varepsilon^{-1})$. The algorithm supports the same type of updates as $\mathcal{I}$ in $O(k\cdot u(n,k)+q(n,k))$ operations. The pre-processing requires $O(p(n,k))$ operations.
\end{itemize}
If we are interested in a vector $v$ such that $(v^\top A v)/v^\top v \ge (1-\varepsilon) \lambda_1$, we have another reduction:
\begin{itemize}
\item If $\mathcal{I}$ supports row queries in $O(q(n,k))$, then $\mathcal{E}$ supports querying the vector $v$ in $O(q(n,k))$.
\end{itemize}
\end{theorem}

There exist matrices with $O(\log n)$ bit-length entries, but whose inverse has $\Omega(n)$ bit-length for some of its entries. We prove in the appendix \ref{app:bitlength}, that our inverse algorithm does not perform computations with such very small/large values, when using the reduction from \Cref{thm:eigenValueVectorReduction} for $k = \text{polylog n}$ and if the input has values of $\text{polylog(n)}$ bit-length. Thus the bound on the number of operations is also a bound on the runtime, when ignoring polylog factors.

\begin{proof}[Proof of \Cref{thm:eigenValueVectorReduction}]
The theorem is a simple implication of \Cref{thm:matrixProductReduction} by applying the power method. We simply maintain $v^\top A^k v$ and $v^\top  A^k$ for a random vector $v$ sampled over the unit sphere. If we are interested in the largest singular value we instead maintain $v^\top (A^\top A)^k v$, though this only allows for element updates to $A$. 
\end{proof}

So far we only handled cases, where the matrix is promised to stay non-singular. Next we state a surprising result from \cite{Sankowski07}, where we can maintain the rank of singular matrices.

\begin{theorem}[{\cite[Corollary 4.1]{Sankowski07} determinant $\rightarrow$ rank}]\label{thm:rankReduction}
Let $\mathcal{D}$ be a dynamic determinant algorithm, that can maintain the determinant of an $n \times n$ matrix, which is promised to stay non-singular throughout the updates. Assume algorithm $\mathcal{D}$ requires $O(p(n))$ field operations for the pre-processing and $O(u(n))$ operations for element updates.

Then there exists a randomized dynamic rank algorithm, that can maintain (w.h.p) the rank of an $n \times n$ matrix, i.e. the algorithm works if the matrix does becomes singular. The algorithm supports element updates in $O(u(3n))$ and the pre-processing requires $O(p(3n)+(3n)^\omega)$ operations.
\end{theorem}

We will give an outline of the reduction, as we have to argue why it works in our look-ahead setting. The reduction is adaptive, which means, depending on the change of the determinant during one update, the reduction has to perform a different update next. This means this reduction has no full look-ahead, however, the updates the reduction has to perform are not arbitrary and their structure can be exploited in a weaker look-ahead setting, which we defined in \Cref{sec:lookAhead}.

\begin{proof}[Proof-outline]
When we have to maintain the rank of some matrix $A$, we will maintain the determinant of the following matrix:
\begin{align*}
\tilde{A} = \left(
\begin{array}{ccc}
A & X & 0 \\
Y & 0 &\I \\
0 &\I &\I^{(k)}
\end{array}
\right)
\end{align*}
Here the matrices $X$ and $Y$ have independent and uniformly at random chosen entries from some field extension $F^{ex}$ of $F$, where $|F^{ex}| = \Omega(n^2)$ and the matrix $\I^{(k)}$ is the identity matrix, where only the first $k$ diagonal entries are 1, all other entries are zero. During the pre-processing we compute the rank of $A$ using $O(n^\omega)$ field operations and set $k = n - \text{rank}(A)$.

The matrix $\tilde{A}$ is full rank w.h.p iff rank$(A) \ge n - k$ \cite[Lemma 4.1, Theorem 4.1]{Sankowski07}, so if the determinant should become zero during an update to $A$, we simply revert the update and increase $k$ by one (which means performing an update to the $\I^{(k)}$ block). If an update to the $A$ did not result in a zero determinant, then we try to decrease $k$ by one. If decreasing $k$ leads to a zero determinant, then we revert this change again.

Note that reverting updates is not done by performing a new update, as the used algorithm might only work for full-rank matrices. Instead, reverting is done by reverting all the state/memory-changes performed in the data-structure.

\paragraph{Adaption for the look-ahead setting}

For our look-ahead application we will assume that after reverting a change, or when no update was performed to the $\I^{(k)}$ block, we simply perform an update to the $k$th column of $\I^{(k)}$, but this time we add 0 to the column, i.e. we do not actually change anything.

Using this assumption, we know that for any update to $A$, every update is followed by one update to $\I^{(k)}$ (possibly two updates, but one is reverted, so from the perspective of the dynamic determinant algorithm there is only one additional update). So for any $t$, when we perform $t$ updates to $A$, then we know every second update to $\tilde{A}$ needs to be in the block of $\I^{(k)}$. More specifically if $k'$ is the value of $k$ before performing the $t$ updates to $A$, then the $t$ column indices where we perform updates to $\I^{(k)}$ are all in $\{ k-t,..., k+t \}$. This is an important property, which allows our look-ahead algorithm from \Cref{sec:lookAhead} to work on this type of reduction.
\end{proof}

With the same technique Sankowski also proved, that we can maintain the determinant and inverse, when the matrix is allowed to become singular \cite[Theorem 4.1]{Sankowski07}. While the matrix $A$ is singular, the inverse algorithm can simply return "fail", the important part here is, that it returns correct results again, once the matrix becomes non-singular again after some update.

\begin{theorem}[{\cite[Theorem 4.1]{Sankowski07} non-singular case}]\label{thm:singularUpdates}
Let $\mathcal{I}$ ($\mathcal{D}$) be a dynamic matrix inverse (determinant) algorithm, that can maintain the inverse (determinant) of an $n \times n$ matrix, which is promised to stay non-singular throughout the updates. Assume algorithm $\mathcal{I}$ ($\mathcal{D}$) requires $O(p(n))$ field operations for the pre-processing, $O(u(n))$ operations for updates and $O(q(n))$ operations for queries.

Then there exists a randomized dynamic matrix inverse (determinant) algorithm, that can maintain (w.h.p) the inverse (determinant) of an $n \times n$ matrix, which may become singular after an update. The algorithm supports the same type of update/query operations and requires $O(p(3n))$ field operations for the pre-processing, $O(u(3n))$ operations for updates and $O(q(3n))$ operations for queries.

If the matrix is currently singular, queries will return ``fail".
\end{theorem}

\paragraph{Maintaining a Submatrix explicitly}

We will now explain how to extend a dynamic matrix inverse algorithm with slow query time, to support fast $O(1)$ queries for elements within some small area, i.e. for some set $H \subset [n]$, the elements in $((A^{(t)})^{-1})_{H,H}$ are maintained explicitly. This algorithm was already observed in \cite[Theorem 4]{Sankowski05}, here we re-state it as a black-box reduction.

Using reductions from subsection \ref{sub:graphApplications}, this explicit maintenance of a submatrix allows for new upper bounds for $ST$-reachability. Following the techniques from \cite{Sankowski05}, the explicit maintenance of a submatrix can also be used for a hitting set argument, when reducing from all-pairs-shortest-distances to dynamic matrix inverse, see \Cref{sub:polynomialMatrices}.

\begin{theorem}\label{thm:inverseSubmatrix}
Let $H \subset [n]$ and let $\mathcal{I}$ be a dynamic inverse algorithm supporting queries to elements $A^{-1}_{i,j}$ and partial rows $A^{-1}_{i,H}$ in $O(q(|H|, n))$ operations, assuming element update complexity $O(u(|H|, n))$. Then the algorithm $\mathcal{I}$ can be extended to maintain the submatrix $A^{-1}_{H,H}$ explicitly, where the new update time is given by $O(u(|H|, n) + q(|H|,n)+|H|^2)$.

\end{theorem}

\begin{algorithm}
\caption{MaintainSubmatrix (\Cref{thm:inverseSubmatrix})}\label{alg:inverseSubmatrix}
\begin{algorithmic}[1]
\REQUIRE An element update dynamic matrix inverse algorithm, that allows to query partial rows $(A^{(t)})^{-1}_{j,H}$ for some $H \subset [n]$
\renewcommand{\algorithmicensure}{\textbf{Maintain:}}
\ENSURE  $(A^{(t)})^{-1}_{H,H}$ explicitly
\renewcommand{\algorithmicensure}{\textbf{update operation:}}
\ENSURE (Update to $A^{(t)} = A^{(t-1)} + C$ in position $(i,j)$)
\renewcommand{\algorithmicensure}{\textsc{Update}$(C)$}
\ENSURE
\STATE $t \leftarrow t+1$
\STATE Let $i,j$ be the coordinates of the non-zero entry of $C$.
\STATE $T^{(t-1,t)} \leftarrow \I$
\STATE $T^{(t-1,t)}_{H\cup \{j\},j} \leftarrow T^{(t-1,t)}_{H\cup \{j\},j} + (A^{(t-1)})^{-1}_{H \cup \{j\},i}C_{i,j}$
\STATE $(T^{(t-1,t)})^{-1}_{H \cup \{j\}, [n]} \leftarrow \textsc{PartialInvert}(H\cup \{j\}, T^{(t-1,t)})$ (\Cref{alg:partialInvert}) \label{line:inverseSubmatrixInversion}
\STATE $(A^{(t)})^{-1}_{H,H} \leftarrow (T^{(t-1,t)})^{-1}_{H, H \cup \{j\}} (A^{(t-1)})^{-1}_{H\cup\{j\},H}$
\end{algorithmic}
\end{algorithm}

\begin{proof}
Let $A^{(t)}$ be the matrix after the $t$th update. We have to maintain $(A^{(t)})^{-1}_{H,H}$ explicitly, which will be done in the same way as any other dynamic matrix inverse algorithm presented in \Cref{sec:dynamicInverse} We compute a transformation matrix $T^{(t-1,t)}$ s.t. $A^{(t)} = A^{(t-1)}T^{(t-1,t)}$, which implies $(A^{(t)})^{-1} = (T^{(t-1,t)})^{-1}(A^{(t-1)})^{-1}$. Remember thanks to \Cref{lem:formulaForT}
\begin{align*}
T^{(t-1,t)} = \I + (A^{(t-1)})^{-1} (A^{(t)}-A^{(t-1)})
\end{align*}
When $A^{(t)}$ was changed in some column with index $j$ (so $(A^{(t-1)})^{-1} (A^{(t)}-A^{(t-1)})$ is zero in all columns except $j$), then $(T^{(t-1,t)})^{-1}$ is of the same structure $\I+C$, where $C$ is nonzero only in column $j$ (see \Cref{lem:Tinverse}). To compute $(A^{(t)})^{-1}_{H,H} = ((T^{(t-1,t)})^{-1})_{H,[n]}(A^{(t)})^{-1}_{[n],H}$ we need to know the rows $H$ of $(T^{(t-1,t)})^{-1}$, the submatrix $(A^{(t-1)})^{-1}_{H,H}$ and the partial row $(A^{(t-1)})^{-1}_{j,H}$ (see \Cref{lem:complexityTproduct}).
The submatrix $(A^{(t-1)})^{-1}_{H,H}$ is already known and $(A^{(t-1)})^{-1}_{j,H}$ can be queried, so only rows $H$ of $(T^{(t-1,t)})^{-1}$ are missing to compute $(A^{(t)})^{-1}_{H,H}$.

Computing rows $H$ of $(T^{(t-1,t)})^{-1}$ requires rows $H$ and row $j$ of $T^{(t)}$ (see \Cref{lem:Tinverse}), once we know them, computing rows $H \cup \{j\}$ of $(T^{(t-1,t)})^{-1}$ requires $O(|H|)$ operations, which will be subsumed by other terms. Getting the required rows $H$ and $j$ of $T^{(t-1,t)} = \I + (A^{(t-1)})^{-1} (A^{(t)}-A^{(t-1)})$ means we have to know entry $(A^{(t-1)})^{-1}_{j,i}$ and entries $(A^{(t-1)})^{-1}_{H,i}$. Note that $(A^{(t-1)})^{-1}_{H,i}$ is a partial column, but we can assume that the algorithm $\mathcal{I}$ also supports a partial column query, besides of a partial row query, by simply maintaining both $(A^{(t)})^{-1}$ and $(A^{(t)})^{\top -1}$.

The final computation of the submatrix $((T^{(t-1,t)})^{-1}(A^{(t-1)})^{-1})_{H,H}$ requires $O(|H|^2)$ field operations.
\end{proof}

For $|H| = O(n^{1-\mu})$, $0 \le \mu \le 1$, the algorithm from \Cref{thm:elementUpdate} supports partial row queries in $O(n^{\varepsilon_1+\varepsilon_2} + n^{\varepsilon_2+1-\mu})$, see \Cref{thm:TimpliesInverseElement}. 
We obtain the following corollary:

\begin{corollary}\label{cor:elementUpdateSubmatrix}
Let $0 \le \mu \le 1$ and $H \subset [n]$ of size $|H| = O(n^{1-\mu})$.
For every $0 \le \varepsilon_1 \le \varepsilon_2 \le 1$ there exists a dynamic algorithm for maintaining the inverse of an $n \times n$ matrix $A$, requiring $O(n^\omega)$ field operations during the pre-processing. The algorithm supports changing any entry of $A$ in $O(n^{\varepsilon_2 +\varepsilon_1} + n^{\omega(1,\varepsilon_1,\varepsilon_2)-\varepsilon_1} + n^{\omega(1,1,\varepsilon_2)-\varepsilon_2} + n^{\varepsilon_2+1-\mu}+n^{2-2\mu})$ field operations and querying any entry of $A^{-1}$ in $O(n^{\varepsilon_2 +\varepsilon_1})$ field operations, while querying entries in $(A^{-1})_{H,H}$ require only $O(1)$ operations.

For $1-\mu \le 0.55$, and current values of $\omega$, the update time is $O(n^\fastExponent)$, i.e. the same as \Cref{thm:elementUpdate}, but with faster queries for some submatrix. If $\omega = 2$, then the complexity is $O(n^{1.25})$ for $1-\mu \le 0.5$.
\end{corollary}

In \Cref{sec:applications} \Cref{thm:transitiveClosureReduction}, we will state a reduction from transitive closure to dynamic matrix inverse \cite{Sankowski04}. \Cref{cor:elementUpdateSubmatrix} together with that reduction imply a $O(n^{\fastExponent})$ algorithm for $ST$-reachability, where $|S|,|T| = O(\sqrt{n})$.

\subsection{Polynomial Matrices}\label{sub:polynomialMatrices}

So far we only analyzed the inverse of a matrix over some field, now we extend these results to polynomial matrices. This allows us to get further graph applications such as $k$-cycle detection (\Cref{thm:kCycleReduction}) and all-pairs-shortest-distances (\Cref{thm:shortestDistanceReduction}) and at the end of this subsection, we are also able to construct a division free dynamic determinant and dynamic adjoint algorithm (see \Cref{cor:ringAdjointDeterminant}) that works for matrices over rings instead of fields. The extension of our algorithms to work for polynomial matrices is not a black-box reduction, but should work for most dynamic matrix inverse algorithms as the main idea of the reduction is to use Strassen's tools for division free algorithms \cite{Strassen73} as used in  \cite{Sankowski05}.

The first main result of this subsection will be the following Theorem:
\begin{theorem}\label{cor:polyElementUpdate}
Let $\mathbf{R}$ be some ring and $m \in \mathbb{N}$. Let $A$ be an $n \times n$ matrix over $\mathbf{R}[X]/\langle X^{m} \rangle$ (the ring of polynomials modulo $X^m$).

Then all our dynamic matrix inverse algorithms (\Cref{thm:columnUpdate}, \Cref{thm:elementUpdate}, \Cref{thm:inverseSubmatrix}, \Cref{thm:elementLookAhead} and \Cref{thm:columnLookAhead}) can be extended to maintain the inverse of $\I-X \cdot A$ with updates to $A$.

The number of ring operations over $\mathbf{R}$ for pre-processing, updates and queries increase by a factor of $\tilde{O}(m)$.\footnote{Here $\tilde{O}$ hides polylog factors.}
\end{theorem}

\begin{proof}

Note that we want to maintain the inverse of a matrix $\I-X \cdot A$ in $(\mathbf{R}[X]/\langle X^{m} \rangle)^{n \times n}$, i.e. a polynomial matrix modulo $X^m$. The inverse of such a matrix is given by $(\I - X \cdot A)^{-1} = \sum_{k=0}^{m-1} X^k A^k$. To see this, simply multiply both sides with $\I - X \cdot A$, which yields $(\I-X \cdot A)\sum_{k=0}^{m-1} X^k A^k = \sum_{k=0}^{m-1} X^k A^k - \sum_{k=1}^{m} X^k A^k = \I \mod X^m$.

All field operations performed in our algorithms are expressed as matrix operations (i.e. matrix product or matrix inversion), so to prove that our algorithms work over the ring $\mathbf{R}[X]/\langle X^{m} \rangle$ instead of some field, we only have to check that all these matrix operations are well-defined. For matrix multiplication this is obviously true, additionally the number of required operations increases by only a factor of $\tilde{O}(m)$ when using polynomials instead of field elements, because the product of two polynomials of degree at most $m$ can be computed in $O(m \log m)$ using fast fourier transformations.

For the matrix inverse, we have to argue a bit more. We do not have any division operation available for polynomials, so beforehand it is not quite clear if the matrix inverses, which our algorithms try to compute, exist.

We claim that, if our algorithms maintain the inverse of $\I - X \cdot A \mod X^m$ and allow the updates to only be performed to $A$, then all matrices that our algorithms have to invert are of the form $\I - X \cdot M$ as well, where $M$ is some polynomial matrix over $\mathbf{R}[X]/\langle X^{m} \rangle$. Matrices of this form can be inverted since $(\I - X \cdot M)^{-1} = \sum_{k=0}^{m-1} M^k \mod X^m$.
We will check that all matrices, which our algorithms try to inverse, are of this form at the end of this proof.

Next, we have to verify that we can efficiently compute the inverse of matrices of the form $\I - X \cdot M$. We have $(\I - X \cdot M)^{-1} = \sum_{k=0}^{m-1} M^k = \prod_{k=1}^{\log m} (\I + M^{2^k}) \mod X^m$, so the inverse can be computed in $O(\log m)$ matrix multiplications. Given that all computations are performed modulo $X^m$ the degrees are always bounded by $X^{m-1}$, so each matrix product becomes slower by only $O(m \log m)$ operations. Hence in total we can invert any polynomial $n \times n$ matrix mod $X^m$ in $\tilde{O}(m n^\omega)$ time, so the matrix inversion becomes slower by a factor of $\tilde{O}(m)$.

By performing all matrix inversions in our algorithms this way, we can extend the algorithms to the polynomial setting, where we have to maintain the inverse of a polynomial matrix $\I - X \cdot A \mod X^m$, when the updates are performed to matrix $A$. The number of arithmetic operations performed by the algorithms increase by a factor of $\tilde{O}(m)$.

\paragraph{Verifying matrix inversions}

We have to verify, that all matrices, that our algorithms invert, are of the form $\I - X \cdot A$ for some polynomial matrix $A$. If the matrices are of this form, then we can invert them modulo $X^m$.

\paragraph{\Cref{alg:invert,alg:partialInvert}} This algorithm can invert matrices of the form $A = \I + X \cdot C$, because in line \ref{line:invertPolyinversion} we have $A_{J,J}^{-1} = \I - X \cdot (-C)_{J,J}$.

\paragraph{\Cref{alg:UpdateColumnsInverse}} If the input matrix $C$ is a multiple of $X$, then $T^{-1} C$ is a multiple of $X$, so $T = \I + T^{-1}C$ satisfies the condition to be inverted via \Cref{alg:invert} in line \ref{line:UpdateColumnsInverseInversion}.

\paragraph{\Cref{alg:UpdateInverse}} If the input $\Delta = C + R$ is a multiple of $X$, then \Cref{alg:UpdateColumnsInverse} can be used in line \ref{line:UpdateInverseInversion1} and \ref{line:UpdateInverseInversion2}.

\paragraph{\Cref{alg:MaintainTransform}} If the matrix $A^{(t)}$ is of the form $\I-X \cdot N^{(t)}$, then $A^{(t)}_{I^{(t)},I^{(t)}} - A^{(t-1)}_{I^{(t)},I^{(t)}}$ is a multiple of $X$ and \Cref{alg:UpdateInverse} can be used in line \ref{line:MaintainTransformInversion}.

\paragraph{\Cref{alg:CombinedTransformation}} If $T^{(0)}$ is of the form $\I - X \cdot N$, then it can be inverted during the initialization in line \ref{line:CombinedTransformationInversion1}. If the changes $C^{(t)}$ are a multiple of $X$ then $S^{(t)}$ is a multiple of $X$ and \Cref{alg:UpdateColumnsInverse} can be called during an update in line \ref{line:CombinedTransformationInversion2}. Also if $C^{(t)}$ is a multiple of $X$, then $\mathcal{T}^{(t)}$ is of the form $\I-X \cdot N^{(t)}$ for some matrix $N^{(t)}$ and \Cref{alg:MaintainTransform} can be called in line \ref{line:CombinedTransformationInversion3}.

\paragraph{\Cref{alg:ElementUpdate}} If $A^{(0)}$ is of the form $\I - X \cdot N$, then the matrix can be inverted during the pre-processing in line \ref{line:ElementUpdateInversion1}. We have $T^{(0)} = \I$ so \Cref{alg:CombinedTransformation} can be initialized in line \ref{line:ElementUpdateInversion2}. Likewise during an update we can re-initialize \Cref{alg:CombinedTransformation} in line \ref{line:ElementUpdateInversion3}. If the updates $C^{(t)}$ are multiples of $X$, then $S^{(t)}$ is a multiple of $X$ and \Cref{alg:UpdateColumnsInverse} can be called during an update in line \ref{line:ElementUpdateInversion3}. \Cref{alg:CombinedTransformation} can also be executed in line \ref{line:ElementUpdateInversion4}, if $C^{(t)}$ and thus $\tilde{C}^{(t)}$ is a multiple of $X$.

\paragraph{\Cref{alg:inverseSubmatrix}} If the change $A^{(t)}-A^{(t-1)}$ is a multiple of $X$, then $T^{(t)}$ is of the form $\I + X \cdot C$ and can thus be inverted in line \ref{line:inverseSubmatrixInversion}.

\paragraph{\Cref{alg:partialUpdateInverse}} If the input matrix $C$ is a multiple of $X$, then $M = I + \cdot T^{-1}C$ can be inverted via \Cref{alg:partialUpdateInverse} in line \ref{line:partialUpdateInverseInversion}.

\paragraph{\Cref{alg:lookAhead}} If the change $C^{(t)}$ is a multiple of $X$, then \Cref{alg:partialUpdateInverse} in line \ref{line:lookaheadInvervion} and \ref{line:lookaheadInvervion2} works.

\paragraph{\Cref{alg:CombinedTransformationLookAhead}} If $T^{(0)}$ is of the form $\I - X \cdot N$, then it can be inverted during the initialization in line 1. If the changes $C^{(t)}$ are a multiple of $X$ then $S^{(t)}$ is a multiple of $X$ and \Cref{alg:CombinedTransformation} can be called during an update in line \ref{line:CombinedTransformationLookAheadResetInversion}. Likewise, if $C^{(t)}$ is a multiple of $X$, then \Cref{alg:lookAhead} can be called in line \ref{line:CombinedTransformationLookAheadInversion}.

\end{proof}

The reduction from dynamic matrix inverse to dynamic determinant from \Cref{thm:inverseDeterminantEquivalence} via the identity $\det(M + u v^\top) = \det(M)(1 + v^\top M^{-1} u)$ still holds in the polynomial setting for $M = \I - X \cdot A$, so by maintaining the inverse modulo $X^m$, we can also maintain the determinant modulo $X^m$. The same is also true for the reduction from adjoint to inverse (\Cref{cor:inverseAdjointEquivalence}) via the identity $\adj(M) = \det(M)M^{-1}$.

\begin{corollary}\label{cor:polyElementUpdateDeterminant}
Let $\mathbf{R}$ be some ring and $m \in \mathbb{N}$. Let $A$ be an $n \times n$ matrix over $\mathbf{R}[X]/\langle X^{m} \rangle$ (the ring of polynomials modulo $X^m$).

Then our dynamic matrix inverse algorithms (\Cref{thm:columnUpdate}, \Cref{thm:elementUpdate}, \Cref{thm:elementLookAhead} and \Cref{thm:columnLookAhead}) can be used to maintain the determinant (or adjoint) of $\I-X \cdot A$, supporting updates to $A$.

The number of ring operations over $\mathbf{R}$ for pre-processing, updates and queries increase by a factor of $\tilde{O}(m)$.
\end{corollary}

In \cite{Sankowski05} Sankowski explained how to extend the techniques from \cite{Strassen73} to dynamically maintain the adjoint and the determinant without using divisions, resulting in an algorithm that can be used on rings instead of fields.

For this observe that $\det(\I-X\cdot(\I-A)) = \det(A)$ for $X=1$, so we only have to maintain the determinant of the polynomial matrix $\I-X\cdot(\I-A)$ and evaluate it for $X=1$ after every update. The matrix $\I-X\cdot(\I-A)$ is of the form $\I-X\cdot A'$, so we can use the algorithm from \Cref{cor:polyElementUpdate} to maintain the determinant of the polynomial matrix modulo some $X^m$. Since the polynomial $\det(\I-X\cdot(\I-A))$ is of degree at most $n$, we choose to run our algorithm modulo $X^{n+1}$, so the runtime of our matrix inverse/determinant algorithms increase by a factor of $\tilde{O}(m) = \tilde{O}(n)$.

Similarly we have $\adj(A) = \adj(\I-X\cdot(\I-A))$ for $X=1$ and $\adj(\I-X\cdot(\I-A)) = \det(\I-X\cdot(\I-A)) \cdot (I-X\cdot(\I-A))^{-1} \mod X^{n}$, so we can also maintain the adjoint of $A$ by maintaining the inverse and determinant of the polynomial matrix $\I-X\cdot(\I-A)$.

\begin{corollary}\label{cor:ringAdjointDeterminant}
Let $\mathbf{R}$ be a ring, then the algorithms from \Cref{thm:columnUpdate}, \Cref{thm:elementUpdate}, \Cref{thm:elementLookAhead} and \Cref{thm:columnLookAhead} can be extended to be division-free and to maintain the determinant and adjoint of an matrix $A \in \mathbf{R}^{n \times n}$. The required ring operations for pre-processing, updates and queries increase by a factor of $\tilde{O}(n)$.
\end{corollary}

\subsection{Graph Applications}
\label{sub:graphApplications}

\input{graph_reductions_table.tex}

Most applications of the dynamic matrix inverse for graphs are well-known, and our improvements to these graph problems are direct implications of our improvements to the dynamic matrix inverse.

For all non-bipartite graph applications we have the following equivalences between the type of updates/queries. Here $\mathcal{I}$ is a dynamic matrix inverse (or adjoint or determinant) algorithm and $\mathcal{G}$ is a dynamic algorithm for some graph problem.
\begin{align}
\text{\begin{tabular}{l|l}
Operation of $\mathcal{I}$ & Operation of $\mathcal{G}$ \\
\hline
element update & edge update \\
column update & incoming edges node update \\
row update & outgoing edges node update \\
element query & node pair query \\
column query & target query \\
row query & source query
\end{tabular}}
\label{eqn:updateTypes}
\end{align}

The following theorem will give us the application for the transitive closure problem. We want to note the interesting property that transitive closure with node updates (restricted to incoming edges) and source queries can be done in $O(n^\slowExponent)$ but the combination of having node updates restricted to outgoing edges and source queries have a $\Omega(n^{2-\varepsilon})$ for all $\varepsilon > 0$ lower bound \cite{HenzingerKNS15}.

\begin{theorem}[{\cite[Theorem 6 and 7]{Sankowski04} adjoint $\rightarrow$ transitive closure}]
\label{thm:transitiveClosureReduction}
Let $\mathcal{A}$ be a dynamic adjoint algorithm, that can maintain the adjoint of an $n \times n$ matrix, which is promised to stay non-singular throughout the updates. Assume algorithm $\mathcal{A}$ requires $O(p(n))$ field operations for the pre-processing, $O(u(n))$ operations for updates and $O(q(n))$ operations for queries.

Then there exists a randomized dynamic algorithm $\mathcal{T}$ for transitive closure on graphs with $n$ nodes, with $O(u(n))$ update time, $O(q(n))$ query time and $O(p(n))$ pre-processing time.

The update and query type of $\mathcal{T}$ depend on the update and query type of $\mathcal{I}$ as in \eqref{eqn:updateTypes}.
\end{theorem}

The matrix constructed in the reduction \Cref{thm:transitiveClosureReduction} is invertible w.h.p., further the reachability information is encoded in the non-zero entries of the adjoint. Since for non-singular matrices the adjoint is just the product of inverse and determinant, the inverse has the same non-zero entries as the adjoint. Thus we can also just use the dynamic matrix inverse algorithm directly instead of using a dynamic matrix adjoint algorithm from \Cref{cor:inverseAdjointEquivalence}.

\begin{corollary}[{single-source-reachability $\rightarrow$ strong connectivity}]\label{cor:strongReachabilityReduction}
\footnote{We thank Adam Karczmarz for pointing out this application of dynamic matrix inverse.}
Let $\mathcal{S}$ be a dynamic single-source-reachability algorithm, that can maintain the reachability of a fixed source node $s$ to all other $n$ nodes in the graph. Assume algorithm $\mathcal{S}$ requires $O(p(n))$ time for the pre-processing, $O(u(n))$ time for edge updates and $O(q(n))$ time to query the reachability of $s$.

Then there exists a dynamic algorithm $\mathcal{C}$ for strong connectivity on graphs with $n$ nodes, with $O(u(n)+q(n))$ edge update time and $O(p(n))$ pre-processing time.

\end{corollary}

\begin{proof}
The algorithm $\mathcal{C}$ works as follows: Fix some arbitrary node $v$ and denote this node to be the source. We run the algorithm $\mathcal{S}$ twice where for the second copy the direction of all edges is reverted.
We can check the strong connectivity by checking if every node can reach $v$ and $v$ can reach every other node.
\end{proof}

The proof for the following reduction is omitted here and can be found in \cite[Theorem 5]{Sankowski04}.

\begin{theorem}[{\cite[Theorem 5]{Sankowski04} inverse $\rightarrow$ DAG path counting}]\label{thm:DAGcounting}
Let $\mathcal{I}$ be a dynamic matrix inverse algorithm, that can maintain the inverse of an $n \times n$ matrix, which is promised to stay non-singular throughout the updates. Assume algorithm $\mathcal{I}$ requires $O(p(n))$ field operations for the pre-processing, $O(u(n))$ operations for updates and $O(q(n))$ operations for queries.

Then there exists a dynamic algorithm $\mathcal{D}$ for counting paths in a DAG with $n$ nodes, which is promised to stay acyclic throughout the updates. The algorithm requires $O(u(n))$ arithmetic operations per update, $O(q(n))$ operations for queries and $O(p(n))$ operations for the pre-processing.

The update and query type of $\mathcal{D}$ depend on the update and query type of $\mathcal{I}$ as in \eqref{eqn:updateTypes}.
\end{theorem}

\begin{theorem}[{\cite[Corollary 1]{Sankowski04} determinant $\rightarrow$ spanning tree counting}]\label{thm:spanningTreeReduction}
Let $\mathcal{D}$ be a dynamic determinant algorithm, that can maintain the determinant of an $n \times n$ matrix undergoing element updates, such that the matrix is promised to stay non-singular throughout the updates. Assume algorithm $\mathcal{D}$ requires $O(p(n))$ field operations for the pre-processing, $O(u(n))$ operations for updates.

Then there exists a dynamic algorithm $\mathcal{S}$ for counting spanning trees in undirected graphs with $n$ nodes, which are promised to stay connected throughout the updates. The algorithm requires $O(u(n))$ arithmetic operations per  edge update and $O(p(n))$ operations for the pre-processing.
\end{theorem}

For triangle detection there are two well understood settings: edge update and node update. For the edge update there exists a conditional $\Omega(n^{1-\varepsilon})$ $\forall \varepsilon > 0$ lower bound \cite{HenzingerKNS15}\footnote{This follows from the $\Omega(n^{2-\varepsilon)}$ lower bound for node update variant.} and a trivial matching upper bound. For the case of node updates (i.e. changing all adjacent edges) the results are similar, except that the lower and upper bound are $\Omega(n^{2-\varepsilon})$ $\forall \varepsilon > 0$ and $O(n^2)$.

These node update lower bounds hold only when we allow all edges of the node to be changed. By restricting node updates we are able to get a faster update time than the $\Omega(n^{2-\varepsilon})$ lower bound. We consider the following restricted node updates:
\begin{itemize}
\item \emph{Subgraph} In the subgraph setting, nodes can be "turned on/off". We now consider subgraph triangle detection, were nodes can be "turned on/off" but the adjacent edges can not be changed.
\item \emph{Incoming edges} We consider a directed graph and during a node update, we are only allowed to change the incoming edges of a node, but not its outgoing ones.
\end{itemize}

\begin{theorem}[5-matrix product $\rightarrow$ subgraph triangle detection]\label{thm:subgraphTriangleDetectionReduction}
Assume there exists a dynamic algorithm that can maintain the product of 5 matrices of size $n \times n$, supporting element updates and element queries using $O(u(n))$ and $O(q(n))$ operations respectively and the pre-processing requires $O(p(n))$ operations.

Then there exists a dynamic subgraph triangle detection algorithm, that can detect whether there exists a triangle in an $n$ node graph, supporting node updates in $O(u(n)+q(n))$ update time. The pre-processing time is $O(p(n)+n\cdot q(n))$.
\end{theorem}

\begin{proof}
Let $A$ be the adjacency matrix of the graph and $D$ be a diagonal matrix, then we use the dynamic matrix product algorithm to maintain the product $ADADA$. During the pre-processing we set $D_{ii}$ to 1 for all turned on nodes and 0 otherwise. We then sum the entries on the diagonal, that correspond to turned on nodes, of the product $ADADA$, which is exactly the number of triangles in our graph times 3.

When during an update some node $i$ is turned on, we set the $i$th entry on the diagonal of the matrices $D$ to be 1. Next, we add the $i$-th entry on the diagonal of the product to our current counter of triangles.

When turning a node off, we first subtract the $i$-th entry on the diagonal of the product from our counter and then set the $i$th entry on the diagonal of the $D$ matrices to be 0. 

This way our counter is always the number of triangles times 3. After performing an update and changing the counter, we simply have to output, whether the counter is nonzero.
\end{proof}

\begin{theorem}[3-matrix product $\rightarrow$ directed node update triangle detection]\label{thm:nodeUpdateTriangleDetectionReduction}
Assume there exists a dynamic algorithm that can maintain the product of 3 matrices of size $n \times n$, supporting column updates and element queries using $O(u(n))$ and $O(q(n))$ operations respectively and the pre-processing requires $O(p(n))$ operations.

Then there exists a dynamic triangle detection algorithm, that can detect whether there exists a triangle in an $n$ node graph, supporting node updates where only the incoming edges are changed in $O(u(n)+q(n))$ update time. The pre-processing time is $O(p(n)+n\cdot q(n))$.

As the the value of the counter is bounded by $O(n^3)$, we can use a finite field of size $\Omega(n^3)$, so all field operations can be performed in $O(1)$ time.
\end{theorem}

\begin{proof}
Let $A$ be the adjacency matrix of the graph, then we use the dynamic matrix product algorithm to maintain the product $A^3$. During the pre-processing we sum all the diagonal entries of the product and obtain the number of triangles in the graph times 3.

When performing an update to node $i$, we first query the $i$-th diagonal entry of the product and subtract it from our triangle counter. We then update the adjacency matrices by performing column updates and at the end we query the $i$-th diagonal entry again and add it to our counter. If the counter is nonzero, then the graph has a triangle.

As the the value of the counter is bounded by $O(n^3)$, we can use a finite field of size $\Omega(n^3)$, so all field operations can be performed in $O(1)$ time.
\end{proof}

The proofs we give for the next reductions base heavily on the Schwartz-Zippel Lemma \cite{Schwartz80,Zippel79}, which allows us to efficiently test if a polynomial is the zero-polynomial by evaluating it on randomly sampled inputs. The following formulation of the lemma is taken from \cite{Sankowski05}:
\begin{lemma}[Schwartz-Zippel Lemma \cite{Schwartz80,Zippel79,Sankowski05}]\label{lem:schwartzZippel}
If $p(x_1,...x_m)$ is a non-zero polynomial of degree $d$ with coefficients in a field and $S$ is a subset of the field, then the probability that $p$ evaluates to 0 on a random element $(s_1,s_2,...,s_m) \in S^m$ is at most $d / |S|$.
\end{lemma}

So if for a given polynomial $p$ we sample the input uniformly at random from some finite field $\mathbb{Z}_p$, then the probability of having $p(x_1,...,x_n) = 0$, is at most $n/p$. We can choose $p \sim n^c$ for some constant $c>1$ to get an error probability of $O(n^{1-c})$, while the cost of performing arithmetic operations in that field is bounded by $O((c \log n)^2)$.

\begin{theorem}[determinant $\rightarrow$ cycle detection]\label{thm:cycleDetectionReduction}
Let $\mathcal{D}$ be a dynamic determinant algorithm, that can maintain the determinant of an $n \times n$ matrix, which is promised to stay non-singular throughout the updates. Assume algorithm $\mathcal{D}$ requires $O(p(n))$ field operations for the pre-processing and $O(u(n))$ operations for updates.

Then there exists a randomized dynamic cycle detection algorithm $\mathcal{C}$, that can detect if an $n$ node graph contains a cycle. The algorithm has $O(u(n))$ update time and $O(p(n))$ pre-processing time.

The update type of $\mathcal{C}$ depends on the update type of $\mathcal{D}$ as in \eqref{eqn:updateTypes}.
\end{theorem}

\begin{proof}[Proof of \Cref{thm:cycleDetectionReduction}]
We maintain the determinant of $A+\I$, where $A$ is the adjacency matrix of the graph. Then the determinant $\det(A+\I) = \sum_\sigma \text{sign}(\sigma) \prod_{i=1}^n (A+\I)_{i,\sigma(i)}$ encodes all the valid cycle covers of the graph, when we also allow cycles of the form $(v,v)$ for $v \in V$. This is because the product is nonzero for only those permutations $\sigma$, which represent a cycle cover in the graph. Note that this determinant is of the form $1 + p((A_{i,j})_{i,j})$ where $p((A_{i,j})_{i,j})$ is a polynomial in the entries of $A$ that does not have any constant terms. If we set $A_{i,j}$ to 0 if no edge exists (i.e. every 1 in the adjacency matrix is now some variable $A_{i,j}$), then $p(A)$, then the polynomial $det(A+\I)-1$ is a zero polynomial if and only if the graph contains no cycle, i.e. the only valid cycle cover consists of only self-cycles of the form $(v,v)$. Via Schwartz-Zippel (\Cref{lem:schwartzZippel}) we can test this property with high probability by choosing uniformly chosen elements of a finite field as the non-zero elements of $A$, i.e. when adding an edge $(i,j)$ set $A_{i,j}$ to be a uniformly chosen random number.

Also note that such a random matrix $A+\I$ is invertible w.h.p, because $det(A+\I)$ is another non-zero polynomial in $A$, so Schwartz-Zippel (\Cref{lem:schwartzZippel}) applies. Hence it is enough to use a dynamic determinant algorithm that works only on non-singular matrices.
\end{proof}

\begin{theorem}[rank $\rightarrow$ bipartite matching]\label{thm:bipartiteMatchingReduction}
Let $\mathcal{R}$ be a dynamic matrix rank algorithm, that can maintain the rank of an $n \times n$ matrix. Assume algorithm $\mathcal{R}$ requires $O(p(n))$ field operations for the pre-processing and $O(u(n))$ operations for updates.

Then there exists a randomized dynamic algorithm $\mathcal{M}$ for biparite maximum cardinality matching, that maintains the size of the largest matching in a bipartite graph with $n$ left and $n$ right nodes. The algorithm has $O(u(n))$ update time and $O(p(n))$ pre-processing time.

The update type of $\mathcal{M}$ depend on the update type of $\mathcal{R}$:
\begin{center}
\begin{tabular}{l|l}
Operation of $\mathcal{R}$ & Operation of $\mathcal{M}$ \\
\hline
element update & edge update \\
column update & right node update \\
row update & left node update
\end{tabular}
\end{center}

\end{theorem}

\begin{proof}
The bipartite graph is represented as a symbolic matrix $M$ (i.e. entries can consist of variables) where $M_{i,j} = X_{i,j}$ if left node $i$ and right node $j$ are connected via an edge, otherwise we have $M_{i,j} = 0$.

We claim that the rank of this matrix $M$ is exactly the size of the maximum cardinality matching. Proof: The rank of $M$ is the size of a largest subset of rows and columns $I,J \subset [n]$, such that the submatrix $M_{I,J}$ is of full rank. This submatrix is of full rank if and only if the determinant is nonzero. Note that the determinant of any $m \times m$ matrix $A$ is $\det(A) = \sum_\sigma \text{sign}(\sigma) \prod_{i=1}^m A_{i,\sigma(i)}$. Here the permutation $\sigma$ can be interpreted as a matching of the nodes and the product $\prod_{i=1}^m A_{i,\sigma(i)}$ is nonzero if the matching is valid, i.e. the edges $(i,\sigma(i))$ exist. Hence $\det(M_{I,J})$ is a nonzero polynomial if and only if the nodes $I$ and $J$ can be perfectly matched.

If we replace the variables $X_{i,j}$ with uniformly chosen random numbers from some finite field $\mathbb{Z}_p$, then the probability of the determinant evaluating to zero, even though it is a nonzero polynomial, is at most $n/p$. (\Cref{lem:schwartzZippel}).

So we can choose $p \sim n^c$ for some constant $c$ to get an error probability of $n^{1-c}$ while the size of each field element is bounded by $c \log n$.

\end{proof}

In \cite{Sankowski07} a similar reduction was proven for matching on general matrices. As pointed out in \cite{Sankowski07}, this general matching application implies an upper bound for node disjoint $ST$-paths counting via standard reduction, see for instance \cite{MulmuleyVV87}.

\begin{theorem}[{\cite[Corollary 1]{Sankowski04} determinant $\rightarrow$ perfect matching}]
\label{thm:perfectMatchingReduction}

Let $\mathcal{D}$ be a dynamic determinant algorithm, that can maintain the determinant of an $n \times n$ matrix, which is promised to stay non-singular throughout the updates. Assume algorithm $\mathcal{D}$ requires $O(p(n))$ field operations for the pre-processing and supports element updates in $O(u(n))$ operations.

Then there exists a randomized dynamic algorithm $\mathcal{P}$ for detecting if a perfect matching exists in an $n$ node graph. The algorithm supports edge updates in $O(u(n))$ time and requires $O(p(n))$ pre-processing time. 

\end{theorem}

\begin{theorem}[{\cite[Theorem 2.1]{Sankowski07} perfect matching $\rightarrow$ matching}]\label{thm:generalGraphMatchingReduction}

Let $\mathcal{P}$ be a dynamic perfect matching algorithm, that can detect if a perfect matching exists in an $n$ node graph. Assume algorithm $\mathcal{P}$ requires $O(p(n))$ pre-processing time and supports edge updates in $O(u(n))$ time.

Then there exists a randomized dynamic maximal cardinality matching algorithm $\mathcal{M}$, that can maintain the size of the largest matching in an $n$ node graph. Algorithm $\mathcal{M}$ requires $O(p(3n))$ pre-processing time and supports edge updates in $O(u(3n))$ update time.

\end{theorem}

\begin{theorem}[{\cite{MulmuleyVV87} matching $\rightarrow$ counting $ST$-paths}]\label{thm:stPathCountingReduction}

If there exists a dynamic maximum cardinality matching algorithm $\mathcal{M}$, which maintains the size of the matching, with pre-processing time $O(p(n))$ and edge update time $O(u(n))$, then there exists a dynamic counting vertex disjoint $ST$-path algorithm $\mathcal{P}$ with the same pre-processing and update time.

\end{theorem}

Next we want to prove reductions for problems from parameterized complexity theory. Our algorithm for dynamic matrix inverse can be used to obtain new upper bounds on $k$-path and $k$-cycle detection.

\begin{theorem}[$k$-matrix product $\rightarrow$ $k$-path]\label{thm:kPathReduction}
There exists some function $f(k)$ such that:

Assume there exists a dynamic algorithm that can maintain the product of $k$ matrices of size $n \times n$, supporting element updates and element queries using $O(u(k,n))$ and $O(q(k,n))$ operations respectively and the pre-processing requires $O(p(k,n))$ operations.

Then there exists a dynamic $k$-path algorithm $\mathcal{P}$ with pre-processing time $O(f(k) \cdot p(k,n))$, edge update time $O(f(k) \cdot u(k,n))$ and pair queries, which answer if there exists a path of length $k$ between two nodes, in time $O(f(k) q(k,n))$.

\end{theorem}

\begin{proof}

A $k$-matrix-product allows us to check if there exists a $k$-walk between two nodes. Now we simply use color coding to extend this to $k$-paths:

Assign each node independently a uniformly chosen color from $1..k$. Create a $k$-layered graph where the $i$th layer consists of the nodes with color $i$. We add all edges that go from nodes of color $i$ to nodes of color $i+1$. With some small but nonzero probability, that depends only on $k$, a pair of nodes $(u,v)$, that have a $k$-path between them in the original graph, will also have a $k$-path between them the new graph. If there is no $k$-path between the nodes, then there will also be no $k$-path in the new graph.

This means we simply keep/maintain many (some $f(k)$) of these randomized graphs in parallel, and then at least one of them will have a valid $k$-path, if one exists in the original graph.

\end{proof}

The next applications are based on maintaining the inverse/determinant of a polynomial matrix of the form $(\I-X \cdot A)^{-1} \in \mathbf{R}[X]/\langle X^m \rangle$ for any $m \in \N$ and any ring $\mathbf{R}$. In subsection \ref{sub:polynomialMatrices} we described how to extend matrix inverse algorithm to work on such matrices.

\begin{theorem}[polynomial matrix determinant $\rightarrow$ $k$-cycle detection]\label{thm:kCycleReduction}
There exists some function $f(k)$ such that:

Let $\mathcal{D}$ be a dynamic determinant algorithm, that maintains the determinant of $(\I-X \cdot A)^{-1} \in (\mathbf{R}[X]/\langle X^k \rangle)^{n \times n}$ with pre-processing in $O(p(k,n))$ operations (over $\mathbf{R}$) and element updates to $A$ in $O(u(k,n))$ operations (over $\mathbf{R}$).

Then there exists a dynamic $k$-cycle algorithm $\mathcal{C}$, which detects if an $n$ node graph contains a node disjoint cycle of length $k$, with pre-processing time $O(f(k) \cdot p(k,n))$ and edge update time $O(f(k) \cdot (u(k,n)))$.

\end{theorem}

\begin{proof}

The high-level idea of the reduction is to use color coding on the graph and to then detect a cycle in the color coded graph.

Given a graph, we assign each node a independent and uniformly at random chosen color from $\{1,...,k\}$. Then we will only keep those edges that go from nodes of color $i$ to $i+1$ or from color $k$ to color $1$. If there exists a $k$ cycle in the original graph, then with a small but positive probability that depends only on $k$, the cycle still exists in the color coded graph. However, if no $k$-cycle exists in the original graph, then the color coded cycle will also have no $k$-cycle.

Now if we would simply use a cycle detection algorithm (e.g. via \Cref{thm:cycleDetectionReduction}), we might get a false positive by detecting a cycle of length $s\cdot k$ for some $s > 1$. So we have to make sure that our cycle detection algorithm can only detect cycles of length upto $k$. This is where the polynomial matrix comes into play.

The determinant of the matrix $\I - X \cdot A$, where $A$ is the adjacency matrix of the color coded graph, encodes information about the cycles. We have $\det(\I-X\cdot A) = \sum_\sigma \text{sign}(\sigma) \prod_{i=1}^n (\I-X\cdot A)_{i,\sigma(i)}$, where $\prod_{i=1}^n (\I-X\cdot A)_{i,\sigma(i)}$ is non-zero if and only if the cycle decomposition of $\sigma$ represents valid cycles in the graph represented by adjacency matrix $A$. Note that the degree of $\prod_{i=1}^n (\I-X\cdot A)_{i,\sigma(i)}$ is exactly the number of used edges, and since the color coded graph has no cycles of length less than $k$, the degree of $\prod_{i=1}^n A_{i,\sigma(i)}$ is $k$ if and only if the color coded graph has a $k$ cycle.
The only issue now is that maybe some terms cancel because of the sum and $\text{sign}(\sigma)$. However, using Schwartz-Zippel \Cref{lem:schwartzZippel}, we can prove that if we set the entries of $A$ to be random values over some field $\mathbb{Z}_{p}$ where $p \sim n^c$ for some large enough constant $c$, then w.h.p. the $k$th monomial of $det(\I-X \cdot A)$ will be nonzero if and only if there exists no $k$-cycle in the color coded graph.

Now we simply maintain some exponential number $f(k)$ of independently color coded graphs to detect a $k$-cycle w.h.p.

\end{proof}

The following Theorem is based on Sankowski's all-pair-distance algorithm \cite[Theorem 8]{Sankowski05}.
There the algorithm has a update time, but allows to quickly query the distance between any pair.
Here we try to balance update and query time to obtain a faster $st$-distance algorithm.
In \Cref{tbl:applications} we list this application as \emph{all-pair-distances} to better compare our result with \cite[Theorem 8]{Sankowski05}.
Note that an $st$-distance algorithm can also be used to query the distance between any other pair of nodes $u,v$ by adding edges $(s,u)$ and $(v,t)$.

\begin{theorem}[polynomial matrix inverse $\rightarrow$ $st$-distance]\label{thm:shortestDistanceReduction}

For any $0 \le \mu \le 1$ there exists a $k = \tilde{O}(n^{1-\mu})$, such that for a uniformly at random chosen $H \subset [n]$, $|H| = k$ the following is true:

Let $\F = \Z_p$ be some field of size $p = \Theta(n^c)$ and let $\mathcal{I}$ be a dynamic matrix inverse algorithm, that maintains the inverse of $(\I-X \cdot A)^{-1} \in (\F[X]/\langle X^{n^\mu} \rangle)^{n \times n}$.

Assume $\mathcal{I}$ requires $O(p(n^\mu,n))$ operations (over $\F$) for the pre-processing and supports element updates to $A$ in $O(u(n^\mu,n))$ operations (over $\F$). Assume further that $\mathcal{I}$ supports element queries to the inverse and queries to any partial row $A^{-1}_{i,H}$ in $O(q(n^\mu,n))$ operations (over $\F$).

Then there exists a randomized dynamic algorithm $\mathcal{S}$ for the $st$-distance problem on a graph with $n$ nodes.
The algorithm has pre-processing time $O(p(n^\mu, n))$ and edge update time $\tilde{O}(u(n^\mu,n)+q(n^\mu, n)+n^{2-\mu})$.

Here $\tilde{O}$ hides polylog factors.

\end{theorem}

\begin{proof}
We will start the proof by explaining how to maintain short distances inside our graph. We then extend the result to larger distances.
\paragraph{Short distances}

Assume for now $\F = \mathbb{Z}$, even though it is not a field and let $A$ be
the adjacency matrix of the graph.
We already know for $m = n^\mu$ that
$(\I-X \cdot A)^{-1} = \sum_{k=0}^{m} X^k A^k \mod X^m$, which means the
smallest monomial at entry $(s,t)$ will give us the distance between nodes
$s$ and $t$, if the shortest distance is less than $m$.

The coefficient of this monomial is the number of such shortest paths, so the
coefficient could require a bit-length of $\Omega(n)$, which would result in
an increased complexity.

\paragraph{Bounding the bit-length}

We actually have a finite field $\F = \Z_p$ with $p=\Theta(n^c)$
elements for some constant $c > 3$, so we only require $O(\log n)$ bit to
represent the elements.

We now choose for each non-zero entry of $A$ an independently and uniformly at
random chosen element from $\mathbb{Z}_p$, then the $d$th monomial of $(\I-X \cdot A)^{-1}_{i,j}$ is non-zero with probability $1/n^{c-1}$, if and only if there exists a walk from $i$ to $j$ of length $d$ (see Schwatz-Zippel \Cref{lem:schwartzZippel}).

Via union bound we have all the distances less than $m$ encoded in
$(\I-X\cdot A)^{-1} \mod X^{m}$ with probability at least $1/n^{c-3}$.

\paragraph{Distances larger than $m$}

So far we only know the distances less than $m$.
For $m$ large enough but $m \ll n$, this is already enough to compute the
shortest distances w.h.p using the following property due to Ullman and
Yannakakis \cite[Lemma 2.2]{UllmanY91}:\smallskip

\emph{If we choose $H \subset V$ to be a uniformly at random chosen subset of the vertices, then the probability that a given (acyclic) path has a sequence of more than $(cn \log n)/|H|$ vertices, none of which are in $H$, is, for sufficiently large $n$, bounded by $2^{1-\alpha c}$ for some positive $\alpha$.}\smallskip

So for $m = n^\mu$ and $|H| = \tilde{O}(n^{1-\mu})$ the shortest paths will
use w.h.p at most $m$ nodes not from $H$ in sequence.
This means a shortest path from $s$ to $t$ with length larger than $m$ can be
decomposed into segments $s \to h_1, h_1 \to h_2, ... h_k \to t$ where $h_i \in H$ for $i=1,...,k$.

Let $D_{i,j}$ be the distance matrix obtained from the smallest degree monomials of $(\I-A)^{-1}$.
Now consider a new graph with vertices $H \cup \{ s,t \}$, where an edge $(i,j)$ has weight $D_{i,j}$ (if an entry $(i,j)$ of $(A^{-1})_{H,H}$ is 0, then the edge $(i,j)$ does not exists in this new graph).
We can compute the shortest distance from $s$ to $t$ in that new graph in $O(n^{2(1-\mu)})$ time using Dijkstra's algorithm.

For this construction we require the submatrix $(\I-A)^{-1}_{H,H}$, $(\I-A)^{-1}_{s,H}$ and $(\I-A)^{-1}_{H,t}$.
The submatrix $(\I-A)^{-1}_{s,H}$ can be obtained from
\Cref{thm:inverseSubmatrix}, while $(\I-A)^{-1}_{s,H}$ and
$(\I-A)^{-1}_{H,t}$ can be obtained because we are able to query partial rows
of the inverse (which also means we can get partial columns by simply maintaining the transposed inverse in parallel).

Via \cite{UllmanY91} we now know, that we can get (w.h.p) the shortest distance between $s$ and $t$ in the original graph via the distance in the new smaller graph. The update time is $\tilde{O}(u(n^\mu,n)+q(n^\mu,n)+n^{2-\mu})$.

\end{proof}

%% file: algebraic_reductions_table.tex
\begin{figure}
\center
\small

\begin{tabular}{ll|ll}
\hline 
\multicolumn{2}{c|}{inverse} & \multicolumn{2}{c}{polynomial inverse}\tabularnewline
\hline 
matrix product  & \Cref{thm:matrixProductReduction}  & adjoint (division free)  & \Cref{cor:ringAdjointDeterminant} \tabularnewline
determinant  & \Cref{thm:inverseDeterminantEquivalence}  & determinant (division free)  & \Cref{cor:ringAdjointDeterminant} \tabularnewline
adjoint  & \Cref{cor:inverseAdjointEquivalence}  &  & \tabularnewline
\cline{3-4} 
linear system  & \Cref{thm:linearSystemReduction}  & \multicolumn{2}{l}{linear system}\tabularnewline
\cline{3-4} 
pseudo inverse  & \Cref{thm:pseudoInverseReduction}  & interpolation polynomial  & \Cref{thm:interpolationPolynomial} \tabularnewline
polynomial inverse  & \Cref{cor:polyElementUpdate}  &  & \tabularnewline
 & \multicolumn{1}{c|}{} &  & \multicolumn{1}{c}{}\tabularnewline
\hline 
\multicolumn{2}{c|}{determinant} & \multicolumn{2}{c}{matrix product}\tabularnewline
\hline 
rank  & \Cref{thm:rankReduction}  & largest singular/eigen-value  & \Cref{thm:eigenValueVectorReduction} \tabularnewline
\end{tabular}

\caption{\label{fig:algebraicReductionsOverview}
This table lists all applications from subsection \ref{sub:algebraicApplications} and \ref{sub:polynomialMatrices}.
}
\end{figure}

%% file: graph_reductions_table.tex
\begin{figure}
\center
\small


\begin{tabular}{ll|ll}
\hline 
\multicolumn{2}{c|}{determinant} & \multicolumn{2}{c}{matrix product}\tabularnewline
perfect matching  & \Cref{thm:perfectMatchingReduction}  & subgraph triangle detection  & \Cref{thm:subgraphTriangleDetectionReduction} \tabularnewline
cycle detection  & \Cref{thm:cycleDetectionReduction}  & triangle detection  & \Cref{thm:nodeUpdateTriangleDetectionReduction} \tabularnewline
spanning tree counting & \Cref{thm:spanningTreeReduction} & k-path  & \Cref{thm:kPathReduction} \tabularnewline
 & \multicolumn{1}{c}{} &  & \multicolumn{1}{c}{}\tabularnewline
\hline 
\multicolumn{2}{c|}{adjoint} & \multicolumn{2}{c}{polynomial inverse}\tabularnewline
\hline 
transitive closure  & \Cref{thm:transitiveClosureReduction}  & k-cycle  & \Cref{thm:kCycleReduction} \tabularnewline
 &  & st-distance  & \Cref{thm:shortestDistanceReduction} \tabularnewline
 & \multicolumn{1}{l}{} &  & \tabularnewline
\hline 
\multicolumn{2}{c|}{rank} & \multicolumn{2}{c}{perfect matching}\tabularnewline
\hline 
maximum bipartite matching  & \Cref{thm:bipartiteMatchingReduction}  & maximum matching  & \Cref{cor:strongReachabilityReduction} \tabularnewline
 & \multicolumn{1}{l}{} &  & \tabularnewline
\hline 
\multicolumn{2}{c|}{maximum matching} & \multicolumn{2}{c}{inverse}\tabularnewline
\hline 
counting ST-paths  & \Cref{thm:stPathCountingReduction}  & DAG path counting  & \Cref{thm:DAGcounting} \tabularnewline
\end{tabular}

\caption{\label{fig:graphReductionsOverview}
This table lists all applications from subsection \ref{sub:graphApplications}.
}
\end{figure}

%% file: appendix_lowerbound.tex
\section{Lower Bounds}

\subsection{Lower Bounds assuming $o(n^\omega)$ pre-processing}
\label{sub:trivialLowerBounds}

\begin{theorem}\label{thm:worstCaseOmegaUpdate}
Any dynamic algorithm for triangle detection, perfect matching,
matrix rank or matrix determinant requires $\Omega(n^\omega)$ update time,
if the pre-processing time is $o(n^\omega)$,
assuming the static version requires $\Omega(n^\omega)$ time.\footnote{
For triangle detection this assumption is called \emph{strong triangle conjecture} \cite{AbboudW14}.
For straight line programs there exists a reduction from matrix product to determinant.}
\end{theorem}

\begin{proof}
Assume there exists a dynamic algorithm for one of the above problems with $o(n^\omega)$ pre-processing and update time.
When given some input for the static problem, we change a single entry
(e.g. on edge or one entry of the matrix)
and perform the pre-processing of the dynamic algorithm on this modified input.
Then we perform one update with the dynamic algorithm to revert this change of the input.
After this update, we now know the solution to the static problem for the given input,
even though we spent only $o(n^\omega)$ time in total.
\end{proof}

\Cref{thm:worstCaseOmegaUpdate} shows that for some problems the $o(n^\omega)$ pre-processing time assumption does not allow for any dynamic algorithms with non-trivial worst-case update time.

\subsection{Implications of Mv and OMv results}
\label{sub:implicationOfPreviousResults}

For our lower bounds we removed all intricacies of the dynamic matrix inverse
and reduced the problem to matrix-matrix and matrix-vector products.
Consequently, many previous results for the Mv and OMv problems can also be applied to our new problems.
Here we discuss the impact of these results on our conjectures.

\paragraph{Exploiting matrix-vector product algorithms}

\begin{theorem}\label{thm:matrixVectorAppendix}
There exists an algorithm that can beat the trivial time of
\Cref{def:hintedOMv,def:2hintedOuMv,def:doubleHintedOMv}:
During the first phases no computations are performed
(only the input is read/saved)
and the last phase is improved by a $\log n$ factor.
\end{theorem}

The last phase of our problems
\Cref{def:hintedOMv,def:2hintedOuMv,def:doubleHintedOMv} requires
to output the result of a boolean matrix-vector product.
So any algorithm that can compute a boolean matrix-vector product
faster than trivial, will lead to an improvement in that phase.
The current fastest matrix-vector algorithm leads to \Cref{thm:matrixVectorAppendix},
however, it does not break our conjectures, as the improvement
is only polynomial.

\begin{proof}[Proof of \Cref{thm:matrixVectorAppendix}]
In \cite{Williams07} Williams shows how to multiply
an $n$-dimensional vector with an $n \times n$ matrix in
$O(n^{2} / (\varepsilon \log n))$ time, after pre-processing the matrix in $O(n^{2+\varepsilon})$ time, for any $\varepsilon \in (0,0.5)$.

This algorithm can also be used for rectangular matrices,
by simply splitting the matrix into smaller square matrices.
Thus there exists an algorithm that can beat the trivial time of
\Cref{def:hintedOMv,def:2hintedOuMv,def:doubleHintedOMv}:
During the first phases we simply read/save the input.
In the second last phase, we pre-process the matrix
and in the last phase we compute the matrix-vector product.
\end{proof}

The result from \cite{Williams07} was later improved in \cite{LarsenW-soda17}
to be able to compute matrix-vector products in $O(n^{2}/ 2^{O(\sqrt{n})})$.
This result, too, does not break our conjecture:
(i) The improvement over the trivial time is sub-polynomial.
(ii) The bound is amortized over $2^{O(\sqrt{n})}$ matrix-vector products,
so it does not apply to our worst-case conjecture.
(iii) The algorithm \cite{LarsenW-soda17} can be extended to worst-case time,
but then the pre-processing becomes exponential,
i.e. $O(\exp(n^\varepsilon))$ for some $\varepsilon > 0$.

\paragraph{Reductions to/from Mv}

In the Mv problem we are given an $n \times n$ boolean matrix $M$ and
polynomial pre-processing time.
After the pre-processing, we are given a vector $v$ and need to
return the product $Mv$.
It is conjectured that despite the initial pre-processing of $M$,
computing $Mv$ must use $\Omega(n^{2-o(1)})$ time.

For the first of our problems (\Cref{def:hintedOMv}) there exists a reduction to Mv,
because both are given polynomial pre-processing time of the matrix.

\begin{theorem}
If there exists an algorithm that breaks the $Mv$-conjecture, then there exists an algorithm that breaks \Cref{con:hintedOMv}.
\end{theorem}

\begin{proof}
Let $M,V$ denote the matrices given in \Cref{def:hintedOMv}.
Then we can split the $n \times n^\tau$ matrix $M$ into
$k = O(n^{1-\tau})$ many $n^\tau \times n^\tau$ matrices
$M^{(1)},...,M^{(k)}$.
We initialize $k$ copies of the $Mv$-algorithm on these matrices.
During the phase where $V$ is given, we just save $V$,
but do not perform any computations.
In the last phase, where we are given an index $i$ and must compute $Mv$,
where $v$ is the $i$th column vector of $V$,
we give the $i$-th column vector $v$ of $V$ to each of the $k$ many Mv algorithms.
The resulting vectors $M^{(1)}v, ..., M^{(k)}v$ can be combined to the vector $Mv$.

The time for computing the product is
$O(k n^{2\tau-\varepsilon}) = O(n^{1+\tau-\varepsilon})$
for some constant $\varepsilon > 0$, because we assume that the
$Mv$-algorithm breaks the Mv-conjecture,
so it can compute each $M^{(i)}v$ in $O(n^{2\tau-\varepsilon})$ time.
\end{proof}

For our other problems
\Cref{def:2hintedOuMv,def:doubleHintedOMv} no such
reduction to the $Mv$-problem exists, because the
matrices used during the query phase of
\Cref{def:2hintedOuMv,def:doubleHintedOMv} are not
known from the very start.
Thus we can not offer polynomial pre-processing time
to the $Mv$-algorithm.

\paragraph{Exploiting cell-probe algorithms}

In the cell probe model, computations are free and
only reading/writing from/to the memory has cost.
The motivation for this is that information theoretic
tools can better analyze this model, since only
reading/saving information has a cost.
Lower bounds for the cell probe model also hold for
the standard Word-RAM model.

\begin{theorem}\label{thm:cellProbeAppendix}
There exists an algorithm that can beat the trivial time of
\Cref{def:hintedOMv,def:2hintedOuMv,def:doubleHintedOMv}
in the cell-probe model, by a polynomial factor.
\end{theorem}

This implies that our conjectures do not hold in the cell-probe model.

\begin{proof}[Proof of \cref{thm:cellProbeAppendix}]

In \cite{ChakrabortyKL-stoc18} Chakraborty, Kamma and
Larsen present a cell-probe algorithm that, after
reading an $n \times n$ matrix $M$, saves the matrix
together with additional $O(n^{3/2})$ bits.
Then, when receiving a vector $v$, the algorithm can
compute $Mv$ in $O(n^{3/2})$ cell probe operations.

By splitting a rectangular matrix into smaller square matrices,
this algorithm can also be used for our problems from
\Cref{def:hintedOMv,def:2hintedOuMv,def:doubleHintedOMv}
to speedup the last query phase in the cell-probe model.

During the first phases, the algorithm simply stores the input.
In the second last phase, the algorithm pre-processes the matrix.
(After reading the inputs, we know the restricted matrices
$N_{[n],I}$ and $N_{I,J}$ from
\Cref{def:2hintedOuMv} and \Cref{def:doubleHintedOMv} respectively.)
In the last phase (the query phase), the matrix-vector product is computed -
faster than the dimension of the matrix by a polynomial factor.

\end{proof}

\paragraph{Reductions to/from OMv}

The OMv problem is similar to the Mv problem.
The only difference is that the query phase,
where a vector $v$ is given and the product $Mv$
must be computed, is repeated several times.
This last phase is online in the sense,
that the next vector is only given after the
matrix-vector product for the previous vector
was returned.

There exist no reduction to/from our problems
(\Cref{def:hintedOMv,def:2hintedOuMv,def:doubleHintedOMv})
to the OMv problem.
While the conjectures are similar in structure
as they both describe matrix-vector products,
the amortized nature of OMv does not allow a reduction
from \Cref{def:hintedOMv,def:2hintedOuMv,def:doubleHintedOMv}.
Likewise, the online nature of OMv
(i.e. the input vectors are not known ahead of time)
does not allow reduction to
\Cref{def:hintedOMv,def:2hintedOuMv,def:doubleHintedOMv},
where hints for the vectors are provided.

\subsection{OMv-based Lower Bounds}

In \Cref{tbl:applications,tbl:applications2} we state many lower bounds based on the OMv conjecture, that are not actually stated in \cite{HenzingerKNS15}, but instead very simple observations. For completeness sake we will give short proofs for these lower bounds.

\begin{theorem}\label{thm:columnRowUpdateLB}
Assuming the OMv conjecture \cite{HenzingerKNS15}, any dynamic algorithm with polynomial pre-processing time that maintains a $k$-matrix product ($k \ge 3$), linear system, determinant, inverse, adjoint or rank, while supporting both row and column updates (and element queries, when there is more than value to be maintained), requires $\Omega(n^{3-\varepsilon})$ time for every constant $\varepsilon > 0$ for $O(n)$ updates and queries.

\end{theorem}

\begin{proof}

The lower bounds are based on the OuMv conjecture, which states that after polynomial pre-processing of some matrix $M$, answering $u^\top Mv$ for $n$ pairs vectors requires $\Omega(n^{3-\varepsilon})$ time for all constant $\varepsilon > 0$, if the next vector pair is only given after answering $u^\top Mv$ for the previous one. The conjecture assumes the computation to be over the boolean semi-ring. The use of the boolean semi-ring is no constraints, as we can perform the computations over some finite field $\mathbb{Z}_{p}$ instead, where $p > n$.

The lower bound obviously applies to $k$-matrix product by letting $k=3$. Simply perform one row update to set $v$ and one row update to set $u$. Then query $u^\top Mv$.

For dynamic matrix inverse and dynamic adjoint the same lower bound is obtained via the reduction from $k$-matrix product  (\Cref{thm:matrixProductReduction} and \Cref{cor:inverseAdjointEquivalence}).

The lower bound for dynamic determinant is implied by the reduction from dynamic matrix inverse of \Cref{thm:inverseDeterminantEquivalence}.

A dynamic linear system algorithm, which maintains $A^{-1}b$, can query entries of the inverse by setting $b = e_i$ and then querying an entry of $A^{-1}b$.

The reduction to dynamic rank is a bit longer. From the uMv conjecture we get a $\Omega(n^{2-\varepsilon})$ lower bound for \streach{} \cite{HenzingerKNS15} with node updates, which in turn gives the same lower bound for dynamic perfect bipartite matching with left and right node updates \cite{AbboudW14}. Bipartite matching with left and right node updates can be done via dynamic rank with row and column updates (\Cref{thm:bipartiteMatchingReduction}).

\end{proof}

\begin{theorem}\label{thm:graphLB}
Assuming the OMv conjecture \cite{HenzingerKNS15}, the following dynamic graph problems require $\Omega(n^{2-\varepsilon})$ time per node update ($\Omega(n^{1-\varepsilon})$ for edge updates) for all constant $\varepsilon > 0$.
DAG path counting, Counting s-t-paths, All-pairs-shortest-distances, cycle detection, $k$-cycle detection, $k$-path.
\end{theorem}

\begin{proof}
In \cite{HenzingerKNS15} the conditional lower bounds are proven for transitive closure and $s$-$t$-reachability. This lower bound directly transfers to counting s-t-paths and all-pairs-shortest-distances.

The reduction in \cite{HenzingerKNS15} works by computing some boolean product $u^\top Mv$ for vectors $u,v$ and a matrix $M$, which can be represented as a 4 layered graph. The first and last layer consist of one node, while the two layers in the center have $n$ nodes each. The vectors $u$ and $v$ represent the edges of the first to the second and the third to the fourth layer, while matrix $M$ represent the edges from the second to the third layer. The product $u^\top Mv$ is 1, if there is a path between the nodes in the first and last layer. Note that this graph is a DAG, so the lower bound also holds for DAG path counting and $k$-path ($k \ge 3$). For cycle detection we simply add an edge from the fourth to the first layer, which also gives hardness to $k$-cycle ($k \ge 4$).

\end{proof}

\subsection{Amortized Lower Bounds}
\label{sec:amortizedLowerBounds}

We have noted that the problems and conjectures from \Cref{sec:lowerBounds} could be extended to amortized lower bounds by repeating some of the phases. In this section we will state these extended variants and the amortized lower bounds they would imply. As noted earlier, we feel that the online versions are too complicated to be the right conjectures, and that it is a very interesting open problem to either come up with clean conjectures that capture amortized update time, or break our upper bounds using amortization.  This appendix section should be viewed as a remark - we do not conjecture anything.
We only state the amortized bounds because we already have all the required reductions in this paper and it might be useful for further exploration.

We note that lower bounds in this section hold only for amortization {\em without} the {\em fixed-start assumption}: Almost all dynamic graph algorithms with amortized update time implicitly make an additional assumption that the graph to be preprocessed is an empty graph. Lower bounds in this section, as well as some existing ones, do not hold when this assumption is present. In other words, there is a possibility that an algorithm can break some lower bounds in this section when this assumption is made.

\begin{definition}

Let the computations be performed over the boolean semi-ring and let $t = n^\tau$, $0\le \tau \le 1$. The amortized \hintedOMv problem consists of the following phases:
\begin{enumerate}
\item Input an $n \times t$ matrix $M$ \label{phase:preOMVamortized}
\item Input a $t \times n$ matrix $V$ \label{phase:batchUpdateOMVamortized}
\item Input a sequence of indices $i_1,...,i_t \in [n]$, where $i_{j+1}$ is only given after answering $MV_{i_j}$. \label{phase:queryOMVamortized}
\end{enumerate}
\end{definition}

This problem is essentially the same as \Cref{def:hintedOMv}, but we repeat the last phase in order to obtain an amortized query time. One could conjecture that either phase \ref{phase:batchUpdateOMVamortized} requires $\Omega(n^{\omega(1,\tau,1)-\varepsilon})$ for every $\varepsilon > 0$, or phase \ref{phase:queryOMVamortized} requires $\Omega(n^{1+2\tau-\varepsilon})$ (so $\Omega(n^{1+\tau - \varepsilon})$ for one $i_j$ on average) for every $\varepsilon > 0$.

All reductions from \Cref{sub:columnUpdate} still hold, and they perform both $t = n^\tau$ updates and $n^\tau$ queries in total. Thus for every $0 \le \tau \le 1$ no algorithm can have both amortized update time $O(n^{\omega(1,1,\tau)-\tau - \varepsilon})$ and amortized query time $O(n^{1+\tau-\varepsilon})$ for some constant $\varepsilon > 0$.

\begin{definition}

Let all operations be performed over the boolean semi-ring and let $t = n^{\tau}$ for $0\le \tau \le n$. The \emph{\doubleHintedOMv{}} problem consists of the following phases:
\begin{enumerate}
\item Input matrices $N \in R^{n \times n}, V \in R^{t \times n}$ 
\item Input $I \in [n]^t$. 
\item Input a sequence of indices $j_1,...,j_t \in [n]$, where $j_{k+1}$ is only given after answering $N_{[n],I} V_{[t],j_k}$. \label{phase:queryVectorAmortized}
\end{enumerate}
\end{definition}

Here the conjecture would be the same as in \Cref{con:doubleHintedOMv}, except that the conjectured total time for phase \ref{phase:queryVectorAmortized} is $\Omega(n^{1+2\tau-\varepsilon})$ for every $\varepsilon > 0$.

The reductions from \Cref{sub:elementUpdateRowQuery} require $t = n^\tau$ updates and $n^\tau$ queries to solve this problem, so the lower bounds become amortized: For every $0 \le \tau \le 1$ no algorithm can have both amortized update time $O(n^{\omega(1,1,\tau)-\tau - \varepsilon})$ and amortized query time $O(n^{1+\tau-\varepsilon})$ for some constant $\varepsilon > 0$.

\begin{definition}

Let all operations be performed over the boolean semi-ring and let $t_1 = n^{\tau_1}, t_2 = n^{\tau_2}$, $0\le \tau_1 \le \tau_2 \le n$. The \doubleHintedOuMv{} problem consists of the following phases:
\begin{enumerate}
\item Input matrices $U \in R^{n \times t_1}, N \in R^{n \times n}, V \in R^{t_2 \times n}$ 
\item Input $I \in [n]^{t_1}$. 
\item Repeat for $k=1..t_1/t_2$:
	\begin{enumerate}
		\item Input $J^{(k)} \in [n]^{t_2}$. \label{phase:setJAmortized}
		\item Input a sequence of indices $(i_1,j_1)^{(k)},...,(i_{t_2},j_{t_2})^{(k)} \in [n]$, where $(i_l,j_l)^{(k)}$ is only given after answering $(U N_{I,J^{(k)}} V)_{i^{(k)}_l,j^{(k)}_l}$. \label{phase:queryAmortized}
	\end{enumerate}
\end{enumerate}
\end{definition}

Here the conjecture would be the same as in \Cref{con:2hintedOuMv}, except that we multiply the conjectured time of each phase by the number of repetitions. So the total time for phase \ref{phase:setJAmortized} becomes $\Omega(n^{\omega(1,\tau_1,\tau_2)+tau_1-\tau_2})$ and the total time for all iterations of phase \ref{phase:queryAmortized} becomes $\Omega(n^{2\tau_1+\tau_2)})$.

The reductions from \Cref{sub:elementUpdate} now require $O(t_1)$ updates and $O(t_1)$ queries in total, so the amortized lower bounds are: There exists no algorithm that uses both $O(\min\{n^{\omega(1,1,\tau_1)-\tau_1-\varepsilon}, n^{\omega(1,\tau_1,\tau_2)-\tau_2-\varepsilon}\})$ amortized update time and $O(n^{\tau_1+\tau_2-\varepsilon})$ amortized query time for some $\varepsilon > 0$.

\subsection{Column Update Induces $\Omega(nt)$ Changes in the Inverse}
\label{app:manyChangesExample}

In the overview \Cref{sec:overview} we stated that maintaining the inverse some transformation matrix $T^{(0,t)}$ explicitly requires $\Omega(nt)$ time, because that many entries of the inverse can change.

Denote with $\mathbf{1}_{a \times b}$ an all ones matrix of size $a \times b$ and $\I_{a \times a}$ an $a \times a$ identity matrix, then the following matrix differs from identity in $t-1$ columns:

$
T = \begin{pmatrix}
\I_{t-1 \times t-1} & 0 & 0 \\
\mathbf{1}_{1 \times t-1} & 1 & 0 \\
\mathbf{1}_{n-t \times t-1} & 0 & \I_{n-t \times n-t}
\end{pmatrix}$  $
T^{-1} = \begin{pmatrix}
\I_{t-1 \times t-1} & 0 & 0 \\
-\mathbf{1}_{1 \times t-1} & 1 & 0 \\
-\mathbf{1}_{n-t \times t} & 0 & \I_{n-t \times n-t}
\end{pmatrix}
$

After changing one further column of $T$ we have

$
T' = \begin{pmatrix}
\I_{t-1 \times t-1} & 0 & 0 \\
\mathbf{1}_{1 \times t-1} & 1 & 0 \\
\mathbf{1}_{n-t \times t-1} & \mathbf{1}_{n-t \times 1} & \I_{n-t \times n-t}
\end{pmatrix}
$
$
T'^{-1} = \begin{pmatrix}
\I_{t-1 \times t-1} & 0 & 0 \\
-\mathbf{1}_{1 \times t-1} & 1 & 0 \\
0 & -\mathbf{1}_{n-t \times 1} & \I_{n-t \times n-t}
\end{pmatrix}
$

So $\Omega(nt)$ entries changed in its inverse, which means maintaining a matrix of this structure explicitly requires $\Omega(nt)$ worst-case update time.